\documentclass[a4paper]{llncsreport}
\pdfoutput=1
\usepackage[utf8]{inputenc}
\usepackage[T1]{fontenc}
\usepackage[ngerman,english]{babel}
\usepackage{llncsstuffthm}
\usepackage{llncsstuffset}
\usepackage{cwdefs}
\usepackage{amssymb}
\usepackage{amsfonts}
\usepackage{multirow}
\usepackage{longtable}
\usepackage{booktabs}
\usepackage{enumitem}
\usepackage{amsmath}
\usepackage{bigdelim}
\usepackage{comment}
\usepackage{xspace}
\usepackage{booktabs}
\usepackage{graphicx}
\usepackage{setspace}
\usepackage[style=alphabetic,maxnames=5,backend=biber]{biblatex}
\usepackage{bookmark}
\usepackage[titles]{tocloft}
\usepackage{hyperref}
\hypersetup{
hyperfootnotes=false,
breaklinks=true,
hidelinks,
linktoc=all,
colorlinks=false
}
\usepackage{pdfpages}

\makeatletter
\def\overparen#1{\mathop{\vbox{\ialign{##\crcr\noalign{\kern3\p@}
      \downparenthfill\vspace{-3.5pt}\crcr\noalign{\kern3\p@\nointerlineskip}
      $\hfil\displaystyle{#1}\hfil$\crcr}}}\limits}
\def\underparen#1{\mathop{\vtop{\ialign{##\crcr
      $\hfil\displaystyle{#1}\hfil$\crcr\noalign{\kern3\p@\nointerlineskip}
      \vspace{-3.5pt}\upparenthfill\crcr\noalign{\kern3\p@}}}}\limits}
\def\downparenthfill{$\m@th\braceld\leaders\vrule\hfill\bracerd$}
\def\upparenthfill{$\m@th\bracelu\leaders\vrule\hfill\braceru$}
\makeatother

\title{Heinrich Behmann's Contributions to\\ Second-Order Quantifier 
Elimination\\from the View of Computational Logic}

\author{Christoph Wernhard\\[-8pt]}
\institute{Technische Universit\"{a}t Dresden}

\AtEveryBibitem{
  \clearfield{doi}
}

\NewBibliographyString{bestandsbildner}
\DefineBibliographyStrings{english}{%
                 bestandsbildner  = {Bestandsbildner},
               }

\defbibheading{mybibheading}[\refname]{%
\section*{#1}%
\markboth{\small \hfill}{\small References \hfill}}

\bibliography{bibelim01}

\newcounter{equivcounter}
\newcounter{entailcounter}

\newcommand{\eqlab}[1]
{\refstepcounter{entailcounter}%
\refstepcounter{equivcounter}%
\label{#1}%
\textnormal{EQ~\arabic{equivcounter}}}

\newcommand{\enlab}[1]
{\refstepcounter{equivcounter}%
\refstepcounter{entailcounter}%
\label{#1}%
\textnormal{(EN~\arabic{entailcounter})}}

\newcommand{\de}[1]{\emph{\foreignlanguage{ngerman}{#1}}}
\newcommand{\dename}[1]{\emph{\foreignlanguage{ngerman}{#1}}}
\newcommand{\deq}[1]{\emph{``\foreignlanguage{ngerman}{#1}''}}

\newcommand{\enq}[1]{``#1''}

\newcommand{\ex}{\exists}
\newcommand{\all}{\forall}

\newcommand{\binop}{\otimes}

\newcommand{\dual}{\f{dual}}

\newcommand{\cca}[1]{\overline{< #1}}

\newcommand{\behfalse}{{\curlywedge}}
\newcommand{\behtrue}{{\curlyvee}}

\newcommand{\lskip}{\vspace{0.8ex} \noindent}
\newcommand{\void}[1]{}

\newcommand{\QMONE}{$\f{QMON}_=$\xspace}
\newcommand{\MONE}{$\f{MON}_=$\xspace}
\newcommand{\MON}{$\f{MON}$\xspace}
\newcommand{\QMON}{$\f{QMON}$\xspace}

\newcommand{\mref}[1]{[Ma\-nu\-script \ref{#1}]}
\newcommand{\cref}[1]{[Letter \ref{#1}]}

\newcommand{\geqzero}{$\geq 0$\xspace}
\newcommand{\geqone}{$\geq 1$\xspace}

\DeclareRobustCommand*{\entailedby}{\Relbar\joinrel\mathrel{|}}

\newcommand{\hparen}{\hspace{0.9ex}}
\newcommand{\indebe}{\hphantom{\ex \varphi \all x \all y\,}}
\newcommand{\indebex}{\hphantom{\ex \varphi \all x \all y\, (}}

\begin{document}

\includepdf[pages={1-2}]{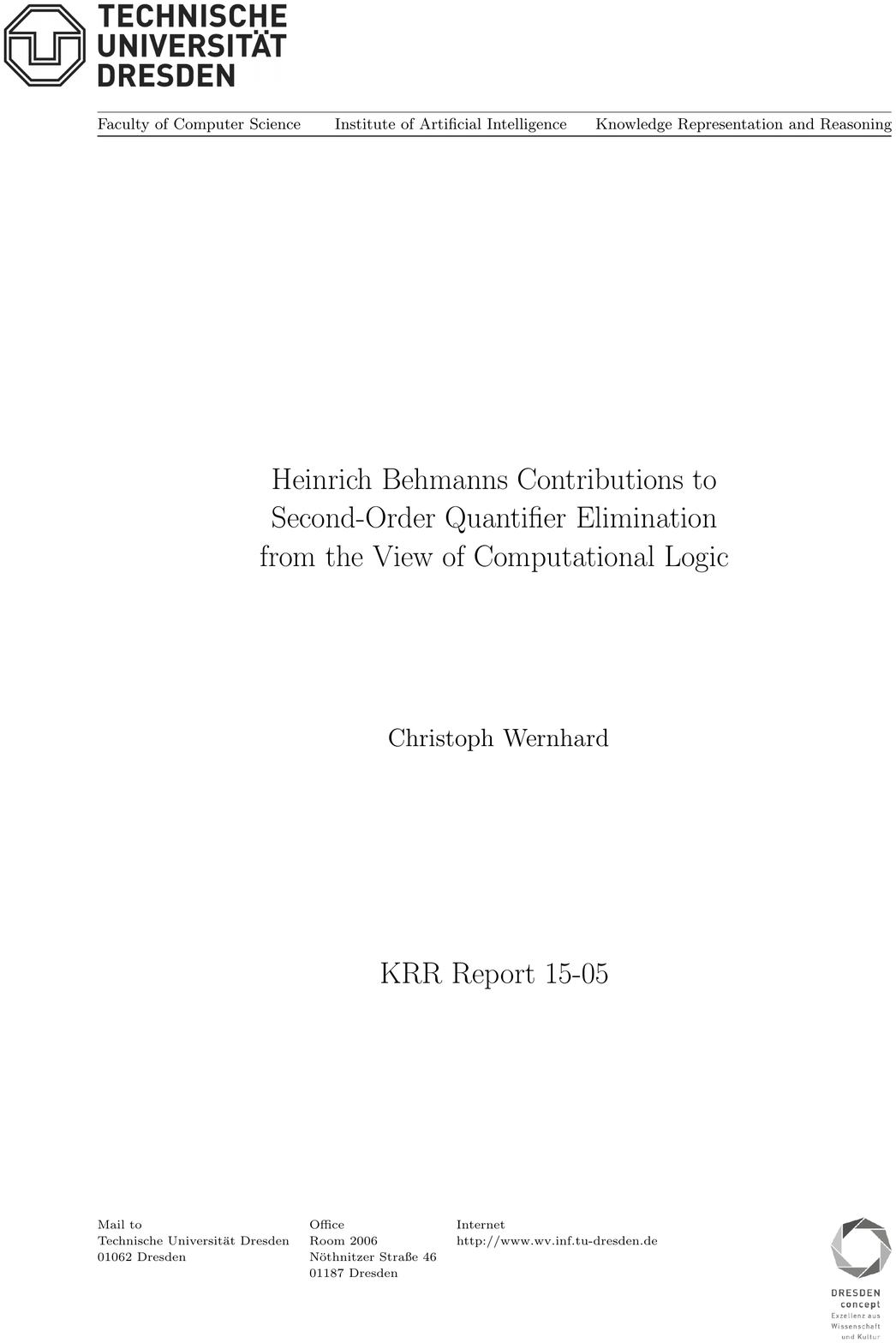}
\setcounter{page}{1}

\maketitle

\begin{abstract}
For relational monadic formulas (the Löwenheim class) sec\-ond-order
quantifier elimination, which is closely related to computation of uniform
interpolants, projection and forgetting -- operations that currently receive
much attention in knowledge processing -- always succeeds. The decidability
proof for this class by Heinrich Behmann from 1922 explicitly proceeds by
elimination with equivalence preserving formula rewriting.  Here we
reconstruct the results from Behmann's publication in detail and discuss
related issues that are relevant in the context of modern approaches to
second-order quantifier elimination in computational logic.
In addition, an extensive documentation of the letters and manuscripts in
Behmann's bequest that concern second-order quantifier elimination is given,
including a commented register and English abstracts of the German sources
with focus on technical material.  In the late 1920s Behmann attempted to
develop an elimination-based decision method for formulas with predicates
whose arity is larger than one.  His manuscripts and the correspondence with
Wilhelm Ackermann show technical aspects that are still of interest today and
give insight into the genesis of Ackermann's landmark paper
\name{Untersuchungen über das Eliminationsproblem der mathematischen Logik}
from 1935, which laid the foundation of the two prevailing modern approaches
to second-order quantifier elimination.
\end{abstract}

\setcounter{tocdepth}{2}

\noindent
Revision: December 19, 2017

\newpage

\newcommand{\marktocheader}
{\markboth{\small \hfill Table of Contents}%
{\small Table of Contents \hfill}}

\cftsetindents{section}{0em}{1.5em}
\cftsetindents{subsection}{1.5em}{2.3em}
\renewcommand{\cftpartfont}{\bf\normalsize}
\renewcommand{\cftpartpagefont}{\bf\normalsize}
\renewcommand{\cftpartpresnum}{Part \marktocheader}
\renewcommand{\cftpartafterpnum}{\vspace{1.5ex}}
\makeatletter
\renewcommand{\@tocrmarg}{2.55em plus1fil}
\makeatother

\tableofcontents

\newpage
\setcounter{chapter}{0}

\part{Introduction -- Contributions of Behmann's Habilitation Thesis
 from the View of Computational Logic}
\label{part-introduction}

\section{Introduction to Part~\ref{part-introduction}}
The Habilitation thesis of Heinrich Behmann (1891--1970), published in 1922 in
\de{Mathematische Annalen} \cite{beh:22}, belongs, along with works by
Löwenheim \cite{loewenheim:15} and Skolem \cite{skolem:19,skolem:20}, to the
standard references on the decision problem for relational monadic first-order
formulas (the Löwenheim class),\footnotemark\ and also for the extension of
this class by second-order quantification upon predicates.  Early such
references of \cite{beh:22} include \cite[p.~77]{hilbert:ackermann:28},
\cite[p.~95]{hilbert:ackermann:second:38}, and
\cite[p.~200]{hilbert:bernays:34}, where also the methods by Behmann are
reproduced \cite[p.~193ff and p.~200ff]{hilbert:bernays:34}.  A detailed
historic account is provided in Church's \name{Introduction to Mathematical
  Logic} \cite[\S49, in particular p.~293]{church:book}. The book by Church
also presents variants of methods from \cite{beh:22}.

\footnotetext{Relational monadic formulas are first-order formulas with only
  unary predicates and no functions other than constants.}

Behmann's early work up to 1921 is presented in historic context by Mancosu
\cite{mancosu:behmann:99}. The focus there is the dissertation from 1918, but
various issues concerning \cite{beh:22}, in particular its embedding into the
context of the Hilbert school, are also documented. The historic analysis of
the development of logics in the period 1917--23 by Zach
\cite{zach:99:completeness} describes Behmann's contributions.  In particular,
it is observed there that Behmann's talk on 10 May 1921 
at the \de{Mathematische Gesellschaft} in Göttingen on the topic of his
Habilitation thesis seems the first documented use of the term
\de{Entscheidungsproblem} (\name{decision problem})
\cite[p.~363]{zach:99:completeness}. A transcript and English translation of
this talk, along with a comprehensive introduction, has recently been
published by Mancosu and Zach \cite{mancosu:zach:2015}.  The first published
explicit statement of the decision problem seems to be in
\cite[p.~166]{beh:22} (see \cite{zach:99:completeness} for an English
translation of the relevant passages).

Behmann reduces the decision problem for relational monadic formulas to the
second-order quantifier elimination problem, that is, the problem to compute
for a given second-order formula an equivalent first-order formula.  In
computational logic, second-order quantifier elimination \cite{soqe}, with
variants called \name{uniform interpolation}, \name{forgetting} and
\name{projection}, is today an area with a wide variety of applications and
techniques.\footnote{As for example
reflected in the \name{SOQE~2017} workshop \cite{soqe2017}.}
Some of today's advanced methods for second-order quantifier
elimination are explicitly based on the so-called \name{Ackermann's Lemma},
due to Wilhelm Ackermann \cite{ackermann:35}, and involve equivalence
preserving rewriting of formulas as key technique
\cite{dls}.
Although Behmann actually uses such rewriting techniques and gave with
\cite{beh:22} at that time Ackermann the impetus to investigate the
elimination problem -- as Ackermann courteously remarked in a letter to
Behmann dated 29~Oct 1934 \cref{corr:ab:1934:10:29}\footnote{Letters and
  manuscripts in Behmann's bequest are listed in Part~\ref{part-app-sources}.}
(see also Sect.~\ref{sec-corr-elim-1934}) -- it appears that \cite{beh:22} so
far has been largely overlooked in the context of second-order quantifier
elimination in computational logic, such as the monograph \cite{soqe}, with
the exception of historic references in
\cite{craig:2008,schmidt:2012:ackermann} and a recent paper by the present
author \cite{cw-relmon}.

In this report we provide a detailed technical reconstruction of the methods
and results from \cite{beh:22} (Part~\ref{part-main-method}) and discuss
various related issues, of which many are still today of relevance in
computational logic (Part~\ref{part-further}).  We summarize follow-up works
by Behmann himself in unpublished manuscripts and in the correspondence with
Wilhelm Ackermann, which mainly concerns elimination in presence of predicates
with arity larger than one (Part~\ref{part-polyadic}). This is supplemented by
commented listings of publications by Behmann and documents in his bequest
that are related to second-order quantifier elimination
(Part~\ref{part-app-sources}). The correspondence with Wilhelm Ackermann, as
far as archived in Behmann's bequest in the \dename{Staatsbibliothek zu
  Berlin}, is registered there completely. Part~\ref{part-conclusion}
concludes the report.

\label{page-paradoxes}
We do not address another major concern of Behmann that is related to
computational logic: his approach to resolve paradoxes, based on the idea that
these emerge from unjustified elimination of shorthands
(\dename{Kurzzeichen}), leading to a variant of lambda conversion and
restricted quantifiers \cite{beh:31:widersprueche,beh:59:limitierte}. He
discussed his approach, which is briefly mentioned by Curry and Feys
in \cite[p.~4, 9, 260f]{curry:combinatoric:1}, in correspondence with, among
others, Ackermann, Bernays, Church, Gödel and Ramsey.

As already indicated, Behmann's Habilitation thesis \cite{beh:22} has so far
mainly been considered in the context of the history of the decision problem.
However, from the point of view of computational logic it is relevant also in
various further respects, not merely for historical reasons, but there are
also technical aspects that are still of significance today, for example, the
successful termination of second-order quantifier elimination methods on
relational monadic formulas \cite{cw-relmon}, as well as methodical aspects,
such as the roles of normal forms.  The remaining sections of this part
discuss these contributions.

\newcounter{contribcounter}

\newcommand{\contrib}[1]{\section{#1}}

\contrib{Specification of the Decision Problem} As already mentioned, the
first explicit statement of the decision problem seems to be in \cite{beh:22}.
For a translation of the relevant passages and discussions see
\cite{zach:99:completeness,mancosu:zach:2015}.

\sectionmark{Contributions of Behmann's
  Habilitation Thesis}

\contrib{Solution of the Decision Problem for Relational Monadic First- and
  Second-Order Formulas with Equality} 
\label{sec-contrib-decision}
As indicated above, this result was
first obtained by Löwenheim, whereas Skolem and Behmann provided further
proofs.  Like Behmann's method, the techniques of Löwenheim and Skolem also
apply if predicate quantification is considered \cite[p.~293]{church:book}.
As further noted in \cite[p.~293]{church:book}, Behmann's method to handle
equality is similar to that of Skolem in some important respects, but seems to
have been found independently.  Behmann himself describes this in a letter
dated 27 December 1927 to Heinrich Scholz \cite[Kasten 3, I~63]{beh:nl},
brought to attention in \cite{mancosu:behmann:99} with excerpts published in
\cite{mancosu:zach:2015} and below -- see p.~\pageref{page-scholz-letter} and
\pageref{page-scholz-lengthy}.  Skolem's proof is outlined from the
perspective of elimination in \cite{craig:2008}.  A methodical aspect of
\cite{beh:22} seems worth mentioning: The decision problem is attacked there
by investigating decidability explicitly for specific \emph{syntactically
  characterized formula classes} (\de{Aussagenbereiche}).

\contrib{Specification of the Problem of Second-Order Quantifier Elimination}
\label{sec-contrib-problem-of-soqe}
As described in \cite{craig:2008}, elimination problems play an important role
in the works of Boole \cite{boole:laws} and Schröder \cite{schroeder}.  It
seems, however, that the problem of second-order quantifier elimination has not
been fully understood and explicitly stated accordingly before \cite{beh:22}.
The second-order quantifier elimination problem is called there a \name{new
  ``elimination problem''} (\de{neue[s] \glqq Eliminationsproblem\grqq}) and
is explicated in the context of the instance that occurs first in that paper,
the elimination of a unary predicate with respect to a formula of relational
monadic first-order logic without equality.\footnotemark\ This specification
can be paraphrased as follows: Given is a formula $\all p\, F$ or $\ex
p\, F$, where $p$ is a unary predicate and $F$ is of monadic first-order logic
without equality. The objective is now to find a, or -- as can be said more
determined -- the (first-order) formula that is equivalent to the given
formula -- with respect to the predicates with exception of $p$, the constants
and the free variables in $F$ -- but does not contain $p$ any more.  Behmann
also gives a second more semantic view on the elimination problem: The
(first-order) relationship among the predicates (with exception of $p$),
constants and free variables in $F$ should be determined that is a necessary
and sufficient condition for $F$ being true for arbitrary predicates~$p$ or
for at least one predicate~$p$, respectively.

Following Schröder \cite{schroeder}, Behmann calls the formula sought after
\name{resultant} (\de{Resultante}). In the context of his elimination method,
Behmann speaks in early manuscripts from 1921 of \name{separation}
(\de{Aussonderung}) instead of \de{Elimination}. For instance, on p.~13 in
\mref{man:beh:21:ms:k9:37}, a method description is headed
\de{Eliminationsverfahren. (Aussonderung?)}.  In \mref{man:beh:21:carbon}, the
manuscript for \cite{beh:22}, on p.~40, the specification of the elimination
problem quoted in footnote~\ref{foot-elimprob}
 uses \de{\glqq Aussonderungsproblem\grqq} in place of
\de{\glqq Eliminationsproblem\grqq}, on p.~45 the originally typed term
\name{Aussonderungshauptform} is altered by a handwritten annotation to
\name{Eliminationshauptform} (German for \name{main form for elimination}).

\footnotetext{\label{foot-elimprob}\selectlanguage{ngerman}\cite[p.~196f]{beh:22}:
  Ich möchte dieses neue \glqq Eliminationsproblem\grqq\ in der folgenden
  Weise bestimmter fassen: Gegeben ist eine Aussage
\[\varphi F_{\varphi f g a b} \text{ oder } 
  \overline{\varphi} F_{\varphi f g a b},\] wo der Operand eine Aussage
  unseres früheren Bereiches $A$ ist und, abgesehen von $\varphi^x$, nur
  konstante Eigenschaften, natürlich in beliebiger endlicher Anzahl, enthält
  -- da sie nämlich innerhalb des obigen Ausdrucks nicht durch Operatoren
  vertreten sind, haben wir sie eben, solange wir unser Augenmerk nur auf
  diesen richten, als konstant anzusehen --, außerdem möglicherweise noch
  konstante und veränderliche Individuen; diese letzten natürlich zugleich als
  Operatoren, daher oben nicht angedeutet. (Das Auftreten von Aussagen $p, q,
  \ldots$ als Grundbestandteile [im Original: Grundbestandteilen] hat im
  gegenwärtigen Zusammenhang kaum praktische Bedeutung.) \emph{Es soll nun
    eine oder } -- wie wir bestimmter sagen dürfen -- \emph{diejenige Aussage
    (erster Ordnung) gefunden werden, die der gegebenen äquivalent ist} --
  natürlich für irgendwelche Werte von $f, g, a, b$ --\emph{, den Begriff
    $\varphi^x$ jedoch nicht mehr enthält.}  Mit anderen Worten: \emph{Es soll
    diejenige Beziehung (erster Ordnung) zwischen den Konstanten $f, g, a, b$
    ermittelt werden, deren Bestehen die notwendige und hinreichende Bedingung
    dafür ist, daß $F_{\varphi f g a b}$ je nachdem für beliebige
    Begriffe~$\varphi^x$ oder für mindestens einen Begriff~$\varphi^x$ eine
    richtige Aussage darstellt.}\selectlanguage{english}}

Behmann \cite[p.~218ff]{beh:22} remarks that Schröder distinguishes between
``elimination problem'' and ``summation problem'', which are in Behmann's view
actually identical.  Schröder's elimination problem is, in Behmann's words, to
find a condition for the satisfiability of $F$ that is free from
$p$.\footnote{Another interpretation of Schröder's concept of the elimination
  problem is to find a consequence of $F$ that has exactly those formulas as
  consequences which are consequences of $F$ and do not contain $p$.  See, for
  example, \cite[p.~200]{schroeder:2.1}.}  Schröder's ``summation problem''
is, according to Behmann, to find a formula that is equivalent to $\ex p\,
F$ and is free from~$p$.  As Behmann describes, Schröder observed the
equivalence of both problems in a note inserted during printing of the third
volume of his \de{Vorlesungen über die Algebra der Logik}
\cite[p.~489--490]{schroeder:3}, whereas the actual identity of both problems
escaped him through his concern for analogy with numerical algebra. Behmann
concludes with commenting that this is a strange evidence for the extent in
which for Schröder content receded in favor of form.\footnotemark%

\footnotetext{In modern view, a divergence between syntactic and semantic
  conceptions of elimination actually arises: For example, the quantified
  Boolean formula $\ex p\, (q \land p)$ is equivalent to the propositional
  formula $q \land (p \lor \lnot p)$, where the Boolean quantifier has been
  ``eliminated''. However, the formula still contains syntactically the
  formerly quantified atom $p$, although, from a semantic point of view,
  redundantly.  The characterization of such redundancy is not always evident,
  for example, for modern variants of second-order quantifier elimination,
  where it its possible to ``quantify upon'' just a particular ground atom.
  On the other hand, for first-order logic, both views coincide in a sense: As
  noted in \cite[Introduction]{otto:interpolation:2000}, the construction of
  interpolants according to Craig's interpolation theorem can be applied to
  compute for a given first-order formula that is known to be
  \emph{equivalent} to a formula expressed with a certain signature
  (predicates, functions and constants), an equivalent formula that is
  \emph{syntactically} in that signature. The existence of an equivalent
  formula in a given signature can be expressed as validity.}

Although Behmann explicitly formulates the problem of eliminating second-order
quantifiers and reduces the decision problem for relational monadic formulas
to that elimination problem, he remains skeptical on whether the
generalization to predicates with arbitrary arities and higher-order concepts
can still be based on the elimination problem.  His argument in
\cite[p.~226f]{beh:22} is summarized in Part~\ref{part-polyadic},
p.~\pageref{page-cite-limitations-of-elimination}. In his letter to Heinrich
Scholz \label{page-scholz-letter} dated 27 December 1927 \cite[Kasten 3,
  I~63]{beh:nl}, answering Scholz's question about who has written before him
and, in particular, who has written at first, about the decision problem
(Postcard from Scholz to Behmann, dated 19 December 1927, \cite[Kasten 3,
  I~63]{beh:nl}), he relates the decision problem to the elimination problem
and remarks that the latter has been treated \enq{first by the Americans, in
  particular Peirce, and later with particular love and persistence by
  Schröder, and eventually found a specialist in Löwenheim, who wrote several
  treatises about it in the \dename{Math.\ Annalen}}. Behmann continues that
the most important of these works by Löwenheim is \name{Über Möglichkeiten im
  Relativkalkül} \cite{loewenheim:15} and mentions that Löwenheim came there
already to important partial results of his work \cite{beh:22}, \enq{however
  -- in a presentation that is neither mathematically strict nor sufficiently
  comprehensible, such that I have construed these properly only after
  publication of my own paper.}\footnotemark

\footnotetext{An excerpt of the German letter is quoted in
  Sect.~\ref{sec-corr-further}, p.~\pageref{page-scholz-lengthy}.  Further parts
  of that letter are summarized in \cite{mancosu:zach:2015}.}

\sectionmark{Solution of the Second-Order Quantifier Elimination Problem}
\contrib{Solution of the Second-Order Quantifier Elimination Problem for
  Relational Monadic Formulas with Equality}

Behmann \cite{beh:22} presents an effective method for eliminating the
second-order quantifiers in a given relational monadic formula with equality
and with predicate quantification.  (The earlier decidability proofs in
\cite{loewenheim:15} and \cite{skolem:19} mentioned above are similarly based
on second-order quantifier elimination in monadic logics.)  Behmann's method
terminates after a finite number of steps with an equivalent relational
monadic \emph{first-order} formula.  The result formula might be with equality
also in cases where the given formula is without equality.  For a given
formula in which all predicates are quantified, the result formula just
expresses constraints on the domain cardinality: The formula is either true
for all domain cardinalities with exception of a finite number or false for
all domain cardinalities with exception of a finite number.  Obviously, valid
and unsatisfiable formulas are special cases of such formulas.

The key technique of Behmann's procedure is to propagate quantifiers inward,
also for the price of expensive operations such as distribution of conjunction
over disjunction.  As suggested by Behmann, we call the resulting form
\name{innex}.\footnote{\de{Die \label{foot-innex}
Endform meines Reduktionsverfahrens wird
    gelegentlich (sprachlich wenig glücklich) als "{}kontrapränex"{}
    bezeichnet; ich ziehe die Benennung "{}innex"{} vor.} Letter from Heinrich
  Behmann to Alonzo Church, 30 January 1959 \cite[Kasten~1, I~11]{beh:nl}.}
This inward propagation is is applied to quantifiers upon individual variables
as well as to quantifiers upon predicates.  A detailed presentation of
Behmann's method is the topic of Part~\ref{part-main-method}, further aspects
of the method will be discussed in Part~\ref{part-further}.

\contrib{Clarification of Schröder's Early Results on Elimination} Most issues
solved in \cite{beh:22} have been raised by Schröder in his \name{Vorlesungen
  über die Algebra der Logik} \cite{schroeder}. Their solutions are developed
by Behmann in a more modern representational framework, with a dedicated
notation for logics, not obfuscated by the aim for correspondence to numeric
algebra.  The work by Behmann seems thus also useful as a guide to Schröder's
results, complementing the outline in \cite{craig:2008}.  We already sketched
Behmann's discussion of Schröder's notion of elimination in
Sect.~\ref{sec-contrib-problem-of-soqe}.  The precise relationship of
Behmann's core result to Schröder's earlier partial result, in particular, his
\name{,,crude resultant''} (\de{\glqq Resultante aus dem Rohen\grqq}), as
described Behmann will be shown in Sect.~\ref{sec-elim-noeq}.

\contrib{Methodology: Computation by Equivalence Preserving Rewriting}

With the requirement to decide a statement after a finite number of steps, the
\de{Entscheidungsproblem} inherently involves some notion of
computation. Computation steps are expressed in \cite{beh:22} as equivalence
preserving rewriting steps of logical formulas, justified by a collection of
formula equivalences. A method starts with a given formula. At each step, a
subformula occurrence is replaced with an equivalent formula, according to
some computation rule (\de{Rechenregel}), that is, an equivalence from the
collection, oriented either from left to right or from right to left.
Computation terminates if the formula has reached a specific syntactic form.
As Behmann states \cite[p.~167]{beh:22}, the particular collection of
equivalences he gives is motivated not by finding a small set of orthogonal
axioms, but by satisfying the needs of practical computation (\de{Bedürfnisse
  des praktischen Rechnens}).\footnote{See
  \cite[p.~351f]{zach:99:completeness} for an English translation of the
  relevant section from \cite{beh:22}.}

The foundation on equivalence, a semantic property, ensures that equivalence
to the originally given formula is maintained as an invariant throughout the
computation.  Today, the representation by rewriting rules or transition
systems that preserve semantic properties is a well established elegant way to
represent computational methods such that they can be analyzed.
Second-order quantification allows to represent also notions like
equi-satisfiability (two formulas are either both satisfiable or
unsatisfiable) as equivalence of formulas, making the preservation of
equivalence a particularly useful invariant.

The modern so-called \name{direct methods} or \name{methods of the Ackermann
  approach}, for second-order quantifier elimination \cite{soqe}, initiated by
\cite{dls:early,dls}, typically founded on Ackermann's Lemma
\cite{ackermann:35}, quite literally follow Behmann's template of applying
equivalence preserving formula rewritings that include various formula
conversions and elimination steps.

\contrib{Methodology: Normal Forms}

The methods introduced in \cite{beh:22} essentially operate by converting
given logical formulas to equivalent formulas in specific normal forms, that
is, formulas with specific syntactic properties.  Conjunctive and disjunctive
normal forms are used dually there, the innex normal forms for quantifiers
upon instance variables as well as upon predicates are developed and counting
(or cardinality) quantifiers are applied (preceded in
\cite{loewenheim:15,skolem:19} -- see \cite{craig:2008}). Behmann's method
rewrites formulas to a certain intermediate form that allows predicate
elimination according to a simple scheme.

The syntax for formulas used in \cite{beh:22} is based on disjunction,
conjunction and negation, exposing the symmetry and duality inherent in these
operations, which is, as criticized by Behmann, obscured in notations based on
implication used by Frege and in the \name{Principia Mathematica}.  For
monadic formulas in a certain normal form, Behmann introduces a special
notation (\de{Klassensymbolik}), where individual variables are suppressed.

In modern computational logic, normal forms play various roles. Systems
typically operate on inputs in conjunctive normal form, obtained from
preprocessors.  Normal forms that allow to perform certain operations in an
inexpensive way are investigated as target formats for knowledge compilation.
The preservation of a certain normal form by calculi is applied to ensure that
outputs are in a certain fragment of first-order logic or can be mapped to
some other logic such as a modal or description logic (the method of
\cite{ks:2013:frocos} for second-order quantifier elimination in description
logics can be considered as an example).  Normal forms can provide
representations of formulas that facilitate to understand their meaning, which
is useful in the development of techniques as well as to present results to
end users.  This aspect of getting an overview on a solution in its totality
has been a continuous concern for Behmann, for example in his comments to
Ackermann's resolution-based elimination technique, summarized below in
Part~\ref{part-polyadic}, or in his later work on the solution problem
(\name{Auflösungsproblem})
\cite{beh:50:aufloesungs:phil:1,beh:51:aufloesungs:phil:2}.

The paradigm of ``model computation'' or ``answer set programming'' in
automated reasoning and logic programming can be considered as a variant of
normal form computation: Such systems enumerate data structures that represent
models. Their solutions can be regarded as a normalized representation of the
input.  A particular special case is enumerating all conjunctive clauses of a
disjunctive normal norm.

The so-called quantifier elimination approach in early model theory of the
1920s is another area where variants of normal form computation play an
essential role, as will be discussed below in Sect.~\ref{sec-nfqe-mfe}.  The
integration of such quantifier elimination methods for decidable theories into
reasoning systems is currently an area of extensive research in automated
deduction, motivated in particular by applications in software and hardware
verification.  Also the evaluation of relational database queries can be
considered as elimination of first-order quantifiers
\cite{kanellakis:95,revesz}.

\part{Behmann's Decision and Elimination Method for
Relational Monadic Formulas}
\label{part-main-method}

\section{Introduction to Part~\ref{part-main-method}}

This part focuses on the main technical material in \cite{beh:22} from the
point of view of second-order quantifier elimination as considered in
computational logic. Modern notation is used throughout, on occasion a
concordance to Behmann's original labeling and various notations, which have
merits on their own, are given. We aim here at a general formalization, where
Behmann sometimes introduces techniques only with exemplary cases that provide
intuition and indicate the general case. Also the structuring of the
presentation deviates from the original paper, aiming at the modern reader.

The rest of this part is structured as follows: In Sect.~\ref{sec-notation}
notation and terminology are introduced and general remarks on the
presentation are given.  Section~\ref{sec-overview} provides an overview on
Behmann's results, proceeding in a ``top-down'' fashion where the more
involved methods and proofs are only sketched.  With a collection of
equivalences and entailments for use in formula rewriting and considerations
on deciding and normalizing propositional logic, Sect.~\ref{sec-starting}
paves the way for the more thorough presentation of Behmann's techniques in
the subsequent sections. First, the general case of monadic formulas
\emph{with equality} is considered. Sect.~\ref{sec-ccnf} describes a
normalization method for such formulas.  The second-order quantifier
elimination method, which applies to the normalized formulas, is then shown in
detail in Sect.~\ref{sec-elim-eq}.  A simplified variant of the general method
is then considered in Sect.~\ref{sec-elim-noeq}.  It applies just to the case
without equality, but facilitates discussion of other issues, in particular
the correspondence to earlier works by Schröder, as shown by Behmann.

\section{Notational Conventions and Preliminary Remarks}
\label{sec-notation}

\subsection{Syntax}

We briefly write \defname{predicate}, \defname{function} and
\defname{constant} for \defname{predicate symbol}, \defname{function symbol}
and \defname{constant symbol}, respectively.  For second-order logic with
equality, that is first-order logic with equality and extended by
quantification upon predicates and functions, we use the following syntactic
notation: An \defname{atomic formula}, or briefly \defname{atom}, is either of
the form $pt_1\ldots t_n$, where $n \geq 0$ and where $p$ is a predicate of
arity $n$, or of the form $t_1 = t_2$.  In both cases, each subscripted $t$ is
a \defname{term}, that is an \defname{individual variable} or of the form
$ft_1\ldots t_n$, where $n \geq 0$, the $t_i$ are terms and $f$ is a
\defname{function} of arity $n$.\footnotemark\ A nullary function is also
called \defname{constant}.  An atom of the form $t_1 = t_2$ is called
\defname{equality atom}. Formulas and classes of formulas which may
contain/may not contain equality atoms are called briefly \defname{with
  equality} or \defname{without equality}, respectively.  If we speak of
\defname{first-order} formulas, unless explicitly indicated otherwise, we
assume formulas with equality.

\footnotetext{This parenthesis-free notation for terms and atoms is used in
  the modern textbook \cite{eft:german}.}

\nocite{eft:english}

A \defname{formula} is constructed from atoms, the constant operators $\true$
(\name{true}), $\false$ (\name{false}), and a finite number of applications of
the unary connectives $\lnot$ (\name{negation}), the binary connectives
$\land$ (\name{conjunction}) and $\lor$ (\name{disjunction}), as well as
quantifications with $\all$ (\name{universal quantification}) and $\ex$
(\name{existential quantification}).  Negated equality $\neq$, further binary
operators~$\imp, \revimp, \equi$, as well as $n$-ary versions of $\land$ and
$\lor$ can be understood as meta-level notation.  For these $n$-ary versions,
the cases $n=0$, which are $\true$ and $\false$, respectively, are included.
The scope of $\lnot$, the quantifiers, and the $n$-ary versions of $\land$ and
$\lor$ in prefix notation is the immediate subformula to the right.

A subformula occurrence in a given formula is \defname{positive}
(\defname{negative}) if it is in the scope of an even (odd) number of
negations.  A \defname{literal} is an atom or a negated atom.  For a
formula~$F$ of the form $\lnot G$, the \defname{complement}~$\du{F}$ of~$F$ is
$G$, if $F$ has a form different from $\lnot G$, then the complement is $\lnot
F$.

A quantifier occurrence may be upon an \name{individual variable}
(``first-order'') or upon a \name{predicate} or \name{function}
(``second-order''). We call the former also an \defname{individual quantifier}
and the latter \defname{predicate quantifier}.
An occurrence of an individual variable, predicate or function that is not
bound by a quantifier occurrence in a formula is called \defname{free} in that
formula.
As common in discussions of first-order logic, we distinguish between
constants and free occurrences of individual variables.  However, we do not
make an analogous distinction for predicates and functions, since it would not
be of relevance in the considered contexts. Thus, an occurrence of a predicate
or function in a formula is just free or bound by a quantifier.  In a
first-order formula, all predicate and function occurrences are free.

\subsection{Boolean Combination of Basic Formulas}  
Following patterns suggested by early model theory (see e.g.
\cite[Sect.~1.5]{chang:keisler}), in this presentation, several normal forms
are characterized as \name{Boolean combination of basic formulas}, that is, as
the formulas that are obtained from certain basic formulas, the constant
operators $\true$, $\false$ and repeated application of the operators $\lnot$,
$\land$ and $\lor$.%
\footnote{This choice of operators has been made for convenience.  Of course,
  technically it would be sufficient to just permit e.g. $\lnot$ and
  $\land$, and express $\true$ and $\false$ as disjunction and conjunction,
  respectively, of an arbitrarily picked basic formula and its negation.}

\subsection{Considered Formula Classes}

We use the following symbols for particularly considered formula classes: \MON
is the class of relational monadic formulas (also called Löwenheim class),
that is, the class of first-order formulas with nullary and unary predicates,
with constants but no other functions, and without equality. \MONE is \MON
with equality. \QMON and \QMONE are \MON and \MONE, resp., extended by
second-order quantification upon predicates.

All of these classes are decidable. \QMONE admits second-order quantifier
elimination, that is, there is an effective method to compute for a given
\QMONE formula~$F$ an equivalent \MONE formula~$F^{\prime}$ in which all
predicates are unquantified predicates in $F$, as well as all constants and
free variables are also in $F$.  In this sense \MONE is closed under
second-order quantifier elimination, which does not hold for \MON, since
elimination applied to a \QMON formula might introduce equality.

A \name{quantified Boolean formula} is a formula where all predicates are
nullary (called then also \defname{Boolean variables} in the literature) and
quantification is allowed just upon predicates.

\subsection{Remarks on the Presentation of Behmann's Results}
\label{sec-vars-constants}

\paragraph{Theorems that Assert the Existence of Effective Methods.}
\label{sec-meta-methods}
We formulate results often as theorem statements that assert the \emph{existence
  of an effective method} to compute for a given formula an equivalent formula
with certain properties. The proof is then typically a description of such a
method. An alternative would be to just state the \emph{existence of an
  equivalent formula with certain properties} as theorem.  This course has not
been followed here, since in contexts where the existence of a method is
relevant, reference to the theorem alone (playing the role of a ``module
interface'') would not be sufficient, but the underlying proof (the ``module
implementation'') would have to be referenced.

\paragraph{Free Individual Variables and Constants.}
Many of the results of \cite{beh:22} apply to formulas -- which may contain
free variables -- such that free variables and constants are handled in
exactly the same way.  While these results are presented here explicitly as
properties of formulas, they are presented in \cite{beh:22} as properties of
\emph{sentences} (\de{Aussagen}), that is, formulas without free individual
variables, considering free individual variables just as constants (see
\cite[footnote~25, p.~196]{beh:22}\footnotemark).

\footnotetext{Behmann's footnote translates as: I thus call the basic
  components of an expression ``variable'' (\de{\glqq veränderlich\grqq}) or
  ``constant'' (\de{\glqq konstant\grqq}), depending on whether it is
  represented within the expression by quantifiers (\de{Operatoren}) or not.
  As far as I can see, this is exactly the sense in which this distinction is
  actually made in mathematics.}

\paragraph{Consideration of Duality.}
\label{sec-duality}
Methods based on equivalence preserving transformations of classical logic
formulas typically come in two dual variants, where the roles of conjunction
and disjunction as well as the roles of existential and universal
quantification are switched (as made more precise with
Prop.~\ref{prop-equi-dual} in Sect.~\ref{sec-prop-equi-dual}).  In
\cite{beh:22}, such methods are in general explicitly developed for one of the
variants and the dual variant is then indicated.  In the presentation here,
the discussion of dual variants is completely neglected, with exception of a
few specific cases.  Actually, from a technical point of view, the dual
variants can be completely disregarded, since for inputs where they would seem
adequate, their behavior would be just simulated by the original variant with
the only difference that instead of atomic formulas (or basic formulas of
other forms) their complements are used.  As an example, consider the
equivalence $\ex x\, (px \lor qx) \equiv \ex x\, px \lor \ex x\,
qx$ and its dual $\all x\, (px \land qx) \equiv \all x\, px \land
\all x\, qx$. The dual can be derived from the first variant with negated
atoms in the following steps, which additionally only involve inward
propagation of negation and expansion/contraction of universal quantifiers:
$\all x\, (px \land qx)$ $\equiv$ $\lnot \ex x\, \lnot (px \land qx)$
$\equiv$ $\lnot \ex x\, (\lnot px \lor \lnot qx)$ $\equiv$ $\lnot (\ex
x\, \lnot px \lor \ex x\, \lnot qx)$ $\equiv$ $\lnot \ex x\, \lnot
px \land \lnot \ex x\, \lnot qx$ $\equiv$ $\all x\, px \land
\all x\, qx$.

\section{Overview on Behmann's Results and Methods}
\label{sec-overview}

The final result of \cite{beh:22} can be stated as a theorem about \MONE
formulas of a certain syntactic form. It allows to derive various results on
second-order quantifier elimination and decidability of monadic formulas,
including the decidability of \MON.  In this overview section we start with
presenting this theorem and sketch in a ``top-down'' manner the techniques
used in \cite{beh:22} to prove it. We then present the derived results,
proven as corollaries of the theorem. More in-depth proofs of the theorem
itself and discussions of the involved techniques will then be provided in a
``bottom-up'' manner in subsequent sections.

\subsection{Elimination Method}
\label{sec-overview-elimination-method}

The core result of \cite{beh:22} can be stated as follows:
\begin{thm}[Predicate Elimination for \MONE]
\label{thm-beh:22-final}
There is an effective me\-thod to compute from a given predicate~$p$ and
\MONE formula~$F$ a formula~$F^{\prime}$
such that
\begin{enumerate}
\item $F^{\prime}$ is a \MONE formula,
\item $F^{\prime} \equiv \ex p\, F$,
\item $p$ does not occur in $F^{\prime}$,
\item All free individual variables, constants and predicates in
$F^{\prime}$ do occur in~$F$.
\end{enumerate}
\end{thm}

\noindent
The proof given in \cite{beh:22} for Theorem~\ref{thm-beh:22-final} resides on
the conversion of arbitrary \MONE formulas to a certain syntactic normal form
that only allows restricted use of quantification. In particular, the scopes
of quantifier occurrences are not permitted to overlap.  To achieve this
property, it is utilized that equality atoms with quantified variables can be
represented implicitly by counting quantifiers $\ex^{\geq n}$, which
express existence of at least~$n$ individuals.  Formulas with counting
quantifiers can be expanded into equivalent formulas of particular shapes with
standard first-order quantifiers and equality atoms, as discussed in more
detail below in Sect.~\ref{sec-cq}.  The standard quantifier $\ex$ can be
equivalently expressed as $\ex^{\geq 1}$.
In the considered normal form, the argument formula of a counting quantifier
$\ex^{\geq n}$ must be a conjunction of literals of applications of a
unary predicate to the quantified individual variable. The representability of
\MONE formulas in this normal form can be stated as follows:
\pagebreak 
\begin{thm}[Counting Quantifier Normal Form for \MONE]
\label{thm-foqe-monadic-eq}
\noindent
There is an effective method to compute from a given
\MONE formula~$F$ a
formula~$F^{\prime}$ such that
\begin{enumerate}
\item $F^{\prime}$ is a Boolean combination of basic
formulas of the form:
\begin{enumerate}[label=(\alph*),ref=\alph*]
\item \label{basic-nullary} $p$, where $p$ is a nullary predicate,
\item \label{basic-unary} $pt$,
   where $p$ is a unary predicate and $t$ is a constant or an 
  individual variable,
\item \label{basic-equality} $t = s$, where each of $t, s$ is a constant or an 
individual variable,
\item \label{basic-exists} $\ex^{\geq n} x\, \bigwedge_{1 \leq i \leq
  m} L_i[x]$, where $n \geq 1$, $m \geq 0$ and the $L_i[x]$ are pairwise
  different and pairwise non-complementary positive or negative literals with
  a unary predicate applied to the individual variable~$x$,
\end{enumerate}
\item $F^{\prime} \equiv F$,
\item All free individual variables, constants and predicates in $F^{\prime}$ do
  occur in~$F$.
\end{enumerate}
\end{thm}

\noindent
If the given formula~$F$ in Theorem~\ref{thm-foqe-monadic-eq} is without
equality, the allowed basic formulas can be strengthened by excluding the case
$t = s$ (\ref{basic-equality}) and restricting the case (\ref{basic-exists})
to $n = 1$, such that the counting quantifier can be considered as standard
quantifier. The method of \cite{beh:22} to compute the normal form according
to Theorem~\ref{thm-foqe-monadic-eq} proceeds by applying equivalence
preserving formula rewritings to move quantifiers inward such that their
scopes do not overlap. All predicate occurrences in the scope of a quantifier
then have exactly the quantified variable as argument. 

To achieve this, aside from inexpensive transformations such as narrowing
quantifier scopes to subformulas where the quantified variables actually
occur, distribution of existential (universal, resp.) quantifiers over
disjunction (conjunction, resp.), propagating negation inward, and rearranging
binary connectives according to associativity and commutativity, also
\emph{expensive transformations}, in particular distribution of conjunction
over disjunction and vice versa, as familiar from conversion to disjunctive
and conjunctive normal form, are required. Consider the following example,
where initially $\ex y$ is in the scope of $\ex x$:
\begin{equation}
\begin{array}{rl}
& \ex x\, (px \land (qx \lor \ex y\, ry))\\
\equiv & \ex x\, ((px \land qx) \lor (px \land \ex y\, ry))\\
\equiv & \ex x\, (px \land qx) \lor \ex x\, (px \land \ex y\,
  ry)\\
\equiv & \ex x\, (px \land qx) \lor (\ex x\, px \land 
\ex y\, ry).\\
\end{array}
\end{equation}
Equality literals require special handling, as described in the full
exposition in Sect.~\ref{sec-thm-foqe-monadic-eq-fullproof} below.

We now sketch a method as asserted by Theorem~\ref{thm-beh:22-final} for the
case where the input formula is without equality.  Given is the formula
$\ex p\, F$, where $F$ is a \MON formula and $p$ is a predicate. We only
consider unary $p$ here, nullary~$p$ could be handled analogously as a
particularly simple special case.  The input formula is first rewritten to an
equivalent formula in which all occurrences of $p$ are in
subformulas of a specific syntactic form, called \de{Eliminationshauptform}
(main form for elimination) in \cite{beh:22}:
\begin{equation}
\label{eq-hauptform-informal}
\begin{array}{rll}
\ex p\,
&  ( \bigwedge_{1 \leq i \leq a} \all x\, (A_i[x] \lor px) & \land\\
& \hparen \bigwedge_{1 \leq i \leq b} \all x\, (B_i[x] \lor \lnot px) & \land\\
& \hparen \bigwedge_{1 \leq i \leq c} \ex x\, (C_i[x] \land px) & \land\\
&  \hparen \bigwedge_{1 \leq i \leq d} \ex x\, (D_i[x] \land \lnot px)), &
\end{array}
\end{equation}
where $a$, $b$, $c$, $d$ are natural numbers \geqzero and
the $A_i[x]$, $B_i[x]$, $C_i[x]$, $D_i[x]$ are first-order formulas in which
$p$ does not occur.  This can be achieved with the following steps: Normalize
with the method of Theorem~\ref{thm-foqe-monadic-eq} and convert to
disjunctive normal form, generalized such that the role of atoms is played by
the basic formulas.  Rewrite occurrences of $pt$, where $t$ is a constant or
a variable that is free in $F$, first with the equivalence $\lnot pt \equiv
\all x\, (x\neq t \lor \lnot px)$ and then with $pt \equiv \all x\, (x\neq t
\lor px)$.  This step introduces equality, which thus may also be present in
the result.  If $F$ is without equality, the counting quantifiers are
decorated with $1$, directly corresponding to standard quantifiers.  They can
thus be rewritten with the equivalence $\lnot \ex^{\geq 1} x\,
\bigwedge_{1 \leq i \leq m} L_i[x] \equiv \all x\, \bigvee_{1 \leq i \leq
  m} \du{L_i[x]}$, and then with $\ex^{\geq 1} x\, \bigwedge_{1 \leq i
  \leq m} L_i[x] \equiv \ex x\, \bigwedge_{1 \leq i \leq m} L_i[x]$.
After the predicate quantifier upon~$p$ is then moved inward with the
techniques outlined for individual variable quantifiers in the context of
Theorem~\ref{thm-foqe-monadic-eq}, all occurrences of $p$ are in subformulas
that match the \name{Eliminationshauptform}.

The \de{Eliminationshauptform} allows to move the existential individual
quantifiers and the constituents $C_i$ and $D_i$ to the front of the predicate
quantifier, while the occurrences of $p$ with existentially quantified
arguments can be rewritten to universally quantified occurrences that match
the forms $\all x\, (A_i[x] \lor px)$ or $\all x\, (B_i[x] \lor \lnot
px)$, respectively, by applying the equivalences $pt \equiv \all x\, (x\neq
t \lor px)$ and $\lnot pt \equiv \all x\, (x\neq t \lor \lnot px)$.
(Notice that in this step again equality may be introduced also in cases where
the original formula $F$ is without equality.)  In this way, a formula in
\name{Eliminationshauptform} can be further converted such that $p$ only
occurs in a subformula which is in a restricted form of the
\name{Eliminationshauptform} allowing only the two universally quantified
constituents, that is, $c = d = 0$. When restricted in this way, the
\name{Eliminationshauptform} matches the left side of the following
\name{Basic Elimination Lemma}, which gives a first-order equivalent for
second-order formulas matching its left side:
\begin{lem}[Basic Elimination Lemma]
\label{lem-basic-elim}
Let $p$ be a unary predicate and let $F, G$ be first-order formulas with
equality in which $p$ does not occur. It then holds that
\[\ex p\, (\all x\, (F \lor px) \land
              \all x\, (G \lor \lnot px))
  \;\equiv\;
  \all x\, (F \lor G).\]
\end{lem}
Note that $F$ and $G$ in that proposition may contain free variables, in
particular free occurrences of~$x$, which are then bound by the surrounding
universal quantifiers on the left as well as on the right side of the
proposition.  Rewriting all subformulas headed with $\ex p$ according to
Lemma~\ref{lem-basic-elim} then completes the method asserted by
Theorem~\ref{thm-beh:22-final}.

\subsection{Applications to Predicate Elimination and Decidability}

We now turn to applications of Theorem~\ref{thm-beh:22-final}.  Repeatedly
running the method ensured by that theorem allows to eliminate \emph{all}
predicate quantifiers in a \MONE formula. The result is first-order, but might
be with equality even in cases where the input is without equality, since a
single run of the method already might introduce equality.  This transfer of
Theorem~\ref{thm-beh:22-final} to formulas with several predicate quantifiers
is made precise with the following corollary:
\begin{corr}[Elimination in \QMONE]
\label{corr-elim}
There is an effective method to compute for a given \QMONE formula $F$ of a
formula~$F^{\prime}$ such that
\begin{enumerate}
\item $F^{\prime}$ is a \MONE formula,
\item $F^{\prime} \equiv F$,
\item Predicates that just occur bound by a second-order quantifier in~$F$
  do not occur in $F^{\prime}$,
\item All free individual variables, constants and predicates in
$F^{\prime}$ do occur in~$F$.
\end{enumerate}
\end{corr}

\begin{proof}
Return the result of applying the following equivalence preserving rewritings
to $F$: First, exhaustively\footnotemark\ rewrite subformula occurrences of
the form $\all p\, G$ where $p$ is a predicate with the equivalent formula
$\lnot \ex p\, \lnot G$.  Second, exhaustively rewrite subformula occurrences
of the form $\ex p\, G$ where $p$ is a predicate and $G$ is first-order with
(i.e. $\ex p\, G$ is an innermost second-order quantification) to the
equivalent first-order formula obtained according to
Theorem~\ref{thm-beh:22-final}.  \qed
\end{proof}

\footnotetext{By \name{rewriting exhaustively} we mean that in each step a
  subformula occurrence of the indicated form is replaced and the resulting
  overall formula is subjected again to rewriting, until it does no longer
  contain a subformula of the indicated form.}

\label{sec-pure-cq-applic}
\noindent
Basic formulas of form~(\ref{basic-exists}) in
Theorem~\ref{thm-foqe-monadic-eq} include the case where $m = 0$, that is,
$\ex^{\geq n} x\, \true$. If $F$ is without constants, without free
individual variables and such that all predicate occurrences are quantified,
then the result of applying the method according to Corollary~\ref{corr-elim}
followed by normalization according to Theorem~\ref{thm-foqe-monadic-eq} must
be a Boolean combination of basic formulas of just the form $\ex^{\geq n}
x\, \true$, where $n$ is a number \geqone. A formula $\ex^{\geq n} x\,
\true$ is satisfied by exactly those interpretations whose domain has at least
$n$ distinct members.  A formula $\lnot \ex^{\geq n} x\, \true$ by those
whose domain has less than $n$ members.  As observed in \cite{beh:22}, a
Boolean combination of formulas of the form $\ex^{\geq n} x\, \true$ with
$n \geq 1$ is either true for all domain cardinalities with exception of a
finite number or false for all domain cardinalities with exception of a finite
number.  It is not hard to see that validity and satisfiability of Boolean
combinations of formulas of the form~$\ex^{\geq n} x\, \true$ is
decidable. We discuss this in more depth in Sect.~\ref{sec-qp-pure} below.

The decidability problem for \QMONE can be reduced to the elimination problem.
Hence, the decidability of \QMONE, and thus also of its subclasses \QMON,
\MONE and \MON, follows from Corollary~\ref{corr-elim}, and thus indirectly
from Theorem~\ref{thm-beh:22-final}. The following corollary states this from
two perspectives, validity and satisfiability.

\begin{corr}[Decidability of \QMONE]

\slab{corr-decide-valid} There is an effective method to decide whether a
\QMONE formula is valid.

\slab{corr-decide-sat} There is an effective method to decide whether a
\QMONE formula is satisfiable.
\end{corr}

\begin{proof}
(\ref{corr-decide-valid}) Let $F$ be the given \QMONE formula.  Let $p_1,
  \ldots, p_n$ be all predicates with free occurrences in $F$, let $x_1,
  \ldots, x_m$ be the free individual variables in $F$, and let $c_1, \ldots,
  c_k$ be the constants in $F$.  Let $F^{\prime}$ be the formula \[\all p_1
  \ldots \all p_n \all x_1 \ldots \all x_m \all c_1 \ldots \all
  c_k\, F.\] The formula $F^{\prime}$ is valid if and only if $F$ is valid. When
  applied to $F^{\prime}$, the method according to Corollary~\ref{corr-elim}
  followed by normalization according to Theorem~\ref{thm-foqe-monadic-eq}
  then yields an equivalent formula $F^{\prime\prime}$ which is a Boolean
  combination of formulas of the form $\ex^{\geq n} x\, \true$, where $n
  \geq 1$.  Validity of such Boolean combinations can be decided.

(\ref{corr-decide-sat}) Decidability of satisfiability follows trivially from
  the decidability of validity, since a \QMONE formula is satisfiable if and
  only if its negation, which is also a \QMONE formula, is not valid. However,
  it is also possible to express the involved intermediate steps directly in
  terms of satisfiability: Let $F$ be the given formula.  Let $F^{\prime}$ be
  the formula \[\ex p_1 \ldots \ex p_n \ex x_1 \ldots \ex x_m
  \ex c_1 \ldots \ex c_k\, F,\] where the quantified predicates,
  variables and constants are as specified in the proof of
  Prop.~\ref{corr-decide-valid} Then $F^{\prime}$ is satisfiable if and only if
  $F$ is satisfiable. As in the proof for the decidability of validity, when
  applied to $F^{\prime}$, the method according to Corollary~\ref{corr-elim}
  followed by normalization according to Theorem~\ref{thm-foqe-monadic-eq}
  yields an equivalent formula $F^{\prime\prime}$ which is a Boolean
  combination of formulas of the form $\ex^{\geq n} x\,
  \true$. Satisfiability of such Boolean combinations can be decided.  \qed
\end{proof}

\section{Starting Points: Rewrite Rules and Deciding Propositional Logic}
\label{sec-starting}

The methodical approach of \cite{beh:22} essentially consists in developing
effective methods that operate by rewriting of formulas in an equivalence
preserving way to certain normal forms. The rewriting is done according to a
set of rules, that is, oriented equivalences.  The involved normal forms are
in particular Boolean combinations of certain basic formulas as well as
disjunctive and conjunctive normal form, generalized such that the role of
atoms is played by certain basic formulas.
In this section, a collection of the relevant equivalences and entailments
that are used as rules is presented. In addition, the relationship between
clausal normal forms and decision methods is sketched for propositional logic
and quantified Boolean formulas, along with some comments on history.

\subsection{Equivalences and Entailments for Rewriting Formulas}
\label{sec-prop-equi-dual}

Well-know equivalences and entailments between formulas are listed below as
labeled propositions, such that they can be referenced in the sequel.  Their
choice is mainly motivated by their role in the methods of \cite{beh:22}. A
concordance with the rule labels used in \cite{beh:22} is provided with
Table~\ref{tab-concordance-b22} at the end of the section.  The following
Prop.~\ref{prop-equi} gathers equivalences between formulas.  In
\cite{beh:22}, such equivalences are used as rules for \emph{reversible}
inferences (\de{Regeln für umkehrbare Schlüsse}): they can be applied oriented
from left to right as well oriented from right to left to obtain a formula
that is equivalent to a given formula, but has different syntactic properties.

\begin{prop}[Equivalences Useful for Rewriting]
\label{prop-equi}
We consider second-order logic with equality.  For all formulas $F, G, H$,
quantifiers $Q \in \{\all, \ex\}$, individual variables or predicates $v, w$
and binary connectives $\binop \in \{\land, \lor\}$ the following equivalences
hold:

\smallskip

\noindent
Interaction of negation with other operators
\smallskip

\noindent
\begin{tabular}{L{3.5em}Sl}
\eqlab{rule-not-not} & $\lnot \lnot F \equiv F$.\\
\eqlab{rule-not-and} & $\lnot (F \land G) \equiv \lnot F \lor \lnot G$.\\
\eqlab{rule-not-or} & $\lnot (F \lor G) \equiv \lnot F \land \lnot G$.\\
\eqlab{rule-not-all} & $\lnot \all v\, F \equiv \ex v\, \lnot F$.\\
\eqlab{rule-not-ex} & $\lnot \ex v\, F \equiv \all v\, \lnot F$.\\
\end{tabular}

\medskip\noindent
Associativity, commutativity and idempotence of conjunction and disjunction
\smallskip

\noindent
\begin{tabular}{L{3.5em}Sl}
\eqlab{rule-ao-assoc} &
    $(F \binop G) \binop H \equiv F \binop (G \binop H)$.\\
\eqlab{rule-ao-comm} & $F \binop G \equiv G \binop F$.\\
\eqlab{rule-ao-idem} & $F \binop F \equiv F$.
\end{tabular}

\medskip\noindent
Interaction of truth values with other operators
\smallskip

\noindent
\begin{tabular}{L{3.5em}Sl}
\eqlab{rule-tv-not-t} & $\lnot \true \equiv \false$.\\
\eqlab{rule-tv-not-f} & $\lnot \false \equiv \true$.\\
\eqlab{rule-tv-and-t} & $\true \land F \equiv F$. \hspace{1em}
                       $F \land \true \equiv F$.\\
\eqlab{rule-tv-and-f} & $\false \land F \equiv \false$.  \hspace{1em}
                       $F \land \false \equiv \false$.\\
\eqlab{rule-tv-or-t} & $\true \lor F \equiv \true$. \hspace{1em}
                      $F \lor \true \equiv \true$.\\
\eqlab{rule-tv-or-f} & $\false \lor F \equiv F$. \hspace{1em} 
                      $F \lor \false \equiv F$.\\
\eqlab{rule-tv-q-t} & $Q v\ \true \equiv \true$.\\
\eqlab{rule-tv-q-f} & $Q v\ \false \equiv \false$.
\end{tabular}

\medskip\noindent
Cancellation of complementary formulas
\smallskip

\noindent
\begin{tabular}{L{3.5em}Sl}
\eqlab{rule-complem-and} & $F \land \lnot F \equiv \false$.
                      \hspace{1em} $\lnot F \land F \equiv \false$.\\
\eqlab{rule-complem-or} & $F \lor \lnot F \equiv \true$.
                      \hspace{1em} $\lnot F \lor F \equiv \true$.
\end{tabular}

\medskip\noindent
Distribution among conjunction and disjunction
\smallskip

\noindent
\begin{tabular}{L{3.5em}Sl}
\eqlab{rule-dist-dnf} &
   $F \land (G \lor H) \equiv (F \land G) \lor (F \land H)$.\\
&  $(F \lor G) \land H \equiv (F \land H) \lor (G \land H)$.\\
\eqlab{rule-dist-cnf} &
   $F \lor (G \land H) \equiv (F \lor G) \land (F \lor H)$.\\
&  $(F \land G) \lor H \equiv (F \lor H) \land (G \lor H)$.
\end{tabular}

\medskip\noindent
Quantifier shifting
\smallskip

\noindent
\begin{tabular}{L{3.5em}Sl}
\eqlab{rule-all-out-and} &
  $\all v\, F \land \all v\, G \equiv \all v\, (F \land G)$.\\
\eqlab{rule-ex-out-or} &
 $\ex v\, F \lor \ex v\, G \equiv \ex v\, (F \lor G)$.\\
\eqlab{rule-q-out-ao} & $Q v\, F \binop G \equiv Q v\, (F \binop G)$, 
if $v$ does not occur free in G.\\
            & $F \binop Q v\, G  \equiv Q v\, (F \binop G)$, 
if $v$ does not occur free in F.\\
\end{tabular}

\medskip\noindent
Vacuous quantifiers, quantifier switching and variable renaming
\smallskip

\noindent
\begin{tabular}{L{3.5em}SL{10.5cm}}
\eqlab{rule-quant-drop} & 
$Q v\, F \equiv F$, if $v$ does not occur free in $F$.\\
\eqlab{rule-quant-flip} & $Q v Q w\, F \equiv Q w Q v\, F$.\\
\eqlab{rule-var-rename} &
   $Q v F[v] \equiv Q w F[w]$, if $F[v]$ and $F[w]$ are identical
   with the exception\\ & that $F[w]$ is obtained 
   from $F[v]$ by replacing all free occurrences of $v$ with $w$, and
   vice versa, $F[v]$ from $F[w]$ by 
   replacing all free occurrences of $w$ with $v$.
\end{tabular}

\medskip\noindent
Absorption of entailed conjuncts and entailing disjuncts
\smallskip

\noindent
\begin{tabular}{L{3.5em}Sl}
\eqlab{rule-subs-and-absorp} &
 $F \land G \equiv F$, if $F \entails G$.\\
\eqlab{rule-subs-or-absorp} &
 $F \lor G \equiv G$, if $F \entails G$.
\end{tabular}

\medskip\noindent
Clausal simplifications: tautology reduction, subsumption and
unit reduction
\smallskip

\noindent
Here we consider a \name{matrix} (conjunction or disjunction with arity $\geq
0$) of \name{clauses} (disjunctions or conjunctions, respectively, with
arities $\geq 0$) of \name{basic formulas} or negated basic formulas.
Corresponding to the setting in \cite{beh:22}, basic formulas are not
restricted to atoms but can be arbitrary formulas (see also
footnote~\ref{footnote-dnf} on p.~\pageref{footnote-dnf}). The indicated
operations preserve equivalence of the matrix.
\smallskip

\noindent
\begin{tabular}{L{3.5em}SL{10.5cm}}
\eqlab{rule-taut} & 
A clause that contains a basic formula and its complement can be\\ &
 removed.\\
\eqlab{rule-subs} &
A clause whose members are all contained in another clause
can be\\ & removed.\\

\eqlab{rule-unit} & A member of a clause can be removed if its
complement is\\ & the sole member of another clause in the
matrix.
\end{tabular}

\medskip\noindent
Circumlocution of argument terms\footnotemark
\smallskip

\footnotetext{In Behmann's later manuscripts, the term \de{Umschreibung},
  which might be translated as \name{circumlocution}, appears for the right
  sides of these equivalences.}

\noindent
Let $x$ be an individual variable, let $F[x]$ be a formula, let $t$ be a term
that does contain neither $x$ nor a variable that is bound in $F[x]$ (and thus
is itself also different from $x$ and from any variable bound in $F[x]$), and
let $F[t]$ be $F[x]$ with all free occurrences of $x$ replaced by $t$.  It
then holds that

\smallskip
\noindent
\begin{tabular}{L{3.5em}Sl}
\eqlab{rule-pullout-all} & $F[t]\; \equiv\; \all x\, (x\neq t \lor F[x])$.\\
\eqlab{rule-pullout-ex}  & $F[t]\; \equiv\; \ex x\, (x=t \land F[x])$.
\end{tabular}

\end{prop}

\medskip

\noindent
Proposition~\ref{prop-equi-dual} below also provides a basis for reversible
inferences, but does not state an equivalence of formulas. Instead, it states
an equivalence of statements about formulas.  Let $\f{dual}(F)$ denote the
\defname{dual} of a first- or second-order formula~$F$ (whose only operators
are $\true$, $\false$, $\lnot$, $\land$, $\lor$, $\ex$, and $\all$),
that is, the formula obtained from $F$ by switching $\true$ with $\false$,
$\land$ with $\lor$, and $\all$ with $\ex$. The dual of~$F$ is
equivalent to the negation of $F$ after negating each atom occurrence.

\begin{prop}[Preservation of Equivalence under Duality]
\label{prop-equi-dual}
Let $F$ and $G$ be first- or second-order formulas. It then holds that
\[F \equiv G\; \text{ if and only if }\; \dual(F) \equiv \dual(G).\]
\end{prop}

\medskip

\noindent
Proposition~\ref{prop-entail} shows entailments between formulas that are
considered as rules for inferences that are not reversible (\de{Regeln für
  nicht umkehrbare Schlüsse}). Rewriting with them oriented from left to right
yields a formula that is weaker than or equivalent to the original formula.
This can be useful, for example, since validity is preserved by weaker
formulas, since establishing the entailment relationship can give rise to
equivalence preserving rewritings where the entailment is a precondition, and
in the context of methods like resolution that proceed by enriching a given
formula with entailed formulas.

\begin{prop}[Entailments Useful for Rewriting]
\label{prop-entail}
We consider formulas of first- and second-order logic.

\smallskip

\noindent
For formulas $F, G, H$ and first-order variable $x$ it holds that

\smallskip

\noindent
\begin{tabular}{L{3.5em}SL{10.5cm}}
\enlab{inf-v} & 
  $(F \lor G) \land (H \lor \lnot G) \entails F \lor H$.\\
\enlab{inf-v-q} &
  $\all x\, (F \lor G) \land 
  \all x\, (H \lor \lnot G) \entails \all x\, (F \lor H)$.
\end{tabular}

\medskip

\noindent
Let $F[^+G]$ ($F[^-G]$, resp.) be a formula with a positive (negative, resp.)
occurrence of subformula $G$. Let $F[^+H]$ ($F[^-H]$, resp.)  denote $F$ with
the occurrence of $G$ replaced by formula~$H$. It then holds that

\smallskip

\noindent
\begin{tabular}{L{3.5em}SL{10.5cm}}
\enlab{inf-vbar-pos} &
$F[^+G] \land (\lnot G \lor H) \entails F[^+H]$.\\
\enlab{inf-vbar-neg} &
$F[^-G] \land (G \lor \lnot H) \entails F[^-H]$.\\
\end{tabular}

\medskip

\noindent
Let $F[^+G]$, $F[^-G]$, $G$, $H$ be as specified before, with the exception
that the first-order variables $x_1, \ldots, x_n$ possibly occur free in $G$,
$H$ and also elsewhere in $F[^+G]$ or $F[^-G]$, respectively. It then holds
that

\smallskip

\noindent
\begin{tabular}{L{3.5em}SL{10.5cm}}
\enlab{inf-vbar-pos-star} & $F[^+G] \land 
  \all x_1 \ldots \all x_n\, (\lnot G \lor H) 
  \entails F[^+H]$.\\ 
\enlab{inf-vbar-neg-star} & $F[^-G] \land 
  \all x_1 \ldots \all x_n\, (G \lor \lnot H)
  \entails F[^-H]$.\\
\end{tabular}

\end{prop}

\medskip

\noindent
Table~\ref{tab-concordance-b22} provides a concordance of the equivalences and
entailments stated as propositions in this section and the computation rules
(\de{Rechenregeln}) of \cite{beh:22}.  Rules~I--IV$^*$, VI, VII and IX--XI are
equivalences of formulas. Rule~VIII is an equivalence of statements about
formulas.  Rules V--$\overline{\text{V}}^*$ are entailments.  The starred
rules concern quantifiers.  The correspondence given here is not in all cases
one-to-one: The specification in \cite{beh:22} is only informal, such that we
have to interpret some details in a specific way or quietly apply
generalizations that are straightforward from today's point of view.  Also, in
\cite{beh:22} conjunction and disjunction are understood directly as $n$-ary
operators, such that in some cases the effect of the rules from \cite{beh:22}
can only be achieved by repeated application of the listed corresponding
equivalences.  Equivalences~\ref{rule-tv-not-t}--\ref{rule-tv-or-f} which
concern embedded truth-value operators are not labeled in \cite{beh:22}, but
listed on p.~182 and~188.

\begin{table}
\label{tab-concordance-b22}
\small
\centering
\begin{tabular}{lSlSl}
\textit{Label} & \textit{Correspondence} & \textit{Explanation}\\\toprule
I & \ref{rule-not-not} & double negation\\
II & \ref{rule-ao-assoc} & associativity of conjunction and disjunction\\
II$^*$ & \ref{rule-q-out-ao} & modifying the scope of quantifiers\\
III & \ref{rule-ao-comm} & commutativity of conjunction and disjunction\\
III$^*$ & \ref{rule-quant-flip} & permuting quantifiers of the same type\\
IV & \ref{rule-ao-idem} & idempotence of conjunction and disjunction\\
IV$^*$ & \ref{rule-all-out-and}, \ref{rule-ex-out-or}
& merging quantifiers over conjunctions and disjunctions\\
V & \ref{inf-v} & generalized propositional resolvent\\
V$^*$ & \ref{inf-v-q} 
  & generalized propositional resolvent, with quantification\footnotemark\\
$\overline{\text{V}}$ & \ref{inf-vbar-pos}, \ref{inf-vbar-neg}
& substitution of subformula by implied/implying formula\\
$\overline{\text{V}}^\star$ &
\ref{inf-vbar-pos-star}, \ref{inf-vbar-neg-star} & 
like $\overline{\text{V}}$, with quantification over the implication\\
VI & \ref{rule-not-and}, \ref{rule-not-or} & De Morgan's laws\\
VI$^*$ & \ref{rule-not-all}, \ref{rule-not-ex} & negated quantifications\\
VII & \ref{rule-dist-cnf}, \ref{rule-dist-dnf}
 & distribution among conjunction and disjunction\\
VIII & Prop.~\ref{prop-equi-dual}
     & preservation of equivalence under duality\\
IX & \ref{rule-subs}, \ref{rule-unit}
   & subsumption and unit reduction in clausal forms\\
X & \ref{rule-subs-and-absorp}, \ref{rule-subs-or-absorp}
& absorption of entailed conjuncts and entailing disjuncts\\
XI & \ref{rule-pullout-all}, \ref{rule-pullout-ex} &
     circumlocution of argument terms\\
\bottomrule
\end{tabular}
\smallskip

\caption{Concordance with the rule labels in \cite{beh:22}.}
\end{table}

\footnotetext{The description of V$^*$ in the appendix of \cite{beh:22} might
  suggest also the entailment $\ex x\, (F \lor G) \land \ex x\, (H
  \lor \lnot G) \entails \ex x\, (F \lor H)$, which does not hold in
  general.}

\subsection{Deciding Propositional Logic}
\label{sec-propositional}

As outlined in \cite{beh:22}, a propositional formula $F$ whose atoms are
$p_1, \ldots, p_n$ is valid if and only if the quantified Boolean
formula \[\all p_1 \ldots \all p_n\, F\] is true, and satisfiable if and
only if \[\ex p_1 \ldots \ex p_n\, F\] is true.  The Boolean
quantifiers immediately indicate the substitution method
(\de{Einsetzungsverfahren}) to decide propositional validity and
satisfiability: Let $F[p \mapsto G]$ denote $F$ under substitution of all
occurrences of atom $p$ by $G$. Then $\all p\, F$ is true if and only if
\emph{for all} truth value constants~$G \in \{\true, \false\}$ it holds that
$F[p \mapsto G] \equiv \true$, while $\ex p\, F$ is true if and only if
\emph{there exists} a truth value constant~$G \in \{\true, \false\}$ such that
$F[p \mapsto G] \equiv \true$. For nested quantifiers, the corresponding
combined substitutions have to be evaluated.

\nocite{hilbert:2013}

A second method to decide propositional validity, attributed in \cite{beh:22}
to Bernays \cite{bernays:habil}, consists in producing a conjunctive normal
form, which is valid if and only if each of its clauses is valid, that is,
contains a literal and its complement.  In \cite{beh:22} it is noticed that
propositional satisfiability can be analogously decided by conversion to
disjunctive normal form:%
\footnote{Actually, \label{footnote-dnf} the names \de{konjunktive Normalform}
  and \de{disjunktive Normalform} as well as the precise realization of their
  duality seem due to \cite{beh:22} -- see \cite[p.~166]{church:book}.  In
  \cite{bernays:habil}, the work referenced by \cite{beh:22}, Bernays just
  speaks of \de{Normalform}, in the sense of conjunctive normal form. Footnote
  on p.~13 in \cite{bernays:habil} suggests that the duality with disjunctive
  normal form was not common knowledge at that time: \de{Es wäre ein Irrtum,
    nach Analogie [...] zu vermuten, dass ein Produkt-Ausdruck dann und nur
    dann eine beweisbare Formel ist, wenn mindestens eines der Glieder eine
    beweisbare Formel ist.}  (Can be paraphrased as: It would be an error to
  conjecture in analogy that a disjunction is a valid formula if and only if
  at least one of its disjuncts is a valid formula.)  See also
  \cite[Note~299]{church:book}.  In \cite{mancosu:zach:2015} it is observed
  that Hilbert considered conjunctive and disjunctive normal forms as well as
  decidability of propositional logic already in 1905. As further noted in
  \cite{mancosu:zach:2015}, Behmann writes on 27 December 1927 to Scholz
  \cite[Kasten 3, I~63]{beh:nl} that, as far as he remembers, he had learned
  about the solution of the decision problem for propositional logic using
  normal forms directly from Hilbert.

   As explained in \cite[p.~185,
    footnote~19]{beh:22}, Behmann uses \name{normal form} also more generally
  for a conjunction of disjunctions of arbitrary formulas, not necessary
  literals, or a disjunction of conjunctions of arbitrary formulas,
  respectively.}
A disjunctive normal form is satisfiable if and only if at least one of its
(conjunctive) clauses is satisfiable, that is, does not contain a literal and
its complement.  As indicated in \cite{beh:22}, quantified Boolean formulas
that contain universal as well as existential quantifiers then could be
evaluated by successive use of both normal forms.  A third method to decide
quantified Boolean formulas by moving quantifiers inward is also shown in
\cite{beh:22}. It is sketched below in Sect.~\ref{sec-qbf-inner}.

\section{Counting Quantifier Normal Form  for \MONE Formulas}
\label{sec-ccnf}

\label{sec-nfqe-mfe}

The method of \cite{beh:22} for the elimination of second-order quantifiers in
\QMONE formulas, stated here as Corollary~\ref{corr-elim}, based on
Theorem~\ref{thm-beh:22-final}, involves rewriting \MONE formulas to
equivalent formulas that are constructed from a restricted set of basic
formulas.  In the target format, the only quantifiers permitted are counting
quantifiers, and only in occurrences where their scopes are not nested.

The development of methods to express a class of formulas by Boolean
combinations of certain basic formulas, typically with respect to a given
background theory, is the core of \name{elimination of quantifiers}, the
prevailing program of model theory in the 1920s (see
e.g.~\cite[Sect.~1.5]{chang:keisler}, \cite[Sect.~2.7]{hodges:shorter}).
Although \cite{beh:22} does not explicitly reference works clearly associated
with that program, we present here the method to solve the ``problem of normal
form'' (\name{Problem der Normalform}) given in \cite[\S~20]{beh:22} as
instance of such a quantifier elimination method -- with respect to the empty
theory -- following the template given in \cite[Sect.~1.5]{chang:keisler}.

The presentation in this section proceeds ``bottom-up''. First, counting
quantifiers, which are important as constituents of the basic formulas, are
discussed, then the construction of the normal form is shown.

\subsection{Counting Quantifiers}
\label{sec-cq}

As already indicated, formulas with counting quantifiers belong to the basic
formulas of the envisaged target format.  A counting quantifier $\ex^{\geq
  n} x$, where $n$ is a natural number \geqone, expresses existence of at
least $n$ individuals~$x$.  It is well known that these quantifiers can be
defined in terms of standard first-order quantifiers and equality literals.
There are two obvious possibilities to do so, which are both used in
\cite{beh:22}, as we will see below in
Sect.~\ref{sec-two-counting-expansions}. The following proposition lists them
as equivalences:

\begin{prop}[First-Order Expansions of Existential Counting Quantifiers]
\label{prop-exp-counting-ex}
Let $F[x]$ be a first-order formula which possibly has free occurrences of
variable $x$ and let $n$ be a natural number \geqone. Let $x_1, \ldots, x_n$
be distinct variables that are fresh (that is, different from $x$ and not
occurring in $F[x]$), and, for $i \in \{1,\ldots,n\}$, let $F[x_i]$ denote
$F[x]$ with the free occurrences of $x$ replaced by $x_i$.  It then holds that

\medskip

\slab{prop-exp-eq-geq-poly}
$\ex^{\geq n} x\, F[x]\; \equiv\;
  \ex x_1 \ldots \ex x_n\, (
  \bigwedge_{1 \leq i \leq n}
  F[x_i]
  \; \land
  \bigwedge_{i < j \leq n}
  x_i \neq x_j).$

\slab{prop-exp-eq-geq-lin}
$\ex^{\geq n} x\, F[x]\; \equiv\;
  \all x_1 \ldots \all x_{n-1} \ex x\,
   (F[x]
   \, \land
   \bigwedge_{1 \leq i < n} x \neq x_i).$
   
\end{prop}

\noindent
The statement $\ex^{\geq n} x\, \true$ expresses that the domain has at least
$n$ members.  Further properties of $\ex^{\geq n}$ are gathered in the
following proposition:

\begin{prop}[Properties of Existential Counting Quantifiers]
\label{prop-cqex-properties}
For all first-order formulas $F$ and natural numbers~$n, m \geq 1$ it holds
that

\smallskip

\slab{prop-cqex-false} $\ex^{\geq n} x\, \false \equiv \false$.

\slab{prop-cqex-one} $\ex^{\geq 1} x\, F \equiv \ex x\, F$.

\slab{prop-cqex-one-true} $\ex^{\geq 1} x\, \true \equiv \true$.

\slab{prop-cqex-entails}
$\ex^{\geq n} x\, F \entails \ex^{\geq m} x\, F$, if $m \leq n$.
\end{prop}

\noindent
Let $\all^{\cca{n}} x$, where $n$ is a natural number \geqone, be further
counting quantifiers, defined as shorthand for $\lnot \ex^{\geq n} x
\lnot$ .  They express ``for all with the exception of less than $n$
individuals~$x$ it holds that''. In analogy to
Prop.~\ref{prop-exp-counting-ex}, they can be expanded in two way as follows:
\begin{prop}[First-Order Expansions of Universal Counting Quantifiers]
\label{prop-exp-counting-all}
Let $F[x], n, x_1, \ldots x_n, F[x_i]$ be as specified in
Prop.~\ref{prop-exp-counting-ex}. It then holds that

\medskip

\slab{prop-exp-eq-lt-poly}
$\all^{\cca{n}} x\, F[x]\; \equiv\;
  \all x_1 \ldots \all x_n\,
  (\bigvee_{1 \leq i \leq n}
  F[x_i]
  \; \lor
  \bigvee_{i < j \leq n}
  x_i = x_j).$

\slab{prop-exp-eq-lt-lin}
$\all^{\cca{n}} x\, F[x]\; \equiv\;
  \ex x_1 \ldots \ex x_{n-1} \all x\,
   (F[x]
   \, \lor
   \bigvee_{1 \leq i < n} x = x_i).$
\end{prop}

\noindent
The statement $\all^{\cca{n}} x\, \false$ expresses that the domain has
less than $n$ members.  Further properties of $\all^{\cca{n}}$ are gathered
in the following proposition, analogously to Prop.~\ref{prop-cqex-properties}:

\begin{prop}[Properties of Universal Counting Quantifiers]
\label{prop-cqall-properties}
For all first-order formulas $F$ and natural numbers~$n, m \geq 1$ it holds
that

\smallskip

\slab{prop-cqall-true}
$\all^{\cca{n}} x\, \true \equiv \true$.

\slab{prop-cqall-one} 
$\all^{\cca{1}} x\, F \equiv \all x\, F$

\slab{prop-cqall-one-false} $\all^{\cca{1}} x\, \false \equiv \false$.

\slab{prop-cqall-entails}
$\all^{\cca{n}} x\, F \entails \all^{\cca{m}} x\, F$, if $n \leq m$.

\end{prop}

\noindent
In \cite{beh:22} a dedicated symbolic notation that expresses counting
quantifiers for numbers~$n$ with $n$ stacked arcs is introduced. It covers
basic formulas of form (\ref{basic-exists}.) in
Theorem~\ref{thm-foqe-monadic-eq} and their negations.  For example, if
$\alpha[x]$ and $\beta[x]$ are formulas with free variable~$x$,
then \[\underparen{\underparen{\alpha\dot{\beta}}}\] stands for
$\all^{\cca{2}} x\, (\alpha[x] \lor \lnot \beta[x])$ (or, equivalently,
$\lnot \ex^{\geq 2} x\, (\lnot \alpha[x] \land \beta[x])$) and
\[\overparen{\overparen{\alpha\dot{\beta}\rule{0pt}{11pt}}}\] for
$\ex^{\geq 2} x\, (\alpha[x] \land \lnot \beta[x])$. The role of $\true$ and
$\false$ for the empty conjunction and disjunction, respectively, is played in
\cite{beh:22} by symbols $V$ and
\raisebox{\depth}{\rotatebox[origin=c]{180}{$V$}}
for the universal and the empty class,
which would correspond to versions of $\true$ and $\false$ that are notated
like unary predicates.

\label{sec-qp-pure}

As discussed in Sect.~\ref{sec-pure-cq-applic}, elimination of \emph{all}
predicate quantifiers in a formula without constants, without free individual
variables and without free predicate occurrences yields formulas that can be
represented by a Boolean combination of basic formulas of just the form
$\ex^{\geq n} x\, \true$ with $n \geq 1$. We call such a Boolean
combination a \name{pure counting formula}.  A straightforward way to decide
validity of pure counting formulas is by conversion to conjunctive normal
form, replacing literals $\lnot \ex^{\geq n} x\, \true$ with the
equivalent formula $\all^{\cca{n}} x\, \false$ (which we now also accept as basic
formulas) and simplifying each clause by Prop.~\ref{prop-cqex-entails}
and~\ref{prop-cqall-entails} such that each clause is either empty (that is,
is $\false$), contains a single basic formula or contains exactly two basic
formulas, one with existential and the other with universal counting
quantifier. The formula is then valid if and only if each clause is valid,
that is, (i) is not empty, or (ii) contains just $\ex^{\geq 1} x\, \true$,
or (iii) contains $\ex^{n} x\, \true$ and $\all^{\cca{m}} x\, \false$,
where $n \leq m$.

Dually, satisfiability of pure counting formulas can be decided by analogous
transformation to disjunctive normal form. The formula is then satisfiable if
and only if each (conjunctive) clause is satisfiable, that is, (i) is empty,
or (ii) contains just a single basic formula that is different from
$\all^{\cca{1}} x\, \false$, or (iii) contains two basic formulas
$\ex^{n} x\, \true$ and $\all^{\cca{m}} x\, \false$, where $n < m$.

Following \cite{beh:22}, this disjunctive form illustrates how domain
cardinalities are constrained by a pure counting formula: Each (conjunctive)
clause justifies a series of numbers with a lower limit or with lower as well
as upper limits as domain cardinalities.  A pure counting formula is thus
either true for all domain cardinalities with exception of a finite number or
false for all domain cardinalities with exception of a finite number.  For
sufficiently large domains, in particular for all infinite domains, a pure
counting formula is then either valid or unsatisfiable.

\subsection{Conversion to Counting Quantifier Normal Form}

We now approaching the proof of Theorem~\ref{thm-foqe-monadic-eq}. It is
helpful to isolate the following key step as a lemma on its own:
\begin{lem}[Auxiliary Elimination Lemma for First-Order Logic with Equality]
\label{lem-folelim-core}
For all first-order formulas $F[x]$, where $x$ is an individual variable that
possibly occurs free in $F[x]$, for all sequences $T = \{t_1, \ldots, t_n\}$
of $n \geq 0$ distinct constants or variables which are different from~$x$ and
do not occur in $F[x]$ and for all integers $m \geq 1$ let
$\f{NDT}_{F[x],T}(m)$ be the formula
\[\f{NDT}_{F[x],T}(m)\; \eqdef\;
\bigwedge_{\substack{S \subseteq T\\ \left\vert{S}\right\vert = m}}
     (\bigvee_{t \in S} \lnot F[t] \lor 
      \bigvee_{\substack{t_i, t_j \in S\\ i < j}} t_i = t_j),\]
where $F[t]$ denotes $F[x]$ with all free occurrences of $x$ replaced by $t$.
Then \[\ex x\, (F[x] \land \!\! \displaystyle\bigwedge_{1 \leq i \leq n} \!\! x
\neq t_i)\] is equivalent to

\medskip

\slab{lem-folelim-core-disj}
\vspace{-8pt}
\[\displaystyle\bigvee_{1 \leq m \leq n} \!\!
   \ex^{\geq m} x\, (F[x] \land \f{NDT}_{F[x],T}(m))
 \;\lor\; \ex^{\geq n+1} x\, F[x].\]

\slab{lem-folelim-core-conj}
\vspace{-8pt}
\[\ex^{\geq 1} x\, F[x]\; \land
 \displaystyle\bigwedge_{1 \leq m \leq n} \!\!
   \ex^{\geq m+1} x\, (F[x] \lor \f{NDT}_{F[x],T}(m)).\]
\end{lem}

\noindent
The shorthand $\f{NDT}_{F[x],T}(m)$ in Lemma~\ref{lem-folelim-core}
(suggesting \name{No $m$ are Denoted by a Term}) expresses that it is false
that $F$ applies to $m$ different individuals denoted by terms $t \in T$,
which is immediate from the following alternate way to write its definition:
\begin{equation}
\label{eq-form-max-m-f-intuition}
\f{NDT}_{F[x],T}(m)\; \equiv\;
\bigwedge_{\substack{S \subseteq T\\ \left\vert{S}\right\vert = m}}
     \lnot (\bigwedge_{t \in S} F[t] \land
      \bigwedge_{\substack{t_j, t_k \in S\\ j < k}} t_j \neq t_k).
\end{equation}

\noindent
Propositions~\ref{lem-folelim-core-disj} and~\ref{lem-folelim-core-conj} show
a disjunctive form and a conjunctive form in which the existential first-order
quantifier in $\ex x\, (F[x] \land \bigwedge_{1 \leq i \leq n} x \neq t_i)$ can
be ``eliminated'' in favor of counting quantifiers whose scopes do not include
the disequalities $x \neq t_i$. One form can be obtained from the other by
distributing connectives and simplifying.  As remarked in \cite{beh:22},
Lemma~\ref{lem-folelim-core} can be proven on the basis that the left side of
the proposition, that is, $\ex x\, (F[x] \land \bigwedge_{1 \leq i \leq n} x
\neq t_i)$, is implied by each disjunct of the right side of the disjunctive
form (Lemma~\ref{lem-folelim-core-disj}) and implies each conjunct of the
right side of the conjunctive form (Lemma~\ref{lem-folelim-core-conj}).

Based on Lemma~\ref{lem-folelim-core}, we now prove
Theorem~\ref{thm-foqe-monadic-eq}, the main theorem of ``quantifier
elimination'' for monadic first-order formulas with equality, following
roughly the template from \cite[Sect.~1.5]{chang:keisler}.  As explained in
Sect.~\ref{sec-meta-methods}, we explicitly assert the \emph{existence of an
  effective method} in the theorem.

\label{sec-thm-foqe-monadic-eq-fullproof}

\begin{thmref}{Counting Quantifier Normal Form for \MONE}
{\ref{thm-foqe-monadic-eq}} 
\noindent
There is an effective method to compute from a given
\MONE formula~$F$ a
formula~$F^{\prime}$ such that
\begin{enumerate}
\item \label{ref-thm-foqe-monadic-eq-item-boolean} $F^{\prime}$ is a Boolean combination of basic
formulas of the form:
\begin{enumerate}[label=(\alph*),ref=\alph*]
\item \label{ref-basic-nullary} $p$, where $p$ is a nullary predicate,
\item \label{ref-basic-unary} $pt$,
   where $p$ is a unary predicate and $t$ is a constant or an 
  individual variable,
\item \label{ref-basic-equality} $t = s$, where each of $t, s$ is a constant or an 
individual variable,
\item \label{ref-basic-exists} $\ex^{\geq n} x\, \bigwedge_{1 \leq i \leq
  m} L_i[x]$, where $n \geq 1$, $m \geq 0$ and the $L_i[x]$ are pairwise
  different and pairwise non-complementary positive or negative literals with
  a unary predicate applied to the individual variable~$x$,
\end{enumerate}
\item \label{ref-thm-foqe-monadic-eq-item-equiv} $F^{\prime} \equiv F$,
\item \label{ref-thm-foqe-monadic-eq-item-free} 
All free individual variables, constants and predicates in $F^{\prime}$ do
  occur in~$F$.
\end{enumerate}
\end{thmref}

\begin{proof}
Existence of a method that meets
items~(\ref{ref-thm-foqe-monadic-eq-item-boolean}.)
and~(\ref{ref-thm-foqe-monadic-eq-item-equiv}.) of the theorem statement
follows if (i) every atomic formula of the input class is a basic formula, and
(ii) there is an effective method to compute for all formulas of the form
$\ex v\, F$ where $v$ is an individual variable and $F$ a Boolean
combination of basic formulas an equivalent Boolean combination of basic
formulas.  The overall method to compute for a given formula in the input
class an equivalent Boolean combination of basic formulas then consists in
exhaustively rewriting (as explained in the proof of
Corollary~\ref{corr-elim}) subformula occurrences of the form of case~(ii) to
a Boolean combination of basic formulas.

Property~(i) is easy to see: Every atomic \MONE formula is clearly a basic
formula of the specified form (\ref{basic-nullary}.), (\ref{basic-unary}.) or
(\ref{basic-equality}.).

We now prove~(ii): Let $F$ be a Boolean combination of basic formulas and let
$v$ be a variable. We show how the formula $\ex v\, F$ can be converted
with equivalence preserving formula rewritings to a Boolean combination of
basic formulas. Unless indicated otherwise, the referenced equivalences are
applied there from left to right.

\begin{enumerate}[leftmargin=2.5em]
\item Convert $F$ to a disjunction~$F_1 \lor \ldots \lor F_n$, where $n \geq
  0$, of conjunctions of basic formulas or negated basic formulas
  (\ref{rule-not-not}--\ref{rule-not-or}, \ref{rule-dist-dnf}). The result is
  a generalization of disjunctive normal form, where basic formulas play the
  role of atoms.

\item Distribute the existential quantifier over the disjunction
  to obtain $\ex v\, F_1 \lor \ldots \lor \ex v\, F_n$, which is
  equivalent to $\ex v\, F$ (\ref{rule-ex-out-or} right to left).

\item Convert each disjunct $\ex v\, F_i$, for $1 \leq i \leq n$,
  separately to a Boolean combination~$F^{\prime}_i$ of basic formulas and
  disjoin the results. The resulting formula $F^{\prime}_1 \lor \ldots \lor
  F^{\prime}_n$ is a Boolean combination of basic formulas which is equivalent to
  $\ex v F$. The following equivalence preserving transformation steps are
  applied to each disjunct $\ex v\, F_i$:
\end{enumerate}
  
\begin{enumerate}[label={3.\arabic*.},leftmargin=2.5em] 

  \item In case there are two complementary conjuncts return with $F^{\prime}_i
     = \false$ (\ref{rule-ao-assoc}--\ref{rule-ao-comm},
     \ref{rule-complem-and}, \ref{rule-tv-and-f}, \ref{rule-tv-q-f}).

  \item In case $t \neq t$ is a conjunct, where $t$ is a term, return with
    $F^{\prime}_i = \false$ (\ref{rule-tv-and-f}, \ref{rule-tv-q-f}).

  \item Remove conjuncts of the form $t = t$, where $t$ is a term 
   (\ref{rule-tv-and-t}).

  \item Orient conjuncts $t = v$ and $t \neq v$, where $t$ is a term
    (different from $v$, as ensured by the previous steps) to the forms $v =
    t$ and $v \neq t$, respectively. 

 \item Reorder the conjuncts, remove duplicate conjuncts, and propagate the
   existential quantifier inward such that exactly those conjuncts are in its
   scope in which $v$ does occur as a free variable
   (\ref{rule-ao-assoc}--\ref{rule-ao-idem}, \ref{rule-q-out-ao}
   from right to left).
   The resulting formula then has the form
 \[G \land \ex v\, (\bigwedge_{1 \leq i \leq k} L_i[v] \land \bigwedge_{1 \leq i \leq l} v \neq t_i \land
  \bigwedge_{1 \leq i \leq m} v = u_i),\]
  
where $k, l, m \geq 0$, the $L_i[v]$ are pairwise different positive or
negative literals with an unary predicate applied to $v$ and the $t_i$ and
$u_i$ are terms.  In particular, the $t_i$ are pairwise different and also
different from $v$.  The subformula $G$ is contains the conjuncts with no free
occurrences of $v$.

 \item In case $k = l = m = 0$, return with $F^{\prime}_i = G$.

\item In case $m > 0$, return with
   \[F^{\prime}_i\; =\; G \land \bigwedge_{1 \leq i \leq k} L_i[u_1] 
   \land \bigwedge_{1 \leq i \leq n} u_1 \neq t_i \land \bigwedge_{1 < i \leq
     m} u_1 = u_i,\] where the $L_i[u_1]$ are obtained from the literals
   $L_i[v]$ by replacing their argument variable $v$ with $u_1$.  This
   transformation step is justified by applying
   \ref{rule-pullout-ex} from right to left.

\item Finally, in the remaining case $m = 0$, return with
   \[F^{\prime}_i\; =\; G \land G^{\prime},\]
   where $G^{\prime}$ is a Boolean combination of basic formulas that is
    equivalent to
   \[\ex v\, (\bigwedge_{1 \leq i \leq k} L_i[v] \land \bigwedge_{1 \leq i \leq l}
   v \neq t_i),\] as obtained according to either
   Lemma~\ref{lem-folelim-core-disj} or \ref{lem-folelim-core-conj}. 
\end{enumerate}

\noindent
Item (\ref{ref-thm-foqe-monadic-eq-item-free}.) of the theorem statement, that
is, all free individual variables, constants and predicates in $F^{\prime}$ do
occur in~$F$, follows since in the transformations involved to compute
$F^{\prime}$ at no point bound variables are moved outside the scope of their
binding quantifier, the only variables introduced in the transformations are
bound by counting quantifiers, and neither constants nor predicates are newly
introduced.  \ \qed
\end{proof}

\noindent
We conclude this subsection with some examples that illustrate the conversion
of equality literals to counting quantifiers performed by the method of
Theorem~\ref{thm-foqe-monadic-eq}. In the following example the input formula
contains positive equality atoms:

\begin{equation}
\begin{array}{ll}
 & \ex x\, (px 
  \land x \neq a \land x \neq b \land x = c \land x = d)\\
\equiv
 & pc 
  \land c \neq a \land c \neq b \land c = d.
\end{array}
\end{equation}

\noindent
The following two sequences of equivalences show results obtained if $\f{NDT}$
is expanded according to Lemma~\ref{lem-folelim-core-disj}:
\begin{equation}
\begin{array}{ll}
& \ex x\, (px \land x \neq a)\\
\equiv &
  (\ex^{\geq 1} x\, px \land \f{NDT}_{px, \{a\}}(1))\; \lor\; \ex^{\geq 2} x\, px\\
\equiv & 
  (\ex^{\geq 1} x\, px \land \lnot pa)\; \lor\;
    \ex^{\geq 2} x\, px.
\end{array}
\end{equation}
\begin{equation}
\begin{array}{lll}
& \ex x\, (px \land x \neq a \land x \neq b)\\
\equiv & 
     (\ex^{\geq 1} x\, px \land \f{NDT}_{px, \{a,b\}}(1)) & \lor\\
   & (\ex^{\geq 2} x\, px \land \f{NDT}_{px, \{a,b\}}(2)) & \lor\\
   & \ex^{\geq 3} x\, px\\
\equiv & 
     (\ex^{\geq 1} x\, px \land 
     \lnot pa \land \lnot pb)
   & \lor\\
   & (\ex^{\geq 2} x\, px \land 
      (\lnot pa \lor \lnot pb \lor a=b)) & \lor\\
   & \ex^{\geq 3} x\, px.
\end{array}
\end{equation}

\noindent 
For the same input formulas, if $\f{NDT}$ is expanded according to
Lemma~\ref{lem-folelim-core-conj}, then we obtain the following
results:
\begin{equation}
\begin{array}{ll}
& \ex x\, (px \land x \neq a)\\
\equiv &
  \ex^{\geq 1} x\, px \land (\ex^{\geq 2} x\, px \lor
\f{NDT}_{px, \{a\}}(1))\\
\equiv &
  \ex^{\geq 1} x\, px \land (\ex^{\geq 2} x\, px \lor \lnot
pa).
\end{array}
\end{equation}
\begin{equation}
\begin{array}{lll}
 & \ex x\, (px \land x \neq a \land x \neq b)\\
\equiv
   & \ex^{\geq 1} x\, px & \land\\
   & (\ex^{\geq 2} x\, px \lor \f{NDT}_{px, \{a,b\}}(1)) & \land\\
   & (\ex^{\geq 3} x\, px \lor \f{NDT}_{px, \{a,b\}}(2))\\
\equiv
   & \ex^{\geq 1} x\, px & \land\\
   & (\ex^{\geq 2} x\, px \lor 
      (\lnot pa \land \lnot pb)) & \land\\
   & (\ex^{\geq 3} x\, px \lor 
     \lnot pa \lor \lnot pb \lor a=b).
\end{array}
\end{equation}

\section{Predicate Elimination for \MONE}
\label{sec-elim-eq}

In this section Theorem~\ref{thm-beh:22-final}, the core result of
\cite{beh:22}, is proven. It is first shown that an existential predicate
quantification of a formula in the counting quantifier normal form (as
produced by the method of Theorem~\ref{thm-foqe-monadic-eq}) can be converted
such that all occurrences of predicate quantification are in subformulas of a
particular form, which we call here \defname{Generalized
  Eliminationshauptform} because it generalizes the \de{Eliminationshauptform}
((\ref{eq-hauptform-informal}) on p.~\pageref{eq-hauptform-informal}), by the
use of counting quantifiers instead of standard quantifiers.  Like the proper
\de{Eliminationshauptform}, the \dename{Generalized Eliminationshauptform} can
be converted to an equivalent formula that allows to eliminate the
second-order quantification upon~$p$ with Lemma~\ref{lem-basic-elim}.

It is possible to convert the \de{Generalized Eliminationshauptform} first to
a formula with a subformula of the proper \de{Eliminationshauptform} and then
perform elimination just based on the latter. However, following
\cite{beh:22}, elimination is shown here directly for the \de{Generalized
  Eliminationshauptform}. The proof actually involves a conversion such that
the predicate quantification has the form of the special case of the proper
\de{Eliminationshauptform} where the constituents $\bigwedge_{1 \leq i \leq c}
\ex x\, (C_i[x] \land px)$ and $\bigwedge_{1 \leq i \leq d} \ex x\, 
(D_i[x] \land \lnot px)$ are ``empty'', that is, $c = d = 0$, which allows to
perform elimination directly with Lemma~\ref{lem-basic-elim}.  Although
technically covered by the general case with equality, we will discuss the
simpler handling of the proper \de{Eliminationshauptform} in the next section.
The proof of Theorem~\ref{thm-beh:22-final} is split up into several
lemmas. First, the construction of the \de{Generalized Eliminationshauptform}
is shown.

\newcommand{\va}{x}
\newcommand{\vb}{y}
\newcommand{\vc}{u}
\newcommand{\vd}{v}

\begin{lem}[Constructing the \de{Generalized Eliminationshauptform}
for \MONE Formulas]
\label{lem-construction-eq}
There is an effective
method to compute for a given unary predicate $p$ and \MONE formula~$F$ 
formula $F^{\prime}$ such that
\begin{enumerate}
\item $F^{\prime}$ is a \QMONE formula,
\item $F^{\prime} \equiv \ex p F$,
\item $p$ is the only quantified predicate in $F^{\prime}$,
\item \label{form-eq-sub-p} All occurrences of $p$ in $F^{\prime}$ are in
  positive occurrences of subformulas of the form
\[\begin{array}{llll}
\ex p
& (   \bigwedge_{1 \leq i \leq a} & \all x^{\cca{a_i}}\, 
      (A_i[x] \lor px) & \land\\
& \hparen \bigwedge_{1 \leq i \leq b} & \all x^{\cca{b_i}}\,
      (B_i[x] \lor \lnot px) & \land\\
& \hparen \bigwedge_{1 \leq i \leq c} & \ex x^{\geq c_i}\, (C_i[x] \land px) & \land\\
& \hparen \bigwedge_{1 \leq i \leq d} & \ex x^{\geq d_i}\, (D_i[x] \land \lnot px)),
\end{array}\]
where $a$, $b$, $c$, $d$ are natural numbers $\geq 0$ and for the referenced
values of $i$ the $a_i$, $b_i$, $c_i$, $d_i$ are natural numbers $\geq 1$, and
the $A_i[x], B_i[x], C_i[x], D_i[x]$ are first-order formulas in which
$p$ does not occur,
\item All free individual variables, constants and predicates in
$F^{\prime}$ do occur in~$F$.
\end{enumerate}
\end{lem}

\begin{proof}
The result formula $F^{\prime}$ is obtained from $\ex p\, F$ by a sequence
of transformations that preserve equivalence and do neither introduce fresh
constants nor predicates nor variables without a binding quantifier.  First,
transform $F$ according to Theorem~\ref{thm-foqe-monadic-eq} to counting
quantifier normal form, that is, to a Boolean combination of basic formulas in
the sense of Theorem~\ref{thm-foqe-monadic-eq}.  Convert the result to a
disjunction of conjunctions of such basic formulas or negated such basic
formulas (like a disjunctive normal form, but with basic formulas in the role
of atoms).  Remove disjuncts that contain complementary conjuncts
(\ref{rule-taut}).

Let $F^{\prime\prime}$ be this intermediate result, which is equivalent to
$F$.  Proceed with transforming $\ex p\, F^{\prime\prime}$. Distribute the
existential quantifier upon~$p$ over the disjuncts (\ref{rule-ex-out-or} right
to left), remove it from disjuncts in which $p$ does not occur
(\ref{rule-quant-drop}), reorder disjuncts and move the quantifier upon~$p$
inward (\ref{rule-ao-assoc}--\ref{rule-ao-idem}, \ref{rule-q-out-ao} right to
left) such that its argument is a conjunction of basic or negated basic
formulas in which $p$ occurs.  Let $F^{\prime\prime\prime}$ be the
intermediate result obtained so far.

In $F^{\prime\prime\prime}$, replace all subformulas of the form $\lnot
\ex^{\geq n} x\, \bigwedge_{1 \leq i \leq m} L_i[x]$ with the equivalent
formula $\all^{\cca{n}} x\, \bigvee_{1 \leq i \leq m} \du{L_i[x]}$. Replace
all subformulas of the form $\lnot pt$ where $t$ is a constant or an
individual variable that is free in $F$ with the equivalent formula
$\all^{\cca{1}} x\, (x\neq t \lor \lnot px)$ (\ref{rule-pullout-all},
\ref{prop-cqall-one}).  Replace all subformulas of the form $pt$ where $t$ is
a constant or an individual variable that is free in $F$ with the equivalent
formula $\all^{\cca{1}} x\, (x\neq t \lor px)$.  Reorder disjuncts and
conjuncts in the scope of $\ex p$ such that form~(\ref{form-noeq-sub-p}) in
the lemma statement is matched (\ref{rule-ao-assoc}--\ref{rule-ao-comm}). The
formula has now all properties asserted about $F^{\prime}$ by the lemma to
prove.  \qed
\end{proof}

\noindent
If the quantified predicate is nullary, the first steps of the method of
Lemma~\ref{lem-construction-eq} already yield a particularly simple form:
\begin{lem}[Normalizing Quantification Upon Nullary Predicates]\linebreak%
\label{lem-construction-eq-nullary}%
There is an effective
method to compute for a given nullary predicate $p$ and 
\MONE formula~$F$ formula $F^{\prime}$ such that
\begin{enumerate}
\item $F^{\prime}$ is a \QMONE formula,
\item $F^{\prime} \equiv \ex p F$,
\item $p$ is the only quantified predicate in $F^{\prime}$,
\item \label{form-eq-sub-p-nullary} 
All occurrences of $p$ in $F^{\prime}$ are in positive occurrences
of subformulas of
the forms $\ex p\, p$ and $\ex p\, \lnot p$,
\item All free individual variables, constants and predicates in
$F^{\prime}$ do occur in~$F$.
\end{enumerate}
\end{lem}

\begin{proof}
The intermediate result $F^{\prime\prime\prime}$ computed described in the
proof of Lemma~\ref{lem-construction-eq} has all properties asserted about
$F^{\prime}$ by the lemma to prove.
\qed
\end{proof}

\label{sec-two-counting-expansions}

\noindent
The following Lemma shows for quantification upon a unary predicate the
conversion of the \de{Generalized Eliminationshauptform} to a form that allows
to perform elimination directly by application of Lemma~\ref{lem-basic-elim}.
\begin{lem}[From Generalized Eliminationshauptform to 
the Basic Elimination Lemma]
\label{lem-hauptform-eq-elim}
Let $p$ be a unary predicate and let 
$F$ be a formula of the form
\[\begin{array}{llll}
\ex p
&   (  \bigwedge_{1 \leq i \leq a} & \all x^{\cca{a_i}}\, 
      (A_i[x] \lor px) & \land\\
&  \hparen   \bigwedge_{1 \leq i \leq b} & \all x^{\cca{b_i}}\,
       (B_i[x] \lor \lnot px) & \land\\
&  \hparen   \bigwedge_{1 \leq i \leq c} & \ex x^{\geq c_i}\, (C_i[x] \land px) & \land\\
&  \hparen   \bigwedge_{1 \leq i \leq d} & \ex x^{\geq d_i}\, (D_i[x] \land \lnot px)),
\end{array}\]
where $a$, $b$, $c$, $d$ are natural numbers $\geq 0$, for the referenced
values of $i$ the $a_i$, $b_i$, $c_i$, $d_i$ are natural numbers $\geq 1$, and
the $A_i[x], B_i[x], C_i[x], D_i[x]$ are first-order formulas in which $p$
does not occur.  Then $F$ is equivalent to
\[Q\, (G \land \ex p\,
 (\all x\, (A[x] \lor px) \land \all x\, (B[x] \lor \lnot px))),\]
where $Q$ is a quantifier prefix that existentially quantifies upon the
following individual variables which are fresh, that is, do not occur in $F$:
\[\begin{array}{ll}
 \va_{i1} \ldots \va_{i(a_i-1)}, & \text{ for } 1 \leq i \leq a\\
\vb_{i1} \ldots \vb_{i(b_i-1)}, & \text{ for } 1 \leq i \leq b,\\
\vc_{i1} \ldots \vc_{ic_i}, & \text{ for } 1 \leq i \leq c,\\
\vd_{i1} \ldots \vd_{ic_i}, & \text{ for } 1 \leq i \leq d,
\end{array}\]
\noindent
where $G$ is the formula
\[\begin{array}{ll}
 \bigwedge_{1 \leq i \leq c,\; 1 \leq j \leq c_i} 
   (C_i[\vc_{ij}] \land
       \bigwedge_{j < k \leq c_i}
       \vc_{ij} \neq \vc_{ik}) & \; \land\\
\bigwedge_{1 \leq i \leq d,\; 1 \leq j \leq d_i} 
   (D_i[\vd_{ij}] \land
       \bigwedge_{j < k \leq d_i}
       \vd_{ij} \neq \vd_{ik}),
\end{array}\]
with $C_i[\vc_{ij}]$ and $D_i[\vd_{ij}]$ denoting $C_i[x]$ and $D_i[x]$
after replacing all free occurrences of $x$ by $\vc_{ij}$ and $\vd_{ij}$,
respectively,

\smallskip

\noindent
where $A[x]$ is the formula
\[\textstyle\bigwedge_{1 \leq i \leq a} 
(A_i[x] \lor \textstyle\bigvee_{1 \leq j < a_i} x = \va_{ij})
\land
\textstyle\bigwedge_{1 \leq i \leq c,\; 1 \leq j \leq c_i} 
x \neq \vc_{ij},\]

\noindent
and where $B[x]$ is the formula
\[\textstyle\bigwedge_{1 \leq i \leq b}
(B_i[x] \lor \textstyle\bigvee_{1 \leq j < b_i} x = \vb_{ij})
\land
\textstyle\bigwedge_{1 \leq i \leq d,\; 1 \leq j \leq d_i} 
x \neq \vd_{ij}.\]

\end{lem}

\begin{proof}
Universal and existential counting quantifiers in $F$ can be expanded in
alternate ways (Prop.~\ref{prop-exp-eq-lt-lin} and \ref{prop-exp-eq-geq-poly})
such that in both cases existential variables are produced which can be moved
(by \ref{rule-q-out-ao}, \ref{rule-quant-flip}) in front of the existential
predicate quantifier. The particular expansions of universal counting
quantifiers applied there are:

\medskip

\noindent\hspace{1em}
$\begin{array}{rl@{\hspace{1em}}l}
 & \all x^{\cca{a_i}}\, (A_i[x] \lor px)\\
\equiv &
  \ex \va_{i1} \ldots \ex \va_{i(a_i-1)} \all x\,
   (A_i[x] \lor px
   \lor
   \bigvee_{1 \leq j < a_i} x = \va_{ij}) &
   \text{by Prop.~\ref{prop-exp-eq-lt-lin}}\\
\equiv &
  \ex \va_{i1} \ldots \ex \va_{i(a_i-1)} \all x\,
   ((A_i[x] \lor \bigvee_{1 \leq j < a_i} x = \va_{ij}) \lor px),
   & \text{by \ref{rule-ao-assoc}, \ref{rule-ao-comm}}
\end{array}$

\medskip 
\noindent
and analogously

\medskip

\noindent\hspace{1em}
$\begin{array}{rll}
 & \all x^{\cca{b_i}}\, (B_i[x] \lor \lnot px)\\
\equiv &
  \ex \vb_{i1} \ldots \ex \vb_{i(b_i-1)} \all x\,
   ((B_i[x] \lor \bigvee_{1 \leq j < b_i} x = \vb_{ij}) \lor \lnot px).
\end{array}$

\medskip

\noindent
The involved expansions of existential counting quantifiers are:

\medskip

\noindent\hspace{1em}
$\begin{array}{rl@{\hspace{1em}}l}
& \ex^{\geq c} x\, (C_i[x] \land px)\\
 \equiv &
  \ex \vc_{i1} \ldots \ex \vc_{ic}\,
  \bigwedge_{1 \leq j \leq c}
  (C_i[\vc_{ij}] \land p\vc_{ij}
     \land
      \bigwedge_{j < k \leq c}
      \vc_{ij} \neq \vc_{ik}) & \text{by Prop.~\ref{prop-exp-eq-geq-poly}}\\
\equiv &  \ex \vc_{i1} \ldots \ex \vc_{ic}\,
  (\bigwedge_{1 \leq j \leq c} 
   (C_i[\vc_{ij}] \land
       \bigwedge_{j < k \leq c}
       \vc_{ij} \neq \vc_{ik})\; \land &  \text{by \ref{rule-ao-assoc}, 
                                                \ref{rule-ao-comm}}\\
 & \hspace{2em}
   \bigwedge_{1 \leq j \leq c} p\vc_{ij})\\
\equiv &  \ex \vc_{i1} \ldots \ex \vc_{ic}\,
  (\bigwedge_{1 \leq j \leq c} 
   (C_i[\vc_{ij}] \land
       \bigwedge_{j < k \leq c}
       \vc_{ij} \neq \vc_{ik})\; \land & \text{by \ref{rule-pullout-all}}\\
 & \hspace{2em}  
    \bigwedge_{1 \leq j \leq c} \all x\, (x \neq \vc_{ij} \lor px)),
\end{array}$

\medskip

\noindent
and analogously

\medskip 

\noindent\hspace{1em}
$\begin{array}{rl@{\hspace{1em}}l}
& \ex^{\geq d_i} x\, (D_i[x] \land \lnot px)\\
\equiv &  \ex v^{\prime}_{i1} \ldots \ex v^{\prime}_{id_i}\,
  (\bigwedge_{1 \leq j \leq d_i} 
   (D_i[v^{\prime}_{ij}] \land
       \bigwedge_{j < k \leq d_i}
       v^{\prime}_{ij} \neq v^{\prime}_{ik})\; \land\\
 & \hspace{2em} 
    \bigwedge_{1 \leq j \leq d_i} 
    \all x\, (x \neq v^{\prime}_{ij} \lor \lnot px)).
\end{array}$

\medskip
\noindent
The formulas $A[x]$ and $B[x]$ are obtained from the universally quantified
conjuncts in the shown expansions by merging universal quantifiers
(\ref{rule-all-out-and}) and factoring the disjuncts $px$ and $\lnot px$,
respectively (\ref{rule-dist-cnf} left to right).
\qed
\end{proof}

\noindent
We can now combine the lemmas of this section to the proof of
Theorem~\ref{thm-beh:22-final}.

\begin{thmref}{Predicate Elimination for \MONE}
{\ref{thm-beh:22-final}}
There is an effective method to compute from a given predicate~$p$ and
\MONE formula~$F$ a formula~$F^{\prime}$
such that
\begin{enumerate}
\item $F^{\prime}$ is a \MONE formula,
\item $F^{\prime} \equiv \ex p\, F$,
\item $p$ does not occur in $F^{\prime}$,
\item All free individual variables, constants and predicates in
$F^{\prime}$ do occur in~$F$.
\end{enumerate}
\end{thmref}

\begin{proof}
If $p$ is nullary, then apply the method according to
Lemma~\ref{lem-construction-eq-nullary} to $p$ and $F$, followed by replacing
all occurrences of $\ex p\, p$ and of $\ex p\, \lnot p$ with the
equivalent~$\true$.

In case $p$ is unary, apply the method according to
Lemma~\ref{lem-construction-eq} to $p$ and $F$, followed by replacing all
subformulas starting with $\ex p$ with the equivalent formulas according
to Lemma~\ref{lem-hauptform-eq-elim}. In the intermediate result, replace all
subformulas starting with $\ex p$ with the equivalent formulas obtained by
eliminating $\ex p$ according to Lemma~\ref{lem-basic-elim}.
\qed
\end{proof}

\section{Predicate Elimination for Formulas without Equality}
\label{sec-elim-noeq}

Theorem~\ref{thm-beh:22-final} and the material leading to that theorem
applies to formulas \emph{with equality}. In \cite{beh:22}, a simpler variant
that only applies to formulas \emph{without equality} is described first and
in more detail, along with a direct comparison to an earlier result by
Schröder and with a restructuring of the result as a postprocessing operation
after elimination.  In this section we adapt the essential steps of this
variant and discuss aspects that become apparent more clearly with that
simpler method.

As already mentioned in Sect.~\ref{sec-overview-elimination-method}, even if
the input formula is without equality, the steps involved in predicate
quantifier elimination can possibly introduce equality if the predicate to be
eliminated occurs with an argument that is not a universally quantified
variable.  Thus, if there are several predicates to be eliminated one by one,
an input formula without equality can only be safely assumed for elimination
of the first predicate. 

However, as can be seen from the proofs in this section, the equality literals
introduced for inputs without equality actually either have a constant or two
existential variables as arguments, implying that the simpler variant without
dedicated equality handling is sufficient for elimination in formulas $\ex
p_1 \ldots \ex p_n\, F$ where $F$ is a \MON formula. For this case, and
the dual elimination of a sequence of universal predicate quantifiers, Behmann
\cite{beh:22} sketches a generalization of the variant for formulas without
equality, which, however, involves a translation whose size is exponential in
the number of predicate quantifiers. The core idea for Behmann's
generalization is described below at the end of this section.

The following proposition shows a method to compute a normal form for
monadic first-order logic \emph{without equality}, which is like the counting
quantifier normal form of Theorem~\ref{thm-foqe-monadic-eq}, but with certain
further restrictions on the allowed basic formulas. In particular, it involves
no counting quantifiers.
\begin{prop}[Normal Form for \MON]
\label{prop-foqe-monadic-noeq}  
There is an effective method to compute 
for a given \MON formula~$F$
a formula $F^{\prime}$ such that
\begin{enumerate}
\item $F^{\prime}$ is a Boolean combination of basic
formulas of the form:
\begin{enumerate}[label=(\alph*),ref=\alph*]
\item \label{noeq-nullary} $p$, where $p$ is a nullary predicate,
\item \label{noeq-unary} $pt$,
   where $p$ is a unary predicate and $t$ is a constant or individual variable,
\item \label{noeq-exists} $\ex x\, \bigwedge_{1 \leq i \leq m} L_i[x]$,
  where the $L_i[x]$ are pairwise different and pairwise non-complementary
  positive or negative literals with a unary predicate applied to the
  individual variable~$x$,
\end{enumerate}
\item $F^{\prime} \equiv F$,
\item All free individual variables, constants and predicates in $F^{\prime}$ do
  occur in~$F$.
\end{enumerate}
\end{prop}

\begin{proof}
If the proof of Theorem~\ref{thm-foqe-monadic-eq} is instantiated with a
formula $F$ without equality, the resulting formula is a Boolean combination
of basic formulas of the forms~(\ref{noeq-nullary}), (\ref{noeq-unary}) as
shown above and of the form (\ref{basic-exists}) from
Theorem~\ref{thm-foqe-monadic-eq} with $n$ instantiated to~$1$.  The latter
formulas can be converted to the form (\ref{noeq-exists}) shown above by
replacing the counting quantifier $\ex^{\geq 1} x$ with the equivalent
standard quantifier $\ex x$ (Prop.~\ref{prop-cqex-one}).  \qed
\end{proof}

\noindent
The remaining two propositions and lemmas in this section give, based on the
normal form of Prop.~\ref{prop-foqe-monadic-noeq}, the ingredients for a
method to eliminate predicate quantification on \MON formulas, in analogy to
Lemma~\ref{lem-construction-eq} and \ref{lem-construction-eq-nullary}.  The
construction of the \de{Eliminationshauptform} from arbitrary such formulas is
ensured with the following lemma.

\begin{lem}[Constructing the Eliminationshauptform for \MON]
\label{lem-construction-noeq}
There is an effective
method to compute for a given unary predicate $p$ and \MON formula~$F$ 
a formula $F^{\prime}$ such that
\begin{enumerate}
\item $F^{\prime}$ is a \MONE formula,
\item $F^{\prime} \equiv \ex p F$,
\item $p$ is the only quantified predicate in $F^{\prime}$,
\item \label{form-noeq-sub-p} 
All occurrences of $p$ in $F^{\prime}$ are in positive occurrences of subformulas of
the form
\[\begin{array}{rll}
\ex p
&  (\bigwedge_{1 \leq i \leq a} \all x\, (A_i[x] \lor px) & \land\\
&  \hparen \bigwedge_{1 \leq i \leq b} \all x\, (B_i[x] \lor \lnot px) & \land\\
&  \hparen \bigwedge_{1 \leq i \leq c} \ex x\, (C_i[x] \land px) & \land\\
&  \hparen \bigwedge_{1 \leq i \leq d} \ex x\, (D_i[x] \land \lnot px)), &
\end{array}\]
where $a, b, c, d$ are natural numbers \geqzero and the
$A_i[x]$, $B_i[x]$, $C_i[x]$, $D_i[x]$ are formulas in which $p$
does not occur,
\item All free individual variables, constants and predicates in
$F^{\prime}$ do occur in~$F$.
\end{enumerate}
\end{lem}

\begin{proof}
Compute the intermediate $F^{\prime\prime\prime}$ as described in the
proof of Lemma~\ref{lem-construction-eq}. Continue in analogy to that proof:
In $F^{\prime\prime\prime}$, replace all subformulas of the form $\lnot
\ex x\, \bigwedge_{1 \leq i \leq m} L_i[x]$ with the equivalent formula $\all
x\, \bigvee_{1 \leq i \leq m} \du{L_i[x]}$.  Replace all subformulas of the
form $\lnot pt$ where $t$ is a constant or an individual variable that is
free in $F$ with the equivalent formula $\all x\, (x\neq t \lor \lnot px)$. Replace
all subformulas of the form $pt$ where $t$ is a constant or an individual
variable that is free in $F$ with the equivalent formula $\all x\, (x\neq t \lor
px)$. Reorder disjuncts and conjuncts in the scope of $\ex p$ such that
form~(\ref{form-noeq-sub-p}) in the lemma statement is matched.  The formula
has now all properties asserted about $F^{\prime}$ by the lemma to prove. \qed
\end{proof}

\noindent
The method described in the proof of
Lemma~\ref{lem-construction-noeq} introduces equality atoms in case
there are occurrences of the quantified predicate where the argument is a
constant or an individual variable that is free in the original input formula.

The following Lemma renders the elimination result as given in \cite{beh:22}.
It combines conversion from \de{Eliminationshauptform} to a form that allows
direct elimination with Lemma~\ref{lem-basic-elim}, in analogy to
Lemma~\ref{lem-hauptform-eq-elim}, with performing the elimination and
restructuring the result.

\begin{lem}[Elimination on the Eliminationshauptform]
\label{lem-hauptform-noeq-elim}
Let $F$ be a formula of the following form:
\[\begin{array}{rlll}
\ex p
&     (\bigwedge_{1 \leq i \leq a} & \all x\, (A_i[x] \lor px) & \land\\
&     \hparen \bigwedge_{1 \leq i \leq b} & \all x\, (B_i[x] \lor \lnot px) & \land\\
&     \hparen \bigwedge_{1 \leq i \leq c} & \ex x\, (C_i[x] \land px) & \land\\
&     \hparen \bigwedge_{1 \leq i \leq d} & \ex x\, (D_i[x] \land \lnot px)),
\end{array}\]
where $p$ is a unary predicate, $a$, $b$, $c$, $d$ are numbers \geqzero
and the $A_i[x]$, $B_i[x]$, $C_i[x]$, $D_i[x]$ are first-order
formulas in which $p$ does not occur.  Then $F$ is equivalent to
the following formula \MONE formula $F^{\prime}$:
\[\begin{array}{l@{\hspace{1ex}}l}
\all x\, (\bigwedge_{1 \leq i \leq a} A_i[x]\;\; \lor\;\;
             \bigwedge_{1 \leq i \leq b} B_i[x]) & \land\\[3pt]
\ex u_1 \ldots \ex u_c
\ex v_1 \ldots \ex v_d\\
\hspace{1em}(\bigwedge_{\substack{1 \leq i \leq c,\; 1 \leq j \leq d}} 
 u_i \neq v_j & \land\\
\hspace{1em}\hparen \bigwedge_{1 \leq i \leq c} 
  (C_i[u_i] \land 
   \bigwedge_{1 \leq j \leq b} B_j[u_i]) & \land\\
\hspace{1em}\hparen \bigwedge_{1 \leq i \leq d}
   (D_i[v_i] \land 
   \bigwedge_{1 \leq j \leq a} A_j[v_i])),
\end{array}\]
where $u_1, \ldots, u_c$ and $v_1, \ldots, v_d$ are distinct individual
variables that are fresh, that is, do not occur free in $F$, and if $t$ is one
of these variables, then $A_i[t]$, $B_i[t]$, $C_i[t]$, $D_i[t]$ denote
$A_i[x]$, $B_i[x]$, $C_i[x]$, $D_i[x]$, respectively, with all free
occurrences of $x$ replaced by~$t$.
\end{lem}

\begin{proof}
The proof shows equivalence preserving transformations that lead from $F$ to
$F^{\prime}$: First, $F$ is converted to a form that matches the left side of
Lemma~\ref{lem-basic-elim}. That lemma is then applied to eliminate the
existentially quantified $p$. Finally, the formula that results from applying
the lemma is postprocessed to yield $F^{\prime}$.

We define two shorthands $A[x] = \bigwedge_{1 \leq i \leq a} A_i[x]$ and $B[x]
= \bigwedge_{1 \leq i \leq b} B_i[x]$. Then $F$ is equivalent to
\begin{displaymath}
\begin{array}{rlll}
\ex p
&     \multicolumn{2}{l}{(\all x\, (A[x] \lor px)} &\; \land\\
&     \multicolumn{2}{l}{\hparen \all x\, (B[x] \lor \lnot px)} &\; \land\\
&     \hparen \bigwedge_{1 \leq i \leq c} & (\ex x\, C_i[x] \land px) &\; \land\\
&     \hparen \bigwedge_{1 \leq i \leq d} & (\ex x\, D_i[x] \land \lnot px)).
\end{array}
\end{displaymath}
By renaming the existential individual variables (\ref{rule-var-rename}),
moving the existential quantifiers outward (\ref{rule-q-out-ao}) and
reordering them (\ref{rule-quant-flip}) we obtain:
\begin{displaymath}
\begin{array}{rlll}
\ex u_1 \ldots \ex u_c
\ex v_1 \ldots \ex v_d
\ex p
&     \multicolumn{2}{l}{(\all x\, (A[x] \lor px)} &\; \land\\
&     \multicolumn{2}{l}{\hparen \all x\, (B[x] \lor \lnot px)} &\; \land\\
&     \hparen \bigwedge_{1 \leq i \leq c} & (C_i[u_i] \land pu_i) &\; \land\\
&     \hparen \bigwedge_{1 \leq i \leq d} & (D_i[v_i] \land \lnot pv_i)).
\end{array}
\end{displaymath}
Reordering conjuncts (\ref{rule-ao-comm}, \ref{rule-ao-assoc})
and moving the quantifier upon~$p$ inward (\ref{rule-q-out-ao} from
right to left) yields:
\begin{displaymath}
\label{cd-to-front}
\begin{array}{rlll}
\ex u_1 \ldots \ex u_c
\ex v_1 \ldots \ex v_d
&
\multicolumn{2}{l}{(\bigwedge_{1 \leq i \leq c} C_i[u_i] \land
\bigwedge_{1 \leq i \leq d} D_i[v_i]} &\; \land\\[4pt]
&\hparen \ex p
&    (\all x\, (A[x] \lor px) &\; \land\\
&& \hparen \all x\, (B[x] \lor \lnot px) &\; \land\\
&& \hparen \bigwedge_{1 \leq i \leq c} pu_i &\; \land\\
&& \hparen \bigwedge_{1 \leq i \leq d} \lnot pv_i)).
\end{array}
\end{displaymath}
Next ``pull out'' the existentially quantified arguments
from $p$ (\ref{rule-pullout-all}):
\begin{displaymath}
\begin{array}{rlll}
\ex u_1 \ldots \ex u_c
\ex v_1 \ldots \ex v_d
&
\multicolumn{2}{l}{(\bigwedge_{1 \leq i \leq c} C_i[u_i] \land
\bigwedge_{1 \leq i \leq d} D_i[v_i]} &\; \land\\[4pt]
&\hparen \ex p
&    (\all x\, (A[x] \lor px) &\; \land\\
&& \hparen \all x\, (B[x] \lor \lnot px) &\; \land\\
&& \hparen \bigwedge_{1 \leq i \leq c} (\all x\, x \neq u_i \lor px) &\; \land\\
&& \hparen \bigwedge_{1 \leq i \leq d} (\all x\, x \neq v_i \lor \lnot px))).
\end{array}
\end{displaymath}
By merging universal quantifiers (\ref{rule-all-out-and}) and
factoring the occurrences of $px$ and $\lnot px$
(\ref{rule-dist-cnf} from right to left) we get:
\begin{displaymath}
\begin{array}{rlll}
\ex u_1 \ldots \ex u_c
\ex v_1 \ldots \ex v_d
&
\multicolumn{2}{l}{(\bigwedge_{1 \leq i \leq c} C_i[u_i] \land
\bigwedge_{1 \leq i \leq d} D_i[v_i]} &\; \land\\
& \hparen \ex p
&     (\all x\, ((A[x] \land \bigwedge_{1 \leq i \leq c}\, x \neq u_i)
                  \lor px) &\; \land\\
&& \hparen    \all x\, ((B[x] \land \bigwedge_{1 \leq i \leq d}\, x \neq v_i)
                  \lor \lnot px))).
\end{array}
\end{displaymath}
The subformula starting with $\ex p$ matches the left side of
Lemma~\ref{lem-basic-elim}, which justifies to rewrite it by a first-order
formula that does no longer contain $p$, resulting in:
\begin{displaymath}
\begin{array}{rlll}
\ex u_1 \ldots \ex u_c
\ex v_1 \ldots \ex v_d
&
\multicolumn{2}{l}{(\bigwedge_{1 \leq i \leq c} C_i[u_i] \land
\bigwedge_{1 \leq i \leq d} D_i[v_i]} &\; \land\\
& \hparen \all x\, & ((A[x] \land \bigwedge_{1 \leq i \leq c}\, x \neq u_i)\; \lor\\
&& \hparen (B[x] \land \bigwedge_{1 \leq i \leq d}\, x \neq v_i))).
\end{array}
\end{displaymath}
By distributing disjunction over conjunction (\ref{rule-dist-cnf}),
distributing universal quantification into conjunction (\ref{rule-all-out-and}
from right to left), ``pulling in'' argument terms (\ref{rule-pullout-all}
from right to left) and reordering conjuncts (\ref{rule-ao-assoc},
\ref{rule-ao-comm}) we obtain:
\begin{displaymath}
\begin{array}{llll}
\all x\, (A[x] \lor B[x]) &&\; \land\\
\ex u_1 \ldots \ex u_c
\ex v_1 \ldots \ex v_d
& (\bigwedge_{1 \leq i \leq c,\; 1 \leq j \leq d} u_i \neq v_j &\; \land\\
& \hparen \bigwedge_{1 \leq i \leq c} (C_i[u_i] \land B[u_i]) &\; \land\\
& \hparen \bigwedge_{1 \leq i \leq d} (D_i[v_i] \land A[v_i])),\\
\end{array}
\end{displaymath}
where, in analogy to the previously used notation, $A[v_i]$ and $B[u_i]$
denotes $A[x]$ and $B[x]$ with all free occurrences of $x$ replaced by $v_i$
or $u_i$, respectively.
Expanding the shorthands $A[x]$, $B[x]$, $A[v_i]$, $B[u_i]$ then
yields the formula $F^{\prime}$ from the proposition statement.
\qed
\end{proof}

\noindent
The method described in the proof of Lemma~\ref{lem-hauptform-noeq-elim}
introduces equality to handle occurrences of the quantified predicate with an
existentially quantified argument.  With the method described in
Lemma~\ref{lem-construction-noeq} this is the second place where equality is
introduced in the overall method to eliminate predicates from formulas without
equality.

Behmann remarks \cite[p.~201f]{beh:22} that the first division of the second
volume of Schröder's \de{Algebra der Logik} \cite[p.~400 2]{schroeder:2.1}
concludes with recommending the problem solved with
Lemma~\ref{lem-hauptform-noeq-elim} as an open issue for future research.
That Schröder failed to solve the problem is attributed by Behmann to the
insufficient means of representation in the \de{Algebra der Logik}.
In the proof of Lemma~\ref{lem-hauptform-noeq-elim} after the elimination step a
further equivalence preserving transformation is applied, where occurrences of
$A[x]$ and $B[x]$ are instantiated with existentially quantified variables.
(To the proof of Theorem~\ref{thm-beh:22-final} given above, the analogous
step could be added.) This step facilitates to compare the result of
Lemma~\ref{lem-hauptform-noeq-elim} with an earlier incomplete result by
Schröder: Behmann \cite[p.~203f]{beh:22} remarks that the following formula,
related to the last displayed formula in the proof of
Lemma~\ref{lem-hauptform-noeq-elim}, is what Schröder calls \name{,,crude
  resultant''} (\de{\glqq Resultante aus dem Rohen\grqq})
\cite[\S41]{schroeder:2.1} (see also \cite{craig:2008} for a modern discussion
of Schröder's work on elimination):
\begin{equation}
\begin{array}{lll}
\multicolumn{2}{l}{\all x\, (A[x] \lor B[x])} & \land\\
\bigwedge_{1 \leq i \leq c} & \ex x\, (C_i[x] \land B[x]) & \land\\
\bigwedge_{1 \leq i \leq d} & \ex x\, (D_i[x] \land A[x]).
\end{array}
\end{equation}
This formula is obtained from the last displayed formula in the proof by
dropping the conjunct $\textstyle\bigwedge_{1 \leq i \leq c,\; 1 \leq j \leq d}
u_i \neq v_j,$ which leads to a weaker formula, followed by the equivalence
preserving inward propagation of existential quantifiers and renaming of
variables.

As mentioned in the beginning of this section, in \cite[\S~17]{beh:22} a
generalization of the techniques for formulas without equality to the
simultaneous elimination of a sequence of predicate quantifiers that are all
existential or all universal is outlined. It can be based on the following
adaption of Lemma~\ref{lem-basic-elim} to the simultaneous elimination of $n$
existentially quantified unary predicates from a conjunction of $2^n$
disjunctions:
\begin{equation}
\ex p_1 \ldots \ex p_n\!\!\!\!\!\!
\bigwedge_{S \subseteq \{1, \ldots, n\}}\!\!\!\!\!\!
\all x\, (F_S \lor \bigvee_{i \in S} p_ix \lor
           \!\!\!\!\!\!\!\! \bigvee_{i \in \{1, \ldots, n\} - S} 
           \!\!\!\!\!\!\!\! \lnot p_ix)
\;\;\equiv\;\;
\all x\, \!\!\!\!\!\!\bigvee_{S \subseteq \{1, \ldots, n\}} 
\!\!\!\!\!\! F_S,
\end{equation}
which holds if the predicates $p_1, \ldots, p_n$ do not occur in the formulas
$F_S$ (subsets $S$ of $\{1, \ldots, n\}$ are used in $F_S$ as index
subscript).

\part{Further Issues Addressed in and Related to Behmann's Habilitation Thesis}
\label{part-further}

\section{Introduction to Part~\ref{part-further}}

In this part various issues that are addressed by Behmann in \cite{beh:22}
aside of the main results as well as specific connections with related
techniques are discussed. For a further discussion of Behmann's elimination
technique in relation to modern second-order quantifier elimination methods
following the direct or Ackermann approach we refer to \cite{cw-relmon}.

\section{Innex and Related Forms}
\label{sec-innex}
\label{sec-qbf-inner}

Recall that the key technique of \cite{beh:22} is propagating quantifiers
inward, also for the price of expensive operations such as distribution of
conjunction over disjunction and vice versa.\footnotemark\ This inward
propagation is applied to quantifiers upon individual variables
(Theorem~\ref{thm-foqe-monadic-eq}, Prop.~\ref{prop-foqe-monadic-noeq}) as
well as to quantifiers upon predicates (Lemma~\ref{lem-construction-eq},
\ref{lem-construction-eq-nullary} \ref{lem-construction-noeq}).

\footnotetext{On formulas that are not in negation normal form, distribution
  of conjunction over disjunction alone or distribution of disjunction over
  conjunction alone is sufficient, because either one can then express the
  other, consider for example: $F \lor (G \land H) \equiv \lnot (\lnot F \land
  (\lnot G \lor \lnot H)) \equiv \lnot ((\lnot F \land \lnot G) \lor (\lnot F
  \land \lnot H)) \equiv (F \lor G) \land (F \lor H)$.}

As remarked in \cite[p.~193]{beh:22}, just propagating Boolean quantifiers
inward in this manner already yields a decision method for quantified Boolean
formulas that have no free predicates (that is, in other terminology, no free
Boolean variables), and hence also for propositional logic. This can be easily
seen from Corollary~\ref{corr-elim}: If the method asserted in the corollary
is applied to a quantified Boolean formula, by the underlying lemmas, in
particular Lemma~\ref{lem-construction-eq-nullary}, inward propagation of
$\ex p$ yields a formula in which all occurrences of Boolean
quantification are in formulas of the form $\ex p\, p$ and $\ex p\,
\lnot p$.

Given the decidability of the Bernays-Schönfinkel class
\cite{bernays:schoenfinkel:28} (published six years after \cite{beh:22}),
monadic first-order logic can be decided essentially just by propagating all
individual quantifiers inward in the manner indicated above with possibly
expensive distribution operations. From the result formula, a prenex formula
in the Bernays-Schönfinkel class can then be obtained by first propagating
existential quantifiers outward and then the universal quantifiers.  This
conversion has been sketched in \cite[p.~36]{dreben:goldfarb}.

Techniques to propagate quantifiers upon individual variables inward are known
for a long time in automated theorem proving, under the names
\name{mini\-scope form} and \name{antiprenexing}. They are typically
considered as preprocessing operations
\cite{loveland:1978,nonnengart:weidenbach:handbook}, without taking
potentially expensive steps like distribution of connectives into
account. Exceptions are early works by \cite{wang:1960} -- according to
\cite[footnote~2]{ernst:1971}\footnote{The definition of \name{miniscope form}
  in \cite{wang:1960}, however, does not imply this and would accept, for
  example, $\ex x\, (px \land (qx \lor rx))$.}  -- and
\cite{bibel:1974}, where the possibility to take advantage of distribution is
mentioned in the conclusion.  The commonly considered preprocessing methods
might terminate inward propagation with a formula where a potentially
expensive step would allow further inward propagation. As a consequence, these
methods are in general not sufficient to decide \MON.

Two variants to propagate individual quantifiers inward have been discussed in
the literature. To describe them, we assume quantifiers $\ex$ and
$\all$ as they appear in a formula in negation normal form, obtained by
propagating negation inward with (\ref{rule-not-not}--\ref{rule-not-ex}) such
that negated existential quantification is expressed by universal
quantification.
In the first variant, often called \name{miniscope form}, existential as well
as universal quantifiers are both propagated inward
\cite{wang:1960,bibel:1974,loveland:1978}.  The main motivation is there to
reduce the number of arguments of Skolem functions by reducing the number of
universal variables that have a given existential quantifier in their scope.
In the second variant, termed \name{antiprenexing} in
\cite{egly:antiprenex},\footnote{\name{Antiprenex form} has been used in
  \cite{bibel:1974} for the first variant.}  both types of quantifiers are
handled differently: only universal quantifiers are propagated inward,
existential quantifiers are propagated outward (\ref{rule-ex-out-or},
\ref{rule-q-out-ao}). Aside of reducing the number of arguments of Skolem
functions, the motivation is there to reduce also the number of Skolem
functions: Applying \ref{rule-ex-out-or} from right to left would effect
duplication of existential quantifiers, each requiring a different Skolem
function.

As observed in \cite{beh:22}, the elimination technique for monadic formulas
of Lemma~\ref{lem-hauptform-eq-elim} and \ref{lem-hauptform-noeq-elim}
involves propagating predicate quantifiers and universal individual
quantifiers \emph{inward}, while propagating existential individual
quantifiers \emph{outward}.
A systematic investigation of quantifier propagation schemes with respect to
elimination of predicate quantifiers seems still an open issue.

\section{Quine's Expansion}

In \cite{quine:decidability} Quine presents in a decision method for \MON that
is, as remarked in \cite[pp.~253,~293]{church:book}, a variant of Behmann's
method. Like the latter, it is based on producing the normal form of
Prop.~\ref{prop-foqe-monadic-noeq} by propagating quantifiers upon instance
variables inward.  However, instead of performing distribution of conjunction
over disjunction to enable inward propagation of quantifiers, rewriting with
the following equivalence Prop.~\ref{quines-expansion-original} is applied.
This equivalence is shown first in a dual version as
Prop.~\ref{quines-expansion-dual} that corresponds to the setting of
Prop.~\ref{prop-foqe-monadic-noeq}.
\begin{prop}[Quine's Expansion]
\label{prop-quine-expansion}
Let $F[G]$ is a first-order formula with occurrences of a subformula $G$ in
which $x$ does not occur free and whose free variables are not in scope of a
quantifier within $F[G]$. Formulas $F[\true]$ and $F[\false]$ denote $F[G]$
with all the occurrences of $G$ replaced by $\true$ or $\false$, respectively.
Then

\medskip

\slab{quines-expansion-dual}
$\ex x F[G] \; \equiv \;
  (G \lor \ex x\, F[\false]) \land
  (\lnot G \lor \ex x\, F[\true])$.

\slab{quines-expansion-original}
$\all x F[G] \; \equiv \;
  (G \land \all x\, F[\true]) \lor
  (\lnot G \land \all x\, F[\false]).$
\end{prop}

\noindent
Obviously, applying the expansions according to
Prop.~\ref{prop-quine-expansion} could be immediately followed by truth-value
simplification (\ref{rule-tv-not-t}--\ref{rule-tv-q-f}), which is actually
assumed in the original presentation by Quine.  The version of
Prop.~\ref{quines-expansion-original} given in \cite{quine:decidability}, is
actually a generalization of the well-known propositional Shannon expansion.
The method of \cite{quine:decidability} proceeds by exhaustively rewriting
innermost subformulas that match the precondition of the expansion.

\section{Other Methods for Deciding Relational Monadic Formulas}

Alternative decision methods for \MON formulas include resolution: Equipped
with an appropriate ordering and condensation, it decides \MON formulas,
although the associated Herbrand universe might be infinite due to
Skolemization \cite{resolution:decision:2001}.  A superposition-based decision
method for \MONE is given in \cite{bachmair:1993}. 

Deciding satisfiability is for \MON and for \MONE \textsf{NEXPTIME}-complete,
as presented in \cite[Sect.~6.2]{classical:decision} along with more
fine-grained results.  Upper bounds makes use of the fact that a satisfiable
\MONE (\MON, resp.) formula has a model whose cardinality is a most $q2^m$
($2^m$, resp.), where $q$ is the quantifier rank and $m$ is the number of
predicates.  The underlying method of \cite{lewis:80} for deciding \MON
verifies a given interpretation by repeatedly constructing an innex form with
respect to a single innermost quantifier occurrence and then replacing the
corresponding obtained quantified subformulas with $\true$ or $\false$
according to the interpretation.  The processing of an existential (universal,
resp.) innermost quantifier proceeds by conversion of its argument formula to
disjunctive (conjunctive, resp.) normal form, followed by distributing the
quantifier over the conjunctive clauses (clauses, resp.) and then, in each
conjunctive clause (clause, resp.), narrowing the quantifier scope to those
literals that contain the quantified variable.  Determination of the upper
bound makes use of the fact that all atoms occurring in the intermediate
formulas are already present in the input formula.

\section{Normal Form with Respect to a Predicate}
\label{sec-nf-predicate}

Behmann notes that for practical application it is not necessary to construct
the \de{Eliminationshauptform} ((\ref{eq-hauptform-informal}) on
p.~\pageref{eq-hauptform-informal}), via the \emph{fully developed}
disjunctive normal form (as done in the methods described in the proofs of
Theorem~\ref{thm-foqe-monadic-eq} and Lemma~\ref{lem-construction-eq} and
their correspondents for the case without equality), but that it 
suffices if just
the predicate to eliminate is separated from other predicates with respect to
the associated variables \cite[p.~201]{beh:22}.  Although Behmann remarks that
the corresponding forms, developed just with respect to a given predicate,
observe a number of strange laws and are of not much less theoretical and
practical importance than normal forms, he does not pursue this issue in
\cite{beh:22},\footnote{So far, there has been no publication or manuscript by
  Behmann identified where he would present the indicated material.} but just
gives an example:
\begin{equation}
\ex p\, \ex x\, ((qx \lor rx) \land px)
\end{equation}
immediately matches the \de{Eliminationshauptform} with $a = b = d = 0$, $c =
1$ and $C_1[x] = qx \lor rx$ such there is no point in distributing the
conjunction with $px$ over $qx \lor rx$.

There seem at least superficial relationships with predicate elimination
techniques that involve concepts like \name{standardized} and \name{in good
  scope} which are parameterized with a predicate to eliminate
\cite{dls:conradie} or the conversion to formulas where conjuncts are
\name{linkless outside a set of predicates}, a property that permits to
distribute existential quantification upon predicates not in the set
(``outside'' the set) over conjunction \cite{cw-tableaux}.

\section{An Application of Elimination: Modeling Syllogistic Reasoning}

As noted in \cite{craig:2008}, the view of syllogisms as instances of
elimination problems where the conclusion is the result of eliminating the
middle term from the conjunction of the premises is due to Boole
\cite{boole:laws}. Also Schröder pursues this representation of syllogisms
\cite[\S~42--44]{schroeder:2.1}.  If the three terms involved in syllogisms
are expressed as unary predicates, the premises and conclusions correspond to
sentences of first-order logic in the
relational monadic fragment. Behmann
\cite[\S~17]{beh:22} exemplarily demonstrates his second-order quantifier
elimination method with modeling some syllogisms and related statements by
Schröder.  Here, these examples are shown with some intermediate steps that
indicate the involved rewritings.\footnote{In addition, in \cite{beh:22} a
  ``singular'' reading of syllogisms is discussed, where the middle term is
  understood as individual instead of a predicate and elimination is not
  involved in modeling.}

\medskip

\noindent
Syllogism \name{Ferio}.
\begin{equation}
\begin{array}{lll}
       & \ex q\, (\all x\, (\lnot qx \lor \lnot px)
                     \land
                     \ex x\, (rx \land qx))\\
\equiv & \ex u\, (ru \land \ex q\,
                     (\all x\, (\lnot qx \lor \lnot px)
                     \land
                      \all x\, (x \neq u \lor qx)))\\
\equiv & \ex u\, (ru \land \lnot pu).
\end{array}
\end{equation}

\noindent
Syllogism \name{Darapti}. Here the implicitly understood non-emptiness of $q$
is added as a third auxiliary premise~$\ex x\, qx$.
\begin{equation}
\begin{array}{lll}
       & \ex q\, (\all x\, (\lnot qx \lor px)
                    \land
                    \all x\, (\lnot qx \lor rx)
                    \land 
                    \ex x\, qx)\\
\equiv & \ex u \ex q\, (\all x\, (\lnot qx \lor px)
                    \land
                         \all x\, (\lnot qx \lor rx)
                    \land 
                       \all x\, (x \neq u \lor qx))\\
\equiv & \ex u\, pu \land ru.
\end{array}
\end{equation}

\noindent
Behmann shows the following example from \cite[p.~361]{schroeder:2.1}, with the
premises ``Some~$p$ are not~$q$. Some~$q$ are not~$r$.''  Since the first
conjuncts of both premises are particular and negative no conclusions would
be expected in the sense of traditional syllogistic reasoning.  Elimination,
however, allows to conclude that there exist two distinct individuals, one
in~$p$ and the other not in~$r$:
\begin{equation}
\begin{array}{lll}
& \ex q\, (\ex x\, (px \land \lnot qx) \land
              \ex x\, (qx \land \lnot rx))\\
\equiv & \ex u \ex v\, (pu \land \lnot rv \land
\ex q\, (\all x\, (x \neq u \lor \lnot qx) \land
              \all (x\, qx \lor x \neq v)))\\
\equiv & \ex u \ex v\, (pu \land \lnot rv \land u \neq v).
\end{array}
\end{equation}

\noindent
The following example given by Behmann is a ``composed syllogism'' from
Schrö\-der \cite[p.~283]{schroeder:2.1}, which involves elimination of two
predicates~$q_1$ and $q_2$.
\begin{equation}
\begin{array}{lll}
& \ex q_1 \ex q_2\,
       (\all x\, (\lnot q_1x \lor px)
\land  \all x\, (\lnot q_2x \lor rx)
\land  \ex x\, (q_1x \land q_2x))\\
\equiv & \ex q_1\,
       (\all x\, (\lnot q_1x \lor px)
\land \ex q_2\,
(\all x\, (\lnot q_2x \lor rx)
\land   \ex x\, (q_1x \land q_2x)))\\
\equiv & \ex q_1\,
       (\all x\, (\lnot q_1x \lor px) \land
\ex u\, (q_1u\; \land\\
& \hspace{2em} \ex q_2\,
(\all x\, (\lnot q_2x \lor rx)
\land  \all x\, (x \neq u \lor q_2x))))\\
\equiv & \ex q_1\,
      (\all x\, (\lnot q_1x \lor px)
\land \ex u\, (q_1u \land ru))\\
\equiv & \ex u\, (ru \land 
      \ex q_1\,
       (\all x\, (\lnot q_1x \lor px)
\land \all x\, (x \neq u \lor q_1x)))\\
\equiv & \ex u\, (ru \land pu).\\
\end{array}
\end{equation}

\sectionmark{Polyadic Formulas Allowing
 Monadic Elimination Techniques}
\section{Polyadic Formulas Allowing
 Monadic Elimination Techniques and Handling of
 Auxiliary Definitions}
\label{sec-essentially-monadic}

The elimination properties of monadic logic also apply in certain cases where
involved predicates have more than one argument, for example, if arguments
except of one are instantiated with a constant or with a variable that is free
in the scope of the respective predicate quantifiers.  Behmann gives in
\cite[\S~19]{beh:22} several examples for applying his elimination
method to such
formulas. His first example is from \cite[p.~491]{schroeder:3}:
\begin{equation}
\begin{array}{rl}
& \ex p\,
  (\all z\, (fxz \lor pzy \lor hzy)
  \land
  \all z\, (gxz \lor \lnot pzy \lor hzy))\\
\equiv & \ex p^{\prime}\,
  (\all z\, (f^{\prime}z \lor p^{\prime}z \lor h^{\prime}z)
  \land
  \all z\, (g^{\prime}z \lor \lnot p^{\prime}z \lor h^{\prime}z))\\
\equiv &
  \all z\, (f^{\prime}z \lor g^{\prime}z \lor h^{\prime}z)\\
\equiv &  
  \all z\, (fxz \lor gxz \lor hzy),
\end{array}
\end{equation}
where the following shorthands are used to hide the free variables from view:
$f^{\prime}z = fxz$, $g^{\prime}z = gxz$, $h^{\prime}z = hzy$,
$p^{\prime}z = pzy$.
The next example is from \cite[p.~308]{schroeder:3}:
\begin{equation}
\begin{array}{rl}
& \all p\, (\ex u\, (pxu \land fuy) 
        \lor \all v\, (\lnot pxv \lor gvy))\\
\equiv & \lnot \ex p^{\prime}\,  
              (\all u\, (\lnot p^{\prime}u \lor \lnot f^{\prime}u) 
        \land \ex v\, (p^{\prime}v \land \lnot g^{\prime}v))\\
\equiv & \lnot \ex v\, (\lnot g^{\prime}v \land \ex p^{\prime}\,  
              (\all u\, (\lnot p^{\prime}u \lor \lnot f^{\prime}u) 
        \land \all u\, (u \neq v \lor p^{\prime}u)))\\
\equiv & \lnot \ex v\, (\lnot g^{\prime}v \land \lnot f^{\prime}v)\\
\equiv & \all v\, (gvy \lor fvy),
\end{array}
\end{equation}
where the following shorthands are used (for arbitrary variables~$z$):
$f^{\prime}z = fzy$, $g^{\prime}z = gzy$, $p^{\prime}z = pxz$.  The
following formula from \cite[p.~510]{schroeder:3} is shown as an example where
elimination with the method of \cite{beh:22} fails, since a proper relation
between $u$ and $v$ is involved.
\[\all p\, (pxy \lor
  \ex v\, (\all u\, (\lnot pxu \lor fuv) \land gvy)).\]
The final example given in \cite[\S~19]{beh:22} is from
\cite[p.~545]{schroeder:3} and involves two instances of the binary
predicate~$p$ to eliminate where the bound variable occurs in different
argument positions:
\begin{equation}
\begin{array}{rl@{\hspace{1.0em}}r}
& \ex p\, (\all z\, (fxz \lor pzy) \land
               \all z\, (\lnot pxz \lor gzy))\\
\equiv & \ex p^{\prime} \ex q^{\prime}\, 
             (\all z\, (f^{\prime}z \lor p^{\prime}z) \land
             \all z\, (\lnot q^{\prime}z \lor g^{\prime}z) \land
            ((p^{\prime}x \land q^{\prime}y) \lor
                  (\lnot p^{\prime}x \land \lnot q^{\prime}y))) & *\\
\equiv &
  \ex p^{\prime} \ex q^{\prime}\, 
           (\all z\, (f^{\prime}z \lor p^{\prime}z) \land
           \all z\, (\lnot q^{\prime}z \lor g^{\prime}z) \land
        p^{\prime}x \land q^{\prime}y) \; \lor \\
        & \ex p^{\prime} \ex q^{\prime}\, 
          (\all z\, (f^{\prime}z \lor p^{\prime}z) \land
           \all z\, (\lnot q^{\prime}z \lor g^{\prime}z) \land
          \lnot p^{\prime}x \land \lnot q^{\prime}y)\\
\equiv & (\ex p^{\prime}\, (\all z\, (f^{\prime}z \lor p^{\prime}z) 
           \land  p^{\prime}x) \land
          \ex q^{\prime}\, (\all z\, (\lnot q^{\prime}z \lor g^{\prime}z) 
           \land q^{\prime}y))\; \lor\\
       & (\ex p^{\prime}\, (\all z\, (f^{\prime}z \lor p^{\prime}z) 
           \land  \lnot p^{\prime}x) \land
          \ex q^{\prime}\, (\all z\, (\lnot q^{\prime}z \lor g^{\prime}z)
          \land \lnot q^{\prime}y))\\
\equiv & \ex q^{\prime}\, (\all z\, (\lnot q^{\prime}z \lor g^{\prime}z) 
         \land q^{\prime}y)\; \lor & **\\
       & \ex p^{\prime}\, (\all z\, (f^{\prime}z \lor p^{\prime}z) 
         \land  \lnot p^{\prime}x)\\
\equiv & \ex q^{\prime} (\all z\, (\lnot q^{\prime}z \lor g^{\prime}z) \land
         \all z\, (z \neq y \lor q^{\prime}z))\; \lor\\
       & \ex p^{\prime} (\all z\, (f^{\prime}z \lor p^{\prime}z) \land  
         \all z\, (z \neq x \lor \lnot p^{\prime}z))\\
\equiv & \all z\, (z \neq y \lor g^{\prime}z) \lor 
         \all z\, (f^{\prime}z \lor z \neq x)\\
\equiv & g^{\prime}y \lor f^{\prime}x\\
\equiv & gyy \lor fxx,
\end{array}
\end{equation}
where the following shorthands are involved: $f^{\prime}z = fxz$,
$g^{\prime}z = gzy$, $p^{\prime}z = pzy$, $q^{\prime}z = pxz$.  The
two instances of $p$ are represented by different shorthands $p^{\prime}$ and
$q^{\prime}$ that are related by the equivalence $p^{\prime}x \equi
q^{\prime}y$, which is added (expressed with disjunction and conjunction) in
step~(*). Step~(**) is obtained by deleting two subformulas for which
elimination yields~$\true$.

So far, the contraction and expansion with the shorthand predicates has been
handled on the meta level.  Second-order quantification would allow to
understand the introduction and elimination of such definitions as equivalence
preserving transformations, with Ackermann's Lemma
\cite{ackermann:35}, a foundation of modern
elimination methods such as
\cite{dls,sqema,schmidt:2012:ackermann,ks:2013:frocos}, as a special case.
The basis for this understanding of introducing and expanding of auxiliary
definitions is the following property, which easily follows from
Lemma~\ref{lem-basic-elim}:
\begin{prop}[Eliminability of Expandable Definitions]
\label{prop-elim-expandable-definition}
Let $p$ be a unary predicate and let $F$ be a first-order formula in which $p$
does not occur.  It then holds that
\[\ex p\, \all x\, (px \equi F)\; \equiv\; \true.\]
\end{prop}
With Prop.~\ref{prop-elim-expandable-definition}, the following proposition
can be derived, which explicates the handling of auxiliary definitions by
means of equivalence preserving formula transformations:

\begin{prop}[Introduction and Elimination of Definitions]
\label{prop-def-elimintro}
Let $p$ be a unary predicate, let $x$ be an variable and let $G[x]$ be a
first-order formula in which~$p$ does not occur.  For a constant or
variable~$t$, let $G[t]$ denote $G[x]$ with all free occurrences of~$x$
replaced by~$t$.  Let $F[G[t_1],\ldots,G[t_n]]$ be a first-order formula in
which $p$ does not occur and which has~$n$ occurrences of subformulas,
instantiated with $G[t_1], \ldots, G[t_n]$, respectively, neither of them in a
context where a variable that occurs free in $G[x]$ is bound.  Let $F[pt_1,
  \ldots, pt_n]$ denote the same formula with the indicated occurrences
$G[t_i]$ replaced by $pt_i$. Then
\[F[G[t_1],\ldots,G[t_n]]\;
\equiv\; 
\ex p\, (\all x\, (px \equi G[x]) \land F[pt_1,\ldots, pt_n]).\]
\end{prop}

\noindent
Prop.~\ref{prop-def-elimintro} can be applied from left to right to
introduce auxiliary predicates~$p$ and from right to left to expand them, by
replacing \emph{all} occurrences of~$p$ with their definientia and then
dropping the definition.  

In \cite[\S~19]{beh:22}, the shorthand predicates are not explicitly handled
in this way, although the notation, where they are specified with the
equivalence symbol $\equi$, might suggest that this could be the underlying
intuition. In later works he explicitly applies a variant of
Prop.~\ref{prop-def-elimintro} to introduce definitions of nullary
predicates, presented in the manuscript \mref{man:beh:34:k8} as a second-order
adaption of \ref{rule-pullout-all} and \ref{rule-pullout-ex} (see also
(\ref{rule-pullout-sol}) in Sect.~\ref{sec-corr-elim-1934}).

In the case where in $F[pt_1,\ldots, pt_n]$ all occurrences of atoms with
predicate~$p$ are, say, positive, it holds that
\begin{equation}
\begin{array}{rl}
& \ex p\, (\all x\, (px \equi G[x]) \land
F[pt_1,\ldots, pt_n])\\
\equiv &  \ex p\, (\all x\, (px \imp G[x])
\land F[pt_1,\ldots, pt_n,])
\end{array}
\end{equation}
which leads to Ackermann's Lemma \cite{ackermann:35}, shown here
for unary predicates~$p$:
\begin{prop}[Ackermann's Lemma]
\label{prop-ackermann-lemma}
Assume the setting of Prop.~\ref{prop-def-elimintro} and that the
indicated subformula occurrences in $F[G[t_1],\ldots,G[t_n]]$ (or,
equivalently, in $F[pt_1,\ldots, pt_n]$) are either all positive or 
are all negative.

\smallskip

\sdlab{prop-ackermann-lemma-pos} If the indicated subformula
occurrences are positive, then
\[\ex p\, (\all x\, (px \imp G[x]) \land F[pt_1,\ldots, pt_n])
         \; \equiv \;
         F[G[t_1],\ldots,G[t_n]].\]
\sdlab{prop-ackermann-lemma-neg} If the indicated subformula
occurrences are negative, then
\[\ex p\, (\all x\, (px \revimp G[x]) \land F[pt_1,\ldots, pt_n]) 
         \; \equiv \;
         F[G[t_1],\ldots,G[t_n]].\]
\end{prop}
The Basic Elimination Lemma (Lemma~\ref{lem-basic-elim}) is obviously an
instance of Ackermann's Lemma. Vice versa, Ackermann's Lemma can be proven
such that the only elimination step is performed according to the Basic
Elimination Lemma, or according to Prop.~\ref{prop-elim-expandable-definition}.

\section{Ackermann's Quantifier Switching}
\label{sec-ackermann-switching}

In \cite{ackermann:35:arity}, a short sequel to \cite{ackermann:35},
Ackermann shows a precondition which allows to move existential predicate
quantification to the right of universal individual quantification, where the
arity of the quantified predicate is reduced:

\begin{lem}[Ackermann's Quantifier Switching]
\label{lem-ackermann-switching}
Let $p$ be a predicate with arity $n+1$, where $n \geq 0$. Let $F =
F[pxt_{11}\ldots t_{1n}\ldots  pxt_{m1}\ldots t_{mn}]$, where $m \geq
1$, be a formula of second-order logic in which $p$  has the
exactly $m$ indicated occurrences.  Assume further that $p$ and $x$ occur only
free in $F$.  Let $q$ be a predicate with arity $n$ that does not occur in $F$
and let $F[qt_{11}\ldots t_{1n},\ldots,$ $qt_{m1}\ldots t_{mn}]$ denote
$F$ with each occurrence $pxt_{ij}\ldots t_{ij}$ of $p$ replaced by
$qt_{ij}\ldots t_{ij}$, for $1 \leq i \leq n$, $1 \leq j \leq m$. Under the
assumption of the axiom of choice it then holds that
\[\begin{array}{rl}
& \ex p \all x\, F[pxt_{11}\ldots t_{1n},\;\ldots,\;
                          pxt_{m1}\ldots t_{mn}]\\
\equiv &
\all x \ex q\, F[qt_{11}\ldots t_{1n},\;\ldots,\;
                        qt_{m1}\ldots t_{mn}].
\end{array}
\]
\end{lem}

\noindent
Van Benthem \cite[p.~211]{benthem:book} mentions this equivalence with
application from right to left to achieve prenex form with respect to
second-order quantifiers.  Church discusses it in \cite[\S~56]{church:book} in
the context of well-ordering of the individuals. On its basis decidability of
the description logic $\mathcal{ALC}$ and related modal logics can be shown by
constructions of equi-satisfiable relational monadic second-order formulas
\cite{cw-relmon}.  Ackermann applies this equivalence in
\cite{ackermann:35:arity} to avoid Skolemization and to convert formulas such
that monadic techniques or Ackermann's Lemma become applicable.  He shows five
examples from \cite{schroeder}. The first of these rewrites the input formula
such that elimination methods for \MON becomes applicable. The formula to
which these are applied contains binary predicates, but as for Behmann's
examples in Sect.~\ref{sec-essentially-monadic}, in all occurrences of binary
predicates one of the arguments is free in the scope of the predicate
quantifier to eliminate. The example proceeds in the following steps:
\begin{equation}
\label{examp-ackermann-26}
\begin{array}{rll}
& \ex f\, (\all x \ex y\, (axy \land fxy) \land \all x
  \ex z\, (bxz \land \lnot fxz))\\
\equiv & \ex f \all x\, (\ex y\,  (axy \land fxy)
           \land \ex z\, (bxz \land \lnot fxz))
       & \ref{rule-all-out-and}\\
\equiv &
   \all x\, \ex g\, (\ex y\,  (axy \land gy)
         \land \ex z\, (bxz \land \lnot gz))
       & \text{Lemma~\ref{lem-ackermann-switching}}\\
\equiv & \all x \ex y \ex z\,  (axy \land bxz \land y \neq z).
& \text{Lemma~\ref{lem-hauptform-noeq-elim}}
\end{array}
\end{equation}
Ackermann notes in \cite{ackermann:35:arity} that he has discussed this
example already in \cite{ackermann:35} to illustrate another method: Applying
the variant of resolution-based elimination from \cite{ackermann:35} (of which
the modern SCAN \cite{scan} is a refinement -- see
\cite{nonnengart:elim:1999}), which involves conversion to a universal formula
by means of Skolem functions and un-Skolemization after elimination.

Indeed, Ackermann motivates his quantifier switching technique by the fact
that the introduction of Skolem functions (\de{Belegungsfunktionen}) leads to
such intricate tasks that one would like to avoid these functions for those
special cases where the results can also be obtained in other ways.  In this
sense, he also writes to Behmann in October 1934 that he does no longer
consider Skolem functions as an advantage.\footnotemark

\footnotetext{\de{Im übrigen halte ich neuerdings
  die Einführung der Belegungsfunktionen für keinen Vorteil mehr; die mit
  Hilfe der Belegungsfunktionen auszudrückenden Probleme, so einfach sie sich
  auch symbolisch ausdrücken lassen, werden so schwierig, dass ich da kein
  Weiterkommen sehe. Andererseits lassen sich in speziellen Fällen, wie bei
  meinem Beispiel~(26), erzielten Resultate ebensogut durch Anwendung der
  Formel
\[\all x \ex y\, gxy\; \equiv\;
  \ex f\, (\all x \ex y\, fxy \land \all x \all y\,
  (\lnot fxy \lor gxy))\] und entsprechender Formeln für mehr Variable
  gewinnen.} Letter from Wilhelm Ackermann to Heinrich Behmann, 29
  October 1934 \cref{corr:ab:1934:10:29}. Formulas are rendered in modern
  notation.  The mentioned \de{Beispiel (26)} is the example reproduced above
  as (\ref{examp-ackermann-26}).  The shown equivalence follows from
  Lemma~\ref{lem-ackermann-switching} and~\ref{lem-hauptform-noeq-elim}, but it
  does not immediately match with the referenced example.}

\part{Towards Elimination for Relations: 
The Correspondence between Behmann and Ackermann 1928--1934}
\label{part-polyadic}

\section{Introduction to Part~\ref{part-polyadic}}
\label{sec-intro-part-polyadic}

In this part Behmann's follow-up works to \cite{beh:22}, his related
correspondence with Ackermann and his presentation of the resolution-based
elimination method by Ackermann \cite{ackermann:35} are summarized. The topic
is the decision problem for relations (\de{Entscheidungsproblem für
  Beziehungen}), that is, the decision problem for formulas in which
predicates of arity two or more occur.  The approach is, as for the monadic
case in \cite{beh:22}, to apply second-order quantifier elimination
techniques.

In \cite[p.~226f]{beh:22},
Behmann 
\label{page-cite-limitations-of-elimination} 
conjectures that for the extension of the decision problem to arbitrary
relations and higher concepts it is questionable whether the elimination
problem can serve further as a suitable basis, justified on the following
considerations: \enq{If two classes $\alpha$ and $\beta$ satisfy a condition
  that can be specified purely logical and involves some variable classes --
  for example, that there is a third class which contains $\alpha$ as subclass
  and is itself contained in $\beta$ as subclass --, then we know indeed that
  we can express such a condition certainly also without mentioning such
  variable classes.  The matter is, however, as it seems, much more intricate
  if a variable \emph{relation}\footnote{As pointed out later by Ackermann \cite[p.~393]{ackermann:35},
  it is, however, not essential that the predicate \emph{to eliminate} has an
  arity larger than one -- the same difficulties arise if the predicate to
  eliminate is unary but other predicates with two or more
  arguments do occur.} is allowed, as, for example, at
  the statement that two classes $\alpha$ and $\beta$ have the same
  cardinality, that is, that by a certain relation the elements of one class
  can be mapped in a one-to-one correspondence to that of the other one. Here
  one does not see a possibility to express the condition that two classes
  have the same cardinality in general and without reference to such a
  variable relation.  Presumable, again a completely new idea is required
  here.}\footnotemark

\footnotetext{\de{Was die Erweiterung des Entscheidungsproblems auf
beliebige Beziehungen und höhere Begriffe angeht, so erscheint es immerhin
fraglich, ob auch hier das Eliminationsproblem weiterhin als geeignete
Grundlage dienen können wird, und zwar auf Grund der folgenden Überlegung:
Genügen etwa zwei Klassen $\alpha$ und $\beta$ einer rein logisch angebbaren
Bedingung, innerhalb deren irgendwelche veränderliche Klassen vorkommen --
sagen wir z.~B. derjenigen, daß es eine dritte Klasse gibt, die $\alpha$ als
Teilklasse enthält und ihrerseits als Teilklasse in $\beta$ enthalten ist --,
so wissen wir allerdings, daß wir eine solche Bedingung gewiß auch ohne
Erwähnung derartiger veränderlicher Klassen auszudrücken vermögen.  Die Sache
liegt indessen, wie es scheint, wesentlich verwickelter, sobald eine
veränderliche \emph{Beziehung} in Frage kommt, wie z.~B. bei der Aussage, daß
die Klassen $\alpha$ und $\beta$ gleichzahlig sind, d. h. daß durch eine
gewisse Beziehung die Elemente der einen denen der anderen umkehrbar eindeutig
zugeordnet werden. Hier sieht man durchaus keine Möglichkeit, die Bedingung
der Gleichzahligkeit zweier Klassen allgemein ohne einen Hinweis auf eine
derartige veränderliche Beziehung auszudrücken. Vermutlich wird es hier
also wiederum eines ganz neuen Gedankens bedürfen.} \cite[p.~226f]{beh:22}. 
}

Behmann gave in September 1926 at the \dename{Jahresversammlung der Deutschen
  Mathematiker-Vereinigung} a talk on the decision problem and the logic of
relations (\dename{Entscheidungsproblem und Logik der Beziehungen}). Its
abstract, published as \cite{beh:26:beziehungen}, aroused the curiosity of
Ackermann, who wrote in August 1928 to Behmann, initiating a correspondence
that lasted to November 1928 and comprises five letters.  Related topics were
discussed later in two letters, the first sent by Behmann upon receiving the
offprint of \cite{ackermann:35}, the second by Ackermann in reply.  Their
correspondence, as far as archived in \cite{beh:nl}, then only continues in
January 1953, with five more letters until December 1955, where different
topics are discussed (a complete register of their correspondence in
\cite{beh:nl} is provided in Sect.~\ref{sec-corr}).

Further sources of this presentation include a manuscript
\mref{man:beh:26:beziehungen} for the abstract \cite{beh:26:beziehungen} and
an unpublished manuscript from December 1934 \mref{man:beh:34:k8} \de{Ein
  wichtiger Fortschritt im Entscheidungsproblem der Mathematischen Logik}
(\name{An Important Progress in the Decision Problem of Mathematical Logic}
with subtitle \name{(Ackermann Math. Annalen 110 S.390)}, referring to
\cite{ackermann:35}.

The addressed topics include the use of Skolemization as well as the early
form of resolution by Ackermann \cite{ackermann:35}, applied to express
results of second-order quantifier elimination on universal formulas by a
possibly infinite set of formulas.  Resolution-based second-order quantifier
elimination has been considered in modern times with the SCAN algorithm
\cite{scan}.  The use of logics extended with a fixpoint operator to express
possibly infinite elimination results has been investigated in
\cite{nonnengart:fixpoint} and is today one of the core techniques for
second-order quantifier elimination (or ``forgetting'') in description logics
\cite{ks:2013:frocos}. So far, however, the fixpoint approach is based not on
Ackermann's resolution method, but on Ackermann's Lemma, another result from
\cite{ackermann:35} (see Prop.~\ref{prop-ackermann-lemma} on
p.~\pageref{prop-ackermann-lemma}).  A further work that explicitly relates to
Ackermann's resolution-based method is \cite{craig:bases}.  In his manuscript
from 1934 and his letter to Ackermann, Behmann suggests a graph representation
for the possibly infinite set of resolvents.

We use the same modern syntax for the presentation as in the other parts, but
keep the Greek letters $\varphi, \chi, \psi$, and $\alpha, \beta, \gamma$ from
the original documents. Ackermann originally writes predicates in upper case,
which are rendered here in lower case.  Detailed information about the
consulted manuscripts and letters, as well as a summary of other topics
discussed in the correspondence is provided in Sect.~\ref{sec-manuscripts} and
Sect.~\ref{sec-corr}.

\section{First Considerations on Elimination for Universal Formulas}
\label{sec-corr-elim-1928a}

In the abstract \cite{beh:26:beziehungen} of his talk
\dename{Entscheidungsproblem und Logik der Beziehungen} (\name{Decision
  Problem and Logic of Relations}) given on 23 September 1926 at the
\de{Jahresversammlung der Deutschen Mathematiker-Vereinigung} in Düsseldorf,
Behmann considers elimination of an existential second-order quantifier upon a
predicate~$\varphi$ with arity two or larger applied to a universal
first-order formula, that is, elimination in
\begin{equation}
\ex \varphi \all x_1 \ldots \all x_n\, F,
\end{equation}
where $\varphi$ is a predicate with arity $\geq 2$ and $F$ is a first-order
formula. 
Behmann states in \cite{beh:26:beziehungen}, that this problem allows
reduction to the following sequence of problems:
\begin{equation}
\label{eq-sequence-sol}
\begin{array}{l}
\ex \varphi \all x\, F[\varphi x, x],\\
\ex \varphi \ex \chi \all x \all y\, F[\varphi x,\chi y, x, y],\\
\ex \varphi \ex \chi \ex \psi \all x \all y \all z\, 
F[\varphi x,\chi y,\psi z, x, y, z],\\
\ldots\footnotemark
\end{array}
\end{equation}
where the quantified predicates 
\footnotetext{In a letter to Alonzo Church, dated 15 April 1937, Behmann
  refers in a broader context to this sequence of formulas. See
  Sect.~\ref{sec-letters-church}.}
are unary and the formulas~$F[\ldots]$ are
first-order and such that all occurrences of the quantified predicates have
the indicated variable as argument. In addition, the individual variables
themselves are listed in the square brackets, indicating that they might
also have further occurrences in the formula.\footnote{\cite{beh:26:beziehungen}
  contains some obvious printing errors that have been quietly corrected here
  -- see discussion of manuscript~\mref{man:beh:26:beziehungen}
in Sect.~\ref{sec-manuscripts}.}
Behmann further claims in \cite{beh:26:beziehungen} that formulas of that form
would allow elimination of the predicate quantifiers.

The draft~\mref{man:beh:26:beziehungen} of \cite{beh:26:beziehungen} gives
some further details: The solution of the first component of the sequence is
\begin{equation}
\label{eq-sol-row-one}
\all x \ex p\, F[p,x],
\end{equation}
or
\begin{equation}
\all x\,  (F[\false,x] \lor F[\true,x]),
\end{equation}
respectively, where $p$ is a fresh nullary predicate and $F[G]$ denotes
$F[\varphi x]$ with all occurrences of $\varphi x$ replaced by $G$.  The
equivalence of the first component to~(\ref{eq-sol-row-one}) can indeed be
obtained by Ackermann's arity reduction, Lemma~\ref{lem-ackermann-switching},
which seems implicitly applied here by Behmann.
The following formula is then shown in~\mref{man:beh:26:beziehungen} as
intended solution of the second component:
\begin{equation}
\label{eq-sol-row-two}
\all x \ex p \all y \ex q\, F[p,q,x,y]
\land
\all y \ex q \all x \ex p\, F[p,q,x,y].
\end{equation}
For the second component, Behmann's claim is false.  On 16 August 1928
Ackermann writes to him \cref{corr:ab:1928:8:16}, requesting clarification
by giving the following example:
\begin{equation}
\begin{array}{lll}
\ex \varphi \all x \all y &
((\varphi x \lor \lnot \varphi y\lor \alpha xy) & \land\\
& \hparen (\varphi x \lor \beta x) & \land\\
& \hparen (\lnot \varphi x \lor \gamma x)),
\end{array}
\end{equation}
which can be written as instance of the second component of
(\ref{eq-sequence-sol}):
\begin{equation}
\begin{array}{lll}
\ex \varphi \ex \chi \all x \all y
&
((\varphi x \lor \lnot \chi y\lor \alpha xy) & \land\\
& \hparen (\varphi x \lor \beta x) & \land\\
& \hparen (\chi y \lor \gamma y) & \land\\
& \hparen (\varphi x \lor \chi y \lor x \neq y) & \land\\
& \hparen (\lnot \varphi x \lor \lnot \chi y \lor x \neq y)).
\end{array}
\end{equation}
Ackermann continues that he believes to have thought through for some simple
expressions of that kind that they can be replaced by no logically equivalent
expression that is constructed only from individual quantifiers, $\alpha$,
$\beta$, $\gamma$ and identity.

Behmann replies on 21 August 1928 \cref{corr:ba:1928:8:21} from the holiday
island Föhr at the German North Sea coast that some years ago he had achieved
a certain point in his investigations of the decision problem for relations
but had to leave it unattended since then. 
\label{page-beh-announce-to-ackermann} He announces to send his results to
Ackermann as soon as he returns to Halle.  Behmann continues that he became
aware only between the talk and the correction that his claim to have settled
the case of elimination for universal individual quantifiers was false. He had
informed \label{bieberbach} Bieberbach\footnote{Ludwig Bieberbach (1886--1982)
  was the editor of the respective number of the \de{Jahresbericht der
    Deutschen Mathematiker-Vereinigung}.} but decided to keep the statement
for historic accuracy in essence in the abstract, with the intention of
addressing the issue later in a publication.  Behmann notes that after many
void attempts to find a satisfying proof, he came to a specific example where
his conjectured solution was in fact weaker, that is, implied but not
equivalent to the given second-order formula.  He remarks that Ackermann's
example would be of great interest, since it might show that the schema of
progressive concept elimination would not be sufficient and a fundamentally
different way has to be searched.  Ackermann gave later in
\cite[Section~3]{ackermann:35} a proof that the considered elimination
problems on formulas with only universal individual quantifiers can not be
solved in general.

On 1 September 1928 \cref{corr:ab:1928:9:1} Ackermann thanks Behmann for his
informative reply. He notes that it took a load from his mind since he had
already struggled a lot with the elimination problem in the case where all
individual quantifiers are universal.\footnote{\de{Mir ist dadurch ein Stein
    vom Herzen gefallen, da ich mich mit dem Eliminationsproblem in dem Falle,
    daß alle Dingoperatoren allgemein sind, schon viel herumgequält hatte.}
\cref{corr:ab:1928:9:1}.}
He continues that the probability that an elimination is not always possible
had emerged for him at first from the following example:
\begin{equation}
\label{ex-ackermann-first}
\ex f \all x \all y\, (axy \lor ((fx \lor fy) \land
(\lnot fx \lor \lnot fy))).
\end{equation}
The following formulas are then obtained as consequences:
\begin{equation}
\label{ex-ackermann-first-sol}
\begin{array}{l}
\all x\, axx,\\
\all x \all y \all z\, (axy \lor ayz \lor axz),\\
\all x \all y \all z \all u \all v\, 
(axy \lor ayz \lor azu \lor auv \lor axv),\\
\ldots
\end{array}
\end{equation}
but one searches in vain for an expression that contains all of these as
consequences. Ackermann writes that he does not yet have an exact proof that the
elimination can not be performed, but hopes to find one within some weeks,
even if he expects it to be quite complicated.

The example~(\ref{ex-ackermann-first}) is also used later by Ackermann, as
Example~(20) on p.~410 in \cite{ackermann:35}, to illustrate his
resolution-based elimination method and discussed in his letter to Behmann from
October 1934 \cref{corr:ab:1934:10:29}.

\section{Involvement of Skolemization and Un-Skolemization}
\label{sec-corr-elim-1928b}
\label{sec-skolemization}
In the abstract \cite{beh:26:beziehungen} of the Düsseldorf talk Behmann goes
beyond the case of universal individual quantifiers to cases where universal
and existential quantifiers alternate, as for example in formulas of the form
\begin{equation}
\ex \varphi \all x \ex y\, F[\varphi x,\varphi y,x,y].
\end{equation}
The problem would be solved, if the order of $\all x$ and $\ex y$ could
be switched, which would then allow to move $\ex y$ further left, in front
of the predicate quantifier $\ex \varphi$.  Behmann remarks that this can
be achieved indeed by what today is called Skolemization, introduced in
\cite{beh:26:beziehungen} as a new view on a trick (\dename{Kunstgriff})
described by Schröder in \cite[p.~512ff]{schroeder:3}.\footnote{Second-order
  Skolemization allows to convert an existential individual quantifier in the
  scope of universal quantifiers to an existential function quantifier that is
  left of the universal quantifiers. The underlying equivalence is
\[\all x_1 \ldots x_n \ex y F[y]
  \; \equiv\; \ex f \all x_1 \ldots x_n F[fx_1\ldots x_n],\] where
  $f$ is a $n$-place function symbol that does not occur in $F[y]$ and
  $F[fx_1\ldots x_n]$ is obtained from $F[y]$ by replacing all free occurrences
  of $y$ with $fx_1\ldots x_n$.}  
Behmann writes that instead of predicate quantifiers, now quantifiers upon the
Skolem functions (\de{Belegungsoperatoren}) do appear, which can not be undone
in general -- at least with known means (what today is called
un-Skolemization).  He mentions that a peculiar extension of his
representation schema succeeds in undoing Skolemization, essentially by
allowing aside of the so far solely known and applied succession of
quantifiers $\all$ and $\ex$ more entangled linkages of them,
associated with specific meanings.\footnote{Such generalized forms of
  quantification play a key role in works of Jakko Hintikka, e.g.,
  \cite{hintikka:pmr}.} This seems to anticipate that un-Skolemization after
predicate elimination can in some cases only obtained with Henkin quantifiers
\cite{scan}. In the manuscript \mref{man:beh:26:beziehungen}, he adds that
these linkages concern in particular cycles, which, for example, do no longer
comply with the transitivity of sooner of later.  He concludes the talk
abstract with the remark that control of the decision problem for the
considered class of formulas then appears as equivalent to control of this
extended representation schema.

After reading the elaborate transcript of the Düsseldorf talk sent to him by
Behmann,\footnote{So far, that transcript could not be located -- see remarks on
  \cref{corr:ba:1928:9:29} in Sect.~\ref{sec-corr} and on
  \mref{man:beh:26:beziehungen:incomplete} in Sect.~\ref{sec-manuscripts}.  In
  1934 Behmann also gave a detailed account of his earlier work on the
  decision problem for relations in letter \cref{corr:ba:1934:10:22} and
  manuscript \mref{man:beh:34:k8}, which is discussed below in
  Sect.~\ref{sec-corr-elim-1934}.}
Ackermann expresses in his letter dated 1 November 1928
\cref{corr:ab:1928:11:1} doubts about Behmann's use of Skolemization and
un-Skolemization, conjecturing that it amounts to a re-expression of the
decision problem with functions, which can be used to encode predicates: That
$fx$ holds can be expressed as $\varphi x = 0$, where $\varphi$ is a function
with range $\{0, 1\}$ associated with $f$.\footnotemark

\footnotetext{\de{Glauben Sie aber nicht, daß das Problem, das sich als
    Schlußproblem Ihrer Untersuchung darstellt, ebenso schwierig ist wie das
    allgemeine Entscheidungsproblem? Die kompliziertesten Verkettungen von
    Operatoren, die dabei auftreten, kommen doch schließlich darauf hinaus,
    daß man statt der Prädikate und Relation Funktionen im Sinne der
    Mathematik enführt, d.h. eindeutige Zuordnungen von Individuen zu
    Individuen. Man kann dann aber gleich jedes Entscheidungsproblem in dieser
    Form geben, indem man ein Prädikat $F(x)$ ersetzt durch eine derartige
    Funktion $\varphi(x)$, wo $\varphi(x)$ entweder den Wert $0$ oder $1$ hat,
    und $\varphi(x) = 0$ mit dem Bestehen von $F(x)$ äquivalent ist.}
    \cref{corr:ab:1928:11:1}.}

In that letter, Ackermann also reports that he got notice that Löwenheim in
Berlin has exact proofs that elimination can not be performed on certain logic
expressions, and that it would be desirable that Löwenheim would publish these
because it would demonstrate that the existing way does not lead further
on.\footnote{Löwenheim's remarks in \cite[p.~336]{loewenheim:35} might refer
  to these results.}  After remarking that the investigation of a more modest
question seems to him most promising, Ackermann sketches an idea for a method
to verify valid universal formulas (see discussion of \cref{corr:ab:1928:11:1}
in Sect.~\ref{sec-corr}).

\section{Behmann's Representation of Ackermann's Resolution-Based Method}
\label{sec-corr-elim-1934}

Today, the two prevailing approaches to second-order quantifier elimination
are the so-called direct methods, based on Ackermann's Lemma
\cite{dls:early,dls}, and the resolution-based approach of the SCAN algorithm
\cite{scan}.  Both can be traced back to \cite{ackermann:35}, where the lemma
underlying the direct approach is defined and a variant of the -- seemingly
rediscovered \cite{nonnengart:elim:1999} -- resolution-based approach is
elaborated.  Further works related to Ackermann's resolution-based method have
been already mentioned in Sect.~\ref{sec-intro-part-polyadic}.

Aside \label{page-ackermann-bernays} of the correspondence with Behmann, a
published selection from the correspondence of Wilhelm Ackermann edited by his
son Hans-Richard Ackermann \cite{ackermann:briefwechsel} gives further hints
on the ``pre-history'' of Ackermann's important paper \cite{ackermann:35}: He
sent the manuscript in 1933 to Bernays (Letter from Bernays to Ackermann, 24
December 1933), who then recommended it to Hilbert for publication in
\dename{Mathematische Annalen}. Bernays sent six large pages with
remarks about the manuscript, regarding content as well as presentation, to
Ackermann, which he considered for the version submitted to
Blumenthal\footnote{Otto Blumenthal (1876--1944) was the editor of the
  \dename{Mathematische Annalen} responsible for Ackermann's paper.}  (Letter
from Ackermann to Bernays, 14 January 1934).  Accordingly, the submission date
given in the publication is 13 January 1934.

After receiving the offprint of Ackermann's paper, Behmann writes on
22~October 1934 to Ackermann \cref{corr:ba:1934:10:22}, congratulating him to
his success, noting that the work shows to him that the proper access was
quite hidden.  Behmann mentions that he perceived the existing research line,
which aimed at a generalization of the method of substitution
(\de{Verallgemeinerung des Einsetzungsverfahrens}) as unsatisfying. He admires
Ackermann's methodical-technical generality by posing only little requirements
on normalization (\de{Normierung}) of the given problem, however, as Behmann
remarks, this entails the disadvantage that it then seems not possible to
clearly overview in general the totality of the resolvents that are free from
the predicate to eliminate (\de{die Gesamtheit der Zählausdrücke II allgemein
  zu übersehen}) and to symbolize the resultant in a suitable transparent way.
Behmann proceeds with a technical presentation which relates Ackermann's
results to his own earlier work and suggests alternative more normalized
representations of the elimination resultants obtained by Ackermann's
resolution-based method.

In his reply of 29 October 1934 \cref{corr:ab:1934:10:29}, Ackermann
expresses his pleasure about Behmann's acknowledgment, in particular because
Behmann's work from 1922 \cite{beh:22} was, at its time, the impetus for his
investigation of the problem.\footnotemark\ He then discusses the issue of
getting a complete overview on the resultants (\name{vollkommene Übersicht
  über die Resultanten}), raised by Behmann.

\footnotetext{\de{Für Ihren Brief und die freundlichen Worte über meine Arbeit
    meinen besten Dank.  Ich freue mich über die Anerkennung um so mehr, als
    Ihre Arbeit in den Math. Ann.~86 seiner Zeit mir den Anstoss gegeben hat,
    mich mit dem Problem näher zu beschäftigen} \cref{corr:ab:1934:10:29}. In
  1922 Ackermann finished his basic studies \name{(Grundstudium)} at Göttingen
  \cite{goettingen:22,ackermann:briefwechsel}.}

Following the exchange with Ackermann, Behmann prepared a manuscript
\mref{man:beh:34:k8}, titled \de{Ein wichtiger Fortschritt im
  Entscheidungsproblem der Mathematischen Logik} (\name{An Important Progress
  in the Decision Problem of Mathematical Logic}, with subtitle
\name{(Ackermann Math. Annalen 110 S.390)}, dated 14~December 1934, where the
technical material in part overlaps with the letter \cref{corr:ba:1934:10:22}.
Here he gives a more detailed presentation of his suggestions to represent the
resultants obtained by Ackermann's method.

Behmann's manuscript \mref{man:beh:34:k8} begins with a comprehensive
introductory part.  After specifying notation and introducing the decision
problem, Behmann sketches two methods to decide propositional validity: the
method of truth-tables (\name{Verfahren der Einsetzungsproben}) and the method
of conversion to a conjunctive normal form (\name{Verfahren der konjunktiven
  Normalform}), which is valid if and only if each of its clauses contains a
literal and its complement. He then summarizes his method for second-order
quantifier elimination in relational monadic formulas with equality,
attributed by him to \de{Löwenheim, Skolem, Behmann}. Quantified predicates
are eliminated successively from the inside of the formula by equivalence
preserving transformations by a system of deterministic computation rules
(\de{durch äquivalente Umformung durch ein System zwangsläufiger
  Rechenvorschriften}).  Behmann speaks there also of ``logic algorithm''
(\de{logischer Algorithmus}).  The result expresses ``the sentence to test is
true (or false) for all domain cardinalities with exception (when appropriate)
of the finite number of cardinalities $m, n, \ldots$''. For sufficiently large
finite and for all infinite domains the value of the sentence is the same.

A second approach to decide a sentence is then sketched: It can be tested for
particular cardinalities of the individual domain, based on the fact that
there are theorems that allow to determine a suitable upper bound from its
symbolic structure.  The sentence is then repeatedly evaluated for each domain
cardinality, up to the previously determined bound. Behmann gives a short
example, but concludes that this method will never be applied in practice.

Behmann now turns to adding predicates with arities larger than one, in his
words, \name{variable relations between individuals} (\de{variable Beziehungen
  zwischen Individuen}). He discusses the possibility to generalize the method
of successive elimination of quantified predicates, which lead to success in
the monadic case. Duality justifies to consider just $\ex \varphi$ as
innermost predicate quantifier. Quantified predicates other than $\varphi$ are
then free in the argument formula and behave in the elimination process just
like unquantified predicates. The formula considered for the elimination
problem thus has the form
\begin{equation}
\ex \varphi\, F,
\end{equation}
where in $F$ is a first-order formula.
The argument formula $F$ can be assumed in prenex form. Behmann now restricts
his considerations to the special case where in the quantifier prefix no
existential quantifier is following a universal quantifier.  Since the
existential individual quantifiers can then be moved in front of the
existential predicate quantifier, this restriction amounts to requiring the
individual quantifier prefix of the argument of predicate quantification to be
just universal. The considered elimination problem can then be brought into
one of the following forms:
\begin{equation}
\label{eq-beh-prop-sequence}
\begin{array}{l}
\ex \varphi \all x\, F[\varphi x,x],\\
\ex \varphi \all x \all y \, F[\varphi x,\varphi y,x,y],\\
\ex \varphi \all x \all y \all z\, 
F[\varphi x,\varphi y,\varphi z,x,y,z],\\
\ldots
\end{array}
\end{equation}
where the $F[\ldots]$ are first-order and all occurrences of $\varphi$ have
the indicated variables as argument. In addition, the individual variables
themselves are listed in the square brackets, indicating that they might also
have further occurrences in the formula.  (This series is like
(\ref{eq-sequence-sol}), p.~\pageref{eq-sequence-sol}, except that there a
different predicate is used for each argument variable.)
It suffices to consider predicate quantification just upon \emph{unary}
predicates, because quantification upon predicates with larger arity can be
modeled by a new domain whose individuals are tuples of the original
individuals.
The solution of the first problem in the
sequence~(\ref{eq-beh-prop-sequence}) is
\begin{equation}
\all x\,  (F[\false,x] \lor F[\true,x]),
\end{equation}
where $F[G,x]$ denotes $F[\varphi x,x]$ with all occurrences of $\varphi x$
replaced by $G$.

Behmann proceeds to discusses just the second problem in the
sequence~(\ref{eq-beh-prop-sequence}), because in his view there is no
principal additional complication at the passage to the third and later
problems.  In analogy to \ref{rule-pullout-all} and \ref{rule-pullout-ex}
(\de{Umschreibungssatz}), the following second-order equivalences hold:
\begin{equation}
\label{rule-pullout-sol-draft}
F[p]\; \equiv\; \all q\, (\lnot (q \equi p) \lor F[q])\;
      \equiv\; \ex q\, ((q \equi p) \land F[q]).
\end{equation}
Behmann remarks that this holds for an \emph{extensional} property of
sentences ((\de{extensionale!) Eigenschaft von Aussagen}).
Equivalence~(\ref{rule-pullout-sol-draft}), given by Behmann, is an instance
of a version of Prop.~\ref{prop-def-elimintro} for nullary predicates.  The
equivalence actually applied by Behmann in the sequel can accordingly be
expressed as follows, for formulas~$G$ under preconditions analogously to
those in Prop.~\ref{prop-def-elimintro}:
\begin{equation}
\label{rule-pullout-sol}
F[G]\; \equiv\; \all p\, (\lnot (p \equi G) \lor F[p])\;
      \equiv\; \ex p\, ((p \equi G) \land F[p]).
\end{equation}
With (\ref{rule-pullout-sol}), the following equivalences can be justified:
\begin{equation}
\label{eq-binary-expand}
\begin{array}{r@{\hspace{1em}}lll}
(1.) && \ex \varphi \all x \all y\, F[\varphi x, \varphi y, x, y]\\
(2.) & \equiv & \ex \varphi \all x \all y \all p \all q\, 
    (\lnot (p \equi \varphi x) \lor
    \lnot (q \equi \varphi y) \lor F[p,q,x,y])\\
(3.) & \equiv & \ex \varphi \all x \all y\, 
      ((F[\false,\false,x,y] \lor \varphi x \lor \varphi y) & \land\\
&&     \indebe \hparen
       (F[\false,\true,x,y] \lor \varphi x \lor \lnot \varphi y) & \land\\
&&     \indebe \hparen
       (F[\true,\false,x,y] \lor \lnot \varphi x \lor \varphi y) & \land\\
&&     \indebe \hparen
       (F[\true,\true,x,y] \lor \lnot \varphi x \lor \lnot \varphi y)),
\end{array}
\end{equation}
where $F[G,H,x,y]$ denotes $F[\varphi x,\varphi y,x,y]$ with all occurrences
of $\varphi x$ replaced by $G$ and all occurrences of $\varphi y$ replaced
by $H$. Step~(3.) is obtained by expanding the Boolean quantifiers upon $p$
and $q$, considering that $\lnot (\false \equi \varphi x) \equiv \varphi x$
and $\lnot (\true \equi \varphi x) \equiv \lnot \varphi x$.
Behmann now introduces the following form as shorthand for formulas of the
form (3.) in (\ref{eq-binary-expand}):
\begin{equation}
\label{eq-ebe-problem}
\begin{array}{ll}
\ex \varphi \all x \all y\,
      ((fxy \lor \varphi x \lor \varphi y) & \land\\
    \indebe \hparen
      (gxy \lor \varphi x \lor \lnot \varphi y) & \land\\
    \indebe \hparen
      (hxy \lor \lnot \varphi x \lor \varphi y) & \land\\
     \indebe \hparen
      (kxy \lor \lnot \varphi x \lor \lnot \varphi y)),
\end{array}
\end{equation}
where $f,g,h,k$ can be arbitrary binary relations on the underlying domain of
individuals. In themselves, $f,g,h,k$ are not subjected to any restrictions,
however, certain constraints on them can be
enforced. Formula~(\ref{eq-ebe-problem}) is equivalent to
\begin{equation}
\label{eq-ebe-symm}
\begin{array}{ll}
 \ex \varphi (\all x \all y\,
      ((fxy \lor \varphi x \lor \varphi y) & \land\\
    \indebex \hparen
      (gxy \lor \varphi x \lor \lnot \varphi y) & \land\\
    \indebex \hparen
      (hxy \lor \lnot \varphi x \lor \varphi y) & \land\\
     \indebex \hparen
      (kxy \lor \lnot \varphi x \lor \lnot \varphi y)) & \land\\
 \hphantom{\ex \varphi} \hparen \all x \all y\,
      ((fyx \lor \varphi y \lor \varphi x) & \land\\
    \indebex \hparen
      (gyx \lor \varphi y \lor \lnot \varphi x) & \land\\
    \indebex \hparen
      (hyx \lor \lnot \varphi y \lor \varphi x) & \land\\
    \indebex \hparen
      (kyx \lor \lnot \varphi y \lor \lnot \varphi x)) & \land\\
 \hphantom{\ex \varphi} \hparen \all x \all y\,
      (x \neq y \lor \varphi x \lor \lnot \varphi y) & \land\\
 \hphantom{\ex \varphi} \hparen \all x \all y\,
      (x \neq y \lor \lnot \varphi x \lor \varphi y)),
\end{array}
\end{equation}
where last two lines are necessarily true, and thus
further equivalent to
\begin{equation}
\label{eq-ebe-symm-final}
\begin{array}{ll}
\ex \varphi \all x \all y\,
      (((fxy \land fyx) \lor \varphi x \lor \varphi y) & \land\\
\indebe \hparen ((gxy \land hyx \land x \neq y) 
                       \lor \varphi x \lor \lnot \varphi y) & \land\\
\indebe \hparen ((hxy \land gyx \land x \neq y) 
                        \lor \lnot \varphi x \lor \varphi y) & \land\\
\indebe \hparen ((kxy \land kyx) \lor \lnot \varphi x \lor \lnot \varphi y)).
\end{array}
\end{equation}
Behmann arguments that the form (\ref{eq-ebe-symm-final}) justifies to assume
that in a problem expression of the form (\ref{eq-ebe-problem})
\begin{itemize}
\item the relations $f$ and $k$ are symmetric, and
\item the relations $g$ and $h$ are irreflexive and inverse to each
  other.\footnote{Although intuitively convincing, it seems still from a
    modern point of view not straightforward how a system that just gets
    (\ref{eq-ebe-symm-final}) as input automatically utilizes these
    properties.}
\end{itemize}
In \label{page-note-drop-h} his reply of 20 October 1934
\cref{corr:ab:1934:10:29}, Ackermann notes that the component $(hxy \lor \lnot
\varphi x \lor \varphi y)$ in (\ref{eq-ebe-problem}) is not necessary, since it
can be united with the component $(gxy \lor \varphi x \lor \lnot \varphi y)$ by
switching variables.  Accordingly, Behmann indicates in his
manuscript~\mref{man:beh:34:k8} with pencil a second way of reading, where the
component $(hxy \lor \lnot \varphi x \lor \varphi y)$ is dropped and $g$ is just
constrained by irreflexivity.

In \mref{man:beh:34:k8}, Behmann concludes the section with remarking that
this was the point where he arrived several years before, but without seeing a
prospect to make decisive progress.  Before discussing Ackermann's result, he
recapitulates the intermittent research on the decision problem for relations.
He remarks that any reduction to the monadic case with equality is precluded,
since a sentence like \name{There exists an infinite number of individuals}
can for sure not be expressed by a sentence without predicates, such as, for
example $\ex x \ex y\, x \neq y$. Hence, Bernays, Ackermann,
Schönfinkel, and Schütte in Göttingen passed on to exploring with the much
more primitive method of testing for particular domain cardinalities.  The
issue was to derive from the structure of the given sentence a suitable bound
for the domain cardinality for which the test has to be made to get a result
that applies to arbitrarily large domain cardinalities. However, this is not
sufficient to manage a sentence that is as simple as \name{There exists an
  infinite number of individuals}, because tests can only be made for finite
domain cardinalities, and even there only in theory and not
practically. Behmann continues to summarize Schütte's result
\cite[p.~603]{schuette:1934}, which specifies for equality free sentences with
quantifier prefix $\all \varphi_1 \ldots \all \varphi_h \all x_1
\ldots \all x_k \ex y_1 \ex y_2 \all z_1 \ldots \all z_l$ a
number such that for all domain cardinalities that are larger or equal the
truth value is the same.  Behmann annotates that the result could have well 
been
extended to consider also the case with equality. He proceeds to note that, in
addition, Schütte has shown that for the case of more the two existential
variables in the prefix such a number is no longer determinable in general.
Behmann concludes this section with some acknowledging words on the effort and
persistence of those working in these research directions, in awareness that
the way would not lead to any practically usable results and that they would
get stuck irrevocably at a rather early point.\footnotemark

\footnotetext{\de{Es handelt sich bei den Bearbeitungen in diesen
    Forschungsrichtungen um sehr schwierige, aber mathematisch schöne und
    tiefe Untersuchungen.  Und man muß die darauf verwandte \emph{Mühe und
      Ausdauer} um so mehr bewundern, wenn man bedenkt, daß ja den Bearbeitern
    selbt nicht unbekannt sein konnte, daß sie auf diesem Wege überhaupt
    \emph{zu keinem irgendwie praktisch auswertbaren Ergebnis} gelangen
    konnten und obendrein schon \emph{an einem ziemlich frühen Punkte
      endgültig stecken bleiben} mußten.} \mref{man:beh:34:k8}, p.13f.}

Behmann now turns in \mref{man:beh:34:k8} to Ackermann's result in
\cite{ackermann:35}, noting that Ackermann's starting point was an example
that was as simple as possible but for which he could show that the known
means of representation failed. The axiom of induction (\de{Satz von der
  vollständigen Induktion}) can be expressed as the following second-order
sentence
\begin{equation}
\all \varphi\, ((\varphi 0 \land 
\all m \all n\, ((\varphi m \land m+1=n) \imp \varphi n)) \imp
\all r\, \varphi r),
\end{equation}
where domain of the individual quantifiers is the set of natural numbers.  The
quantifier $\all r$ can be moved in front of $\all \varphi$.  If the
constant $0$ and free variable $r$ are written as $a$ and $b$, respectively,
and $m+n=1$ is abbreviated as $fmn$, then the elimination problem can be
expressed as the first line in the following equivalences:
\begin{equation}
\label{eq-nat-induction}
\begin{array}{rl}
& \all \varphi\, ((\varphi a \land 
\all x \all y\, ((\varphi x \land fxy) \imp \varphi y)) \imp
\varphi b)\\
\equiv & \all \varphi\, (\lnot \varphi a \lor
\ex x \ex y\, (\varphi x \land fxy \land \lnot \varphi y) \lor
\varphi b)\\
\equiv &
\all \varphi\, (\ex x\, (x=a \land \lnot \varphi x) \lor
\ex x \ex y\, (\varphi x \land fxy \land \lnot \varphi y) \lor
\ex y\, (y=b \land \varphi y))\\
\equiv &
\all \varphi \ex x \ex y\, ((y = a \land \lnot \varphi x) \lor
(\varphi x \land fxy \land \lnot \varphi y) \lor
(y=b \land \varphi y)).
\end{array}
\end{equation}
The last two steps in (\ref{eq-nat-induction}) are obtained by
\ref{rule-pullout-ex} and \ref{rule-ex-out-or}, respectively.  The dual or the
negation of the last step matches the problem
specification~(\ref{eq-ebe-problem}).

As Behmann recapitulates, Ackermann has \emph{proven} in \cite{ackermann:35}
that a predicate free formulation of the induction axiom is not possible, but
Ackermann also says that in a certain new sense the elimination can
nevertheless be performed. The induction axiom (in the form of the first line
of (\ref{eq-nat-induction})) expresses that one of the following sentences
holds:
\begin{equation}
\label{eq-form-unduction-orseq}
\begin{array}{l}
a=b,\\
fab,\\
\ex x_1\, (fax_1 \land fx_1b),\\
\ex x_1 \ex x_2\, (fax_1 \land fx_1x_2 \land fx_2b),\\
\ldots
\end{array}
\end{equation}
That is, the induction axiom is equivalent to the disjunction of the
infinite number of these sentences.  
With the notation $f^0xy \equi x=y$, $f^1xy
\equi fxy$, $f^2xy \equi \ex z_1\, (fxz_1 \land fz_1y)$, etc., the resultant
of the elimination can be written as
\begin{equation}
\ex n\, f^nab,
\end{equation}
where the domain of variable $n$ is the set of natural numbers.  The
originally given sentence has then be brought into a form that is free from
$\varphi$, contains on the one hand the numeric variable $n$ as exponent, but,
on the other hand, is very transparent.

Behmann proceeds that the question is now, whether this is not also possible
in the general case, that is, for 
the formulas in (\ref{eq-beh-prop-sequence}).  He says
that Ackermann really succeeded there, giving a schema that generally
characterizes the totality of the partial resultants from which the total
resultant is disjunctively (or in the case $\ex \varphi \ldots$
conjunctively) composed.  However, as Behmann remarks, on the basis of a
\emph{recursion} that is not quite easy to see through and thus lets the
result appear somewhat intransparent.  Refraining from summarizing Ackermann's
result in its original form, Behmann develops a presentation in a normalized
form, as continuation of his own earlier work, starting from the normalized
form~(\ref{eq-ebe-problem}), considered dually as:
\begin{equation}
\label{eq-ebe-problem-disj}
\begin{array}{ll}
\all \varphi \ex x \ex y\,
      ((fxy \land \varphi x \land \varphi y) & \lor\\
    \indebe \hparen
      (gxy \land \varphi x \land \lnot \varphi y) & \lor\\
    \indebe \hparen
      (hxy \land \lnot \varphi x \land \varphi y) & \lor\\
     \indebe \hparen
      (kxy \land \lnot \varphi x \land \lnot \varphi y)),
\end{array}
\end{equation}
where $fxy$, $gxy$, $hxy$ and $kxy$ abbreviate
$F[\true,\true,x,y]$, $F[\true,\false,x,y]$, $F[\false,\true,x,y]$ and
$F[\false,\false,x,y]$, respectively.
The properties that can be assumed are symmetry of $f$ and $k$ as well as that
$g$ and $h$ are inverse to each other and are reflexive (in the original:
\de{totalreflexiv}), thus containing equality as sub-relation.
As before (see p.~\pageref{page-note-drop-h}), pencil annotations in
\mref{man:beh:34:k8} indicate a second way of reading, where the component $(hxy
\lor \lnot \varphi x \lor \varphi y)$ is dropped.
In \mref{man:beh:34:k8} Behmann claims that his representation has for the
shown case of two individual variables the same (\de{sachliche}) generality as
Ackermann's, who -- Behmann seems to refer here to Ackermann's letter
\cref{corr:ab:1934:10:29} -- intentionally refrains from assuming that the
given formula is normalized in the described way. He views the difference as
only technical, where the preparatory work done by the normalization allows to
express the resultant in a discernible more transparent way.

Following Behmann, the application of Ackermann's general recursive definition
(\de{Rekursionsvorschrift}) to (\ref{eq-ebe-problem-disj}) yields a resultant
of the form:
\begin{equation}
\label{eq-beh-resolvents}
\begin{array}{rrlll}
& \ex w \ex x \ex y \ex z\, 
 & (gwx \land fxy \land hyz \land kzw\; \land\\
&& \hparen \lnot \bigwedge_{x_1 \in \{w,x\},\; x_2 \in \{y,z\}} x_1 \neq x_2)\\
\lor 
& \ex s \ex t \ex u \ex v 
   \ex w \ex x \ex y \ex z\,
 &  (gst \land ftu \land huv \land kvw\; \land\\
&& \hparen gwx \land fxy \land hyz \land kzw\; \land\\
&& \hparen \lnot \bigwedge_{x_1 \in \{s,t,w,x\},\; x_2 \in \{u,v,y,z\}} x_1 \neq
   x_2)\\
\lor & \multicolumn{2}{c}{\ldots}
\end{array}
\end{equation}
In \mref{man:beh:34:k8}, a second variant of~(\ref{eq-beh-resolvents})
without $h$ is indicated in pencil: $h$ is replaced there with $\tilde{g}$,
the inverse of $g$.\footnote{In the manuscript~\mref{man:beh:34:k8}, but not
  in the letter~\cref{corr:ba:1934:10:22}, the variables $z$ are universally
  quantified. Obviously a mistake in writing.}

Before we continue with following Behmann's presentation, we take a closer
look at his representation (\ref{eq-beh-resolvents}) of the resultant.  For
$\bigwedge_{x_1 \in \{w,x\},\; x_2 \in \{y,z\}} x_1 \neq x_2$ Behmann has the
dedicated notation $(wx,yz)$.  Expanding disequalities in the first disjunct
of (\ref{eq-beh-resolvents}) shows that it stands for the disjunction of the
following four (conjunctive) clauses:
\begin{equation}
\begin{array}{l@{\hspace{0.5em}}rll}
& \ex w \ex x \ex z\, & 
  (gwx \land fxw \land hwz \land kzw)\\
\lor & \ex w \ex x \ex y\, & 
 (gwx \land fxy \land hyw \land kww)\\
\lor & \ex w \ex x \ex z\, & 
 (gwx \land fxx \land hxz \land kzw)\\
\lor & \ex w \ex x \ex y\, & 
 (gwx \land fxy \land hyx \land kxw).
\end{array}
\end{equation}
It is easy to see that some (conjunctive) clauses that do not contain
$\varphi$ can be obtained by (dual) resolution from (\ref{eq-ebe-problem-disj}),
but are not subsumed by any disjunct of (\ref{eq-beh-resolvents}). For example
\begin{equation}
\label{eq-behmann-how-represented}
\ex x\, (fxx \land kxx).
\end{equation}
Thus, it seems that reflexivity of $g$ and $h$ needs to be considered such
that also disjuncts obtained by removing one or more atoms with predicates $g$
and $h$ while unifying the left and right argument of each removed atom are
considered as implicitly represented by (\ref{eq-beh-resolvents}).\footnotemark

\footnotetext{Ackermann \cite{ackermann:35} 
\label{footnote-ackermann-21} provides a precise
  characterization of the clauses that can be obtained by resolution for his
  Example~(21), a generalization of his Example~(20), shown above as
  (\ref{ex-ackermann-first}). Example~(21) can be written dually as
  (\ref{eq-ebe-problem-disj}) with the two disjuncts containing $g$ and $h$,
  respectively, dropped:
\begin{equation}
\label{eq-ebe-problem-ackermann-21}
\begin{array}{ll}
\all \varphi \ex x \ex y\,
      ((fxy \land \varphi x \land \varphi y) & \lor\\
     \indebe \hparen
      (kxy \land \lnot \varphi x \land \lnot \varphi y)).
\end{array}
\end{equation}
Thus, all (conjunctive) clauses that can be obtained by (dual) resolution from
(\ref{eq-ebe-problem-ackermann-21}) must also be obtainable from
(\ref{eq-ebe-problem-disj}). That is, all clauses satisfying Ackermann's
characterization for his Example~(21) should also be (in dual form)
represented in Behmann's presentation of the resultant.  The (conjunctive)
clause~(\ref{eq-behmann-how-represented}) is an example.
}

Behmann gives a second characterization of the resultant
(\ref{eq-beh-resolvents}) as graph:
\enq{The resultant now means the following: If we consider the 4 relations
$f,g,h,k$ as arrow schemas (with different colors) inscribed into the same
figure, then, according to the resultant, there is at least once a closed
chain of arrows such that a $g$-arrow, an $f$-arrow, an $h$-arrow and a
$k$-arrow follow once or a finite number of times cyclically in sequence,
where, however, at least once in the chain a (starting or ending) point of a
$g$-arrow coincides with a point of an $h$-arrow. I.e., there is at least one
closed 8-course because of the cycle $g,f,h,k$ with one or several rounds.}

\enq{\emph{If the symmetry conditions are omitted}, that is, $f,g,h,k$ are not
  submitted to any constraints, then this means for the chain of arrows
  stipulated by the resultant that the arrows from $f,g,h,k$ may independently
  from each other be represented by $\tilde{f},\tilde{h},\tilde{g},\tilde{k}$,
  and, moreover, the arrows from $g$ and $h$ may (again independently) be
  represented by identity (that is, by circular arrows), i.e. omitted in the
  chain.}\footnotemark

\footnotetext{%
\de{Die Resultante besagt nun folgendes: Denken wir uns die 4
Beziehungen $f,g,h,k$ als Pfeilschemata (mit verschiedenen
Farben) in dieselbe Figur eingetragen, so gibt es gemäß der
Resultante mindestens einmal eine geschlossene Kette von
Pfeilen derart, daß ein $g$-Pfeil, ein $f$-Pfeil, ein
$h$-Pfeil und ein $k$-Pfeil einmal oder endlich oft
hintereinander zyklisch folgen, wobei aber mindestens einmal
in der Kette ein (Anfangs- oder End-)Punkt eines $g$-Pfeils
mit einem Punkt eines $h$-Pfeils zusammenfällt. D.h. es
gibt mindestens eine geschlossene 8-Bahn auf Grund des
Zyklus $g,f,h,k$ mit einer oder mehreren Runden.}

\de{\emph{Verzichtet man auf die Symmetriebedingungen}, 
  unterwirft man also $f,g,h,k$
  keiner Beschränkung, so besagt dies für die durch die Resultante geforderte
  Pfeilkette, daß die Pfeile aus $f,g,h,k$ unabhängig von einander durch
  $\tilde{f},\tilde{h},\tilde{g},\tilde{k}$ vertreten sein dürfen und
  obendrein die Pfeile aus $g$ und $h$ (wiederum unabhängig) durch die
  Identität (also durch Rückkehrpfeile) vertreten sein, d.h. in der Kette
  wegbleiben, dürfen.}  

From \mref{man:beh:34:k8}, p.~18. Similarly in the letter~\cref{corr:ba:1934:10:22}.
The alternate reading of $h$ as inverse of $g$ is indicated in pencil at two
places in the first paragraph: \de{Beziehungen $f,g,k$} and \de{$g$-Pfeil, ein
  $f$-Pfeil, ein $\tilde{g}$-Pfeil und ein $k$-Pfeil}.  In the second
paragraph, \de{d.h. in der Kette wegbleiben} has been added in pencil.  In the
version from the letter \cref{corr:ba:1934:10:22}, the corresponding text ends
with \de{durch die Identität vertreten sein, also ganz ausfallen dürfen.}}

In the letter~\cref{corr:ba:1934:10:22} to Ackermann, Behmann continues to
explain his description and shows for Example~(15) from
\cite{ackermann:35} the resultant in his representation.\footnotemark\ In our
notation this is:
\begin{equation}
\begin{array}{lll}
\all x
  \all y
  \all z
  \all u
  \all v
  \all p
&  (F[\false,\false,x,y] & \lor\\
&  \hparen F[\false,\true,z,y] & \lor\\
&  \hparen F[\true,\true,z,v)] & \lor\\
&  \hparen F[\false,\true,v,p] & \lor\\
&  \hparen F[\false,\false,p,v] & \lor\\
&  \hparen F[\false,\true,u,v] & \lor\\
&  \hparen F[\true,\true,x,u]),
\end{array}
\end{equation}
where $F[G,H,x_1,x_2]$ denotes $F[\varphi x,\varphi y,x,y]$ with all
occurrences of $\varphi x$ replaced by $G$, of $\varphi y$ by $H$, of $x$ by
$x_1$ and of $y$ by $x_2$.

\footnotetext{%
\de{Schreibe ich für $F_{\behtrue\behtrue xy}$ kurz
$(\behtrue\behtrue)$ usf., so besagt die erste Bedingung,
daß die Pfeile stets so aneinander zu fügen sind, daß für
irgend zwei benachbarte Symbole $(pq)$ und $(rs)$ die
Aussagewerte $p$ $q$ und $r$ entgegengesetzt sind, und die
zweite Bedingung, daß die Kette derart zu einer $8$
zusammengebogen ist, daß jeder der beiden Bogen der $8$ für
sich genommen von zwei gleichartigen Symbolen begrenzt
ist. Man sieht deutlich, wie durch die schärfere Normierung
der Aufgabe die Durchsichtigkeit und anschauliche
Erfaßbarkeit der Resultante wesentlich erhöht wird. Durch
das Mitführen der Aussagewertzeichen $\behtrue$ und
$\behfalse$ wird der Sinn des Hoch- und Tiefstellens der
Argumente aufgeklärt und dieses zugleich entbehrlich
gemacht. So lautet Ihr Beispiel~(15):
\[xyzuvp(\behfalse\behfalse xy)
        (\behfalse\behtrue zy)
        (\behtrue\behtrue zv)
        (\behfalse\behtrue vp)
        (\behfalse\behfalse pv)
        (\behfalse\behtrue uv)
        (\behtrue\behtrue xu),\]
schematisch: 
$(\behfalse\behfalse)
(\behfalse\behtrue)
(\behtrue\behtrue)?(\behfalse\behtrue)
(\behfalse\behfalse)?(\behfalse\behtrue)
(\behtrue\behtrue)$,
wo die Unterstreichung den Ersatz von
$h$ durch $\tilde{g}$, d.h. von $F_{\behtrue\behfalse xy}$
durch $F_{\behfalse\behtrue yx}$ und die Zeichen $?$ die zu
identifizierenden Variablen andeuten.}
From \cref{corr:ba:1934:10:22}, p.~2.}

In the manuscript \mref{man:beh:34:k8}, Behmann suggest that the cycle $g,f,h,k$
can be better understood if it is already considered in the underlying problem
specification, by writing (\ref{eq-ebe-problem-disj}) as:
\begin{equation}
\begin{array}{ll}
\all \varphi \ex x \ex y\,
      ((gxy \land \varphi x \land \lnot \varphi y) & \lor\\
    \indebe \hparen
      (fxy \land \varphi x \land \varphi y) & \lor\\
    \indebe \hparen
      (hxy \land \lnot \varphi x \land \varphi y) & \lor\\
     \indebe \hparen
      (kxy \land \lnot \varphi x \land \lnot \varphi y)),
\end{array}
\end{equation}
It can then be seen that in the cycle
   \[(\varphi x, \lnot \varphi y),
     (\varphi x, \varphi y), (\lnot \varphi x, \varphi y), (\lnot
   \varphi x, \lnot \varphi y),\] 
adjacent components of neighboring pairs have complementary signs.

Behmann recalls in \mref{man:beh:34:k8} that his presentation can not be found
in that form in Ackermann's paper \cite{ackermann:35}, but -- seemingly
referring to letters \cref{corr:ba:1934:10:22} and \cref{corr:ab:1934:10:29}
-- that he was told by Ackermann upon request that he had actually been aware
of it, but did not show it explicitly just for the reason that he did not
succeed in finding an analogous representation for the case of three or more
individual quantifiers.

In his letter \cref{corr:ab:1934:10:29}, Ackermann writes that he was well
aware that for the case of two universal quantifiers a complete overview on
the resultants can be obtained and that this result indeed formed the basis of
his work. The first outcome that he found was the resultant of
\begin{equation}
\label{eq-ackermann-corr-1}
\ex f \all x \all y\, ((axy \lor fx \lor fy) \land (dxy
\lor \lnot fx \lor \lnot fy))\footnotemark
\end{equation}
in the clear form (\de{in der übersichtlichen Form}).  (Actually this is
Example~(21) from \cite{ackermann:35}, shown -- up to different predicate
names -- already as (\ref{eq-ebe-problem-ackermann-21}) in
footnote~\ref{footnote-ackermann-21} on
p.~\pageref{footnote-ackermann-21}.)
Ackermann then found that the resultant of
\begin{equation}
\label{eq-ackermann-corr-2}
\begin{array}{lll}
\ex f \all x \all y
& ((axy \lor fx \lor fy) & \land\\ 
& \hparen (bxy \lor fx \lor \lnot fy) & \land\\
& \hparen (dxy \lor \lnot fx \lor \lnot fy))
\end{array}
\end{equation}
is not much different from the resultant of (\ref{eq-ackermann-corr-1}); the
cycles and chains have just to be extended in a suitable way by inserting $b$
and $\tilde{b}$. The general case with two variables can be brought into form
(\ref{eq-ackermann-corr-2}), which has also been discussed in Behmann's letter
-- as Ackermann remarks, the conjunct $(cxy \lor \lnot fx \lor fy)$ considered
in addition by Behmann can be united with the middle component (see also
p.~\pageref{page-note-drop-h}).  Passing on to more universal quantifiers,
Ackermann found that a corresponding normalization (\de{Normierung}) can be
easily achieved.  For example, for three universal quantifiers it is
\begin{equation}
\label{eq-ackermann-corr-3}
\begin{array}{lll}
\ex f \all x \all y \all z
& ((axyz \lor fx \lor fy \lor fz) & \land\\ 
& \hparen (bxyz \lor fx \lor fy \lor \lnot fz) & \land\\
& \hparen (cxyz \lor fx \lor \lnot fy \lor \lnot fz) & \land\\
& \hparen (dxyz \lor \lnot fx \lor \lnot fy \lor \lnot fz)),
\end{array}
\end{equation}
where certain symmetry conditions hold for $a,b,c,d$.
However, this normalization did not help Ackermann in getting an overview on
the resultants at three and more universal quantifiers, such that he had
dropped it again and had to confine himself in the general case to the
recursion method for forming the resultant.
To prevent that the gained overview for two universal quantifiers would go
completely by the board, Ackermann included Example~(20)
\cite[p.~410]{ackermann:35}, which admits a simple interpretation of
meaning. At the discussion of the example, the clear resultant
for~(\ref{eq-ackermann-corr-1}) is given \cite[p.~411]{ackermann:35}.

It seems that the apparently distinguished features of elimination of an
existential predicate quantifier upon a first-order formula with \emph{at most
  two} universal individual quantifiers have so far not got attention beyond
the mentioned correspondence and examples in \cite{ackermann:35}, and thus
might be of interest for further research.

\footnotetext{We use here the predicate names from Ackermann's letter, which
  differ from that used by Behmann, but the correspondence is easy to see.
  Symbolic notation and capitalization are coherent with the rest of this
  paper, different from that used by Ackermann.}

Manuscript \mref{man:beh:34:k8} concludes with relating
Ackermann's result to the general decision problem: \enq{Actually, we just can
  say that \emph{under favorable conditions}, that is, in so far as the
  assumed preconditions on individual variables are satisfied, \emph{the
    innermost eliminations} can be performed, but not yet the further ones,
  since we do have the first obtained resultants \emph{not in closed form}.
  Thus, there are two issues to solve: 1. the liberation from the
  \emph{condition for the individual quantifiers} and 2. the representation of
    the respective resultant in a \emph{closed symbolic representation}
    that is suitable for further eliminations.}
Behmann suggests that the first issue can be addressed by Skolemization, with
the difficulties already described in Sect.~\ref{sec-skolemization}.  On the
second issue, he remarks that in the case where all predicate quantifiers are
universal (or existential) and at the front of the sentence, the problem can
be considered as question of validity of a sentence that is free of predicate
quantifiers.

\part{Register of Publications by Behmann and Documents in his Bequest that
are Relevant to Second-Order Quantifier Elimination}
\label{part-app-sources}

\section{Introduction to Part~\ref{part-app-sources}}

This part provides commented lists of Behmann's publications and unpublished
material from his bequest \cite{beh:nl}, such as manuscripts and
correspondences, as far as they are immediately relevant to the problem of
second-order quantifier elimination.

Accounts of contributions by Behmann in a historic context have been given in
Church's book \name{Introduction to Mathematical Logic} \cite{church:book},
papers by Mancosu \cite{mancosu:behmann:99} and Zach
\cite{zach:99:completeness}, a more recent presentation by Zach
\cite{zach:2007} and the scholarly edition \cite{mancosu:zach:2015} of
Behmann's 1921 talk on the \de{Entscheidungsproblem}.  In Craig's paper on the
history on elimination problems in logic \cite{craig:2008}, Behmann's work
\cite{beh:22} is briefly mentioned.

A recent English biography of Behmann can be found in
\cite{mancosu:zach:2015}.  The most comprehensive publication of biographic
material is in a dedicated chapter in \cite[pp.~105--170]{schenk:2002} (in
German), which also contains a selection of texts by Behmann.\footnotemark\ 
He is positioned between Cantor and Husserl in this
compilation about philosophical thinking in Halle (Saale), where he was
professor of mathematics from 1925 to 1945.
The investigation \cite{eberle:2002} (in German) of the
\de{Martin-Luther-Universität Halle-Wittenberg} during Nazism includes a short
biography of Behmann, but no further references to him.  Behmann is
peripherally mentioned in the treatise about logics in Nazi Germany in the
context of Gentzen's life \cite{menzler:genzen} (English translation:
\cite{menzler:genzen:english}). Behmann's personal file \cite{beh:auh} at
\de{Martin-Luther-Universität Halle-Wittenberg} is preserved in the
university's archive in Halle. Excerpts from the personal file, which include
an 8-page typescript dated 20 October 1945 by Behmann about his activities in
the NSDAP (\dename{Zusammenstellung der für die Beurteilung wesentlichen
  Tatsachen}), are published in \cite{schenk:2002}.

Behmann's scientific bequest \cite{beh:nl} is located in the department for
autographs of \dename{Staatsbibliothek zu Berlin}.  It has been registered by
Peter Bernhard and Christian Thiel \cite{nlv}, after a first registration by
Gerrit Haas and Elke Stemmler \cite{nlv:haas:stemmler}.  So far, it seems that
there are only three publications of material from the bequest: The
aforementioned talk on the \de{Entscheidungsproblem} from 1921
\cite{mancosu:zach:2015}, the correspondence with Gödel (from 1931) with
English translation and an introduction by Charles Parsons
\cite[pp.~12--39]{beh:goedel:correspondence}, and, from \cite{beh:auh}, a
report about Behmann's participation at the 1937 Congress for the Unity of
Science in Paris \cite[p.~105--108]{schenk:2002}.

\footnotetext{There are some errata in \cite{schenk:2002}: P.~109: The birth
  name of Behmann's mother is \name{Knübel} (not \name{Kübel}). P.~111 and
  123: The quoted letter by Runge is dated 28 February 1926 (not November).}

\section{Publications by Behmann Related to Second-Order Quantifier Elimination}
\label{sec-publications}

This compilation lists the publications by Behmann with immediate relevance to
second-order quantifier elimination. These are the Habilitation thesis
\cite{beh:22} with some related documentary publications and abstracts
surrounding it, as well as the abstract \cite{beh:26:beziehungen} of Behmann's
1926 talk on the decision problem for relations. In addition, a later work
\cite{beh:50:aufloesungs:phil:1,beh:51:aufloesungs:phil:2} on the solution
problem (\dename{Auflösungsproblem}), where techniques from \cite{beh:22} are
applied, is listed.  Also two further major publications
\cite{schoenfinkel:24:bausteine,beh:27} by Behmann, or with involvement of
Behmann, respectively, are included, because they fall in the time span
between 1921 and 1927, although they are not directly concerned with
elimination or the decision problem.

As already noted on p.~\pageref{page-paradoxes}, Behmann's extensive
investigations of the paradoxes are not considered here, although one may
speculate whether his underlying idea that paradoxes emerge from unjustified
elimination of shorthands is somehow related to elimination of second-order
quantifiers.

\begin{enumerate}[leftmargin=1.8cm,itemsep=5pt,listparindent=\parindent]

\item[\cite{beh:jahresbericht:30:2:47}] (1921) The annual report of the
  \dename{Deutsche Mathematiker-Vereinigung} shows in a listing of the
  activities of the \dename{Mathematische Gesellschaft in Göttingen} that
  Behmann gave on 10 May 1921 at the \dename{Mathematische Gesellschaft} a
  talk \de{Das Entscheidungsproblem der mathematischen Logik}, which can be
  translated as \name{The Decision Problem of Mathematical Logic}. As noted in
  \cite[p.~363]{zach:99:completeness}, this seems the first documented use of
  \de{Entscheidungsproblem}.  The manuscript~\mref{man:beh:21:ms:k9:37}
  \de{Entscheidungsproblem und Algebra der Logik}, published recently in
  \cite{mancosu:zach:2015}, bears the date of the talk and thus seems to
  underlie it.

\item[\cite{beh:22}] (1922, received by the journal on 16 July 1921) Heinrich
  Behmann: \de{Beiträge zur Algebra der Logik, insbesondere zum
    Entscheidungsproblem}. This is the published version of the thesis for
  Behmann's Habilitation at Göttingen on 9 July 1921.  Some corrections of
  printing errors have been published as \cite{beh:22:corrections}.  A carbon
  copy of the thesis typescript with handwritten corrections that have been
  considered for the printed version is preserved as \mref{man:beh:21:carbon}.
  Document~\mref{man:beh:22:offprint} is an author's offprint of the published
  version with later handwritten corrections.

\item[\cite{beh:22:corrections}] (1922) \de{Druckfehlerberichtigung zu dem
  Aufsatz von H. Behmann \glqq Beiträge zur Algebra der Logik, insbesondere
  zum Entscheidungsproblem\grqq\ in Band 96, S.163--239}. 1922.  Corrections
  of nine wrongly printed symbols in \cite{beh:22}.

\item[\cite{beh:jahresbericht:32:2:66f}] (1923) Heinrich Behmann: \de{Algebra
  der Logik und Entscheidungsproblem}.  Abstract of a talk given on 21
  September 1923 at the \de{Jahresversammlung der Deutschen
    Mathematiker-Vereinigung} in Marburg a. d. Lahn.  A manuscript for the
  abstract is preserved as \mref{man:beh:jahresbericht:32:2:66f}.  The talk
  seems in essence a summary of the results published in \cite{beh:22}.  In
  the subsequent discussion, L.~E.~J.\ Brouwer (who gave a talk in the same
  session) expressed concerns about the phrase \de{usw.} (\name{and so on})
  and the concept of finite number within the presented theory. These concerns
  have been rebutted by remarks that the debated concepts do not play a role
  on their own in the considered statements but can there be reduced
  unobjectionable to basic concepts and that the theory about these statements
  is just a theory within mathematics, not intended as foundation of
  mathematics.

\item[\cite{schoenfinkel:24:bausteine}] (1924, received by the journal on 15
  March 1924) Moses Schönfinkel: \de{Über die Bausteine der mathematischen
    Logik}. In this paper logic combinators have been introduced. As noted in
  the paper, it is based on a talk by Schönfinkel given in 1920 at the
  \de{Mathematische Gesellschaft in Göttingen} and has been prepared for
  publication and supplemented by Behmann.  William Craig asks Behmann in a
  letter dated 16~March 1952, upon advice of Haskell Curry, about an error in
  the supplementary part by Behmann. In his reply dated 7~April 1952, Behmann
  remarks that he did the preparation for publication on behalf of Hilbert,
  without a particularly strong interest and that he considers the new
  direction that emerged from that work as too formalistic. He was aware of
  the error, which had been pointed out to him in 1928 by Alfred Boskovitz
  \cite[Kasten~1, I~13 and I~14]{beh:nl}.\footnotemark\ See also \cite[p.~8,
    p.~184]{curry:combinatoric:1}.

  \footnotetext{Alfred Boskovitz \label{footnote-boskovitz} was a student in
    Göttingen when Behmann was lecturer. In the mid 1920s Boskovitz moved back
    to Budapest, his home town. He carefully reviewed and extended the
    \name{Prinicpia Mathematica}, where he is mentioned in the second edition
    in an acknowledging footnote together with Behmann \cite[p.~xiii]{pm}.
    Respective manuscripts by Boskovitz are in \cite[Kasten~1,
      I~08]{beh:nl}. Behmann uses the opportunity of Curry's request for
    information about Boskovitz \cite[Kasten~1, I~08]{beh:nl} on 12~July 1957
    to draw attention to Boskovitz's work and sends to Curry on 13~August 1957
    \cite[Kasten~1, I~08]{beh:nl} a characterization of Boskovitz as well as
    typed transcripts of a selection of his letters until 1937. Those letters
    and transcripts are now in \cite[Kasten~1, I~08]{beh:nl}.  On 21 June 1936
    Boskovitz writes to Behmann that he expects danger of life and asks
    Behmann to store his mathematical works.  In his letter to Curry, Behmann
    states that he did not have heard from Boskovitz since 1937 and gives to
    Curry the address in Budapest from that time.  However, in \cite[Kasten~1,
      I~08]{beh:nl} there is also a short letter from 1939 by Boskovitz as
    well as a postcard dated 11 November 1942 with a different address in
    Budapest.}
   
\item[\cite{beh:26:beziehungen}] (1927). Heinrich Behmann:
  \de{Entscheidungsproblem und Logik der Beziehungen}.  Abstract of a talk
  given on 23 September 1926 at the \de{Jahresversammlung der Deutschen
    Mathematiker-Vereinigung} in Düsseldorf.  This abstract motivated
  Ackermann to write to Behmann in 1928 \cref{corr:ab:1928:8:16}, which
  initiated their correspondence on elimination for relations.  A draft
  manuscript is preserved as \mref{man:beh:26:beziehungen}.

\item[\cite{beh:27}] (1927) Heinrich Behmann: \dename{Mathematik und Logik}.
  A small introductory textbook on mathematical logic, showing Gottlob Frege
  on the cover. The twofold connection between mathematics and logics is
  emphasized: Mathematical representation and notation allows to make logic
  possible as an ``exact'' science, like mathematics, and, on the other hand,
  the insight that, conversely, pure mathematics is nothing else than logic in
  disguise.

\item[\mbox{\begin{minipage}[b]{1.1cm}
\hfill\cite{beh:50:aufloesungs:phil:1}\\
\hspace*{\fill}\cite{beh:51:aufloesungs:phil:2}\\[-5.6ex]
\end{minipage}}] 
(1950/51). Heinrich Behmann: \dename{Das Auflösungsproblem in der
  Klassenlogik}.  Behmann develops an approach to solve the ``logic solution
  problem'' (\dename{logisches Auflösungsproblem}).  See
  \cite{cw:boolean:frocos,cw:boolean:report} for a modern technical account on
  the solution problem on the basis of predicate logic.  A comprehensive
  presentation in the context of modern Boolean algebra that includes material
  by Schröder and Löwenheim as well as historical notes is provided in
  \cite{rudeanu:74}.  In a somewhat different phrasing but similar in spirit
  than as presented by Behmann, the problem can be described as follows: Given
  is a formula where the set of predicates occurring in it is partitioned into
  two disjoint subsets. The objective of the solution problem is to find a
  representation of the relation of predicate valuations from the first subset
  to those valuations from the second subset for which the formula is true. In
  addition, conditions on the predicates from the first subset are sought,
  under which solution valuations for those from the second subset exist at
  all.

  \lskip
  Behmann relates the solution problem to the elimination problem, discusses
  earlier works by Schröder, Jevons and Boole and extends Boole's approach to
  \MONE.  Behmann's method is based on a normalization: Assume that the
  predicates under consideration are partitioned into $\{p_1, \ldots, p_n\}$
  and $\{p\}$.  A \MON formula without individual constants over these
  predicates can be converted into a disjunction of conjunctions in which each
  conjunct has one of the following basic forms:
\begin{equation}
\label{eq-basic-aufloesungsproblem}
\begin{array}{l@{\hspace{0.5em}}l}
 \text{(a)} &
   \lnot \ex x\, (L_1[x] \land \ldots \land L_n[x] \land px),\\
 \text{(b)} &
   \lnot \ex x\, (L_1[x] \land \ldots \land L_n[x] \land \lnot px),\\
 \text{(c)} &
   \ex x\, (L_1[x] \land \ldots \land L_n[x] \land px),\\
 \text{(d)} &
   \ex x\, (L_1[x] \land \ldots \land L_n[x] \land \lnot px),
\end{array}
\end{equation}
where each $L_i[x]$ is either $p_ix$ or $\lnot p_ix$.  This form can be
obtained with the techniques used in \cite{beh:22} for elimination, followed
by conversion to the ``fully developed'' form, where each basic formula
contains exactly one literal with each predicate in $p_1, \ldots, p_n, p$. The
fully developed form can be achieved by rewriting with
\begin{equation}
\lnot \ex x\, F[x]
\; \equiv\; \lnot \ex x\, (F[x] \land qx)
        \land
        \lnot \ex x\, (F[x] \land \lnot qx)
\end{equation}
and
\begin{equation}
\ex x\, F[x] 
\; \equiv\; \ex x\, (F[x] \land qx)
        \lor
        \ex x\, (F[x] \land \lnot qx).
\end{equation}
Without loss of generality, it can be assumed that the conjunctions of
formulas of forms in (\ref{eq-basic-aufloesungsproblem}) do not contain
contradicting conjuncts. Such a conjunction then corresponds to a solution
set, represented by the mapping from the the $2^n$ formulas $L_1[x] \land
\ldots \land L_n[x]$, where each $L_i[x]$ is either $p_ix$ or $\lnot p_ix$ to
one of nine values depending on which of the forms (a)--(d) containing $L_1[x]
\land \ldots \land L_n[x]$ are present. For given $L_1[x] \land \ldots \land
L_n[x]$, there are nine such combinations of (a)--(d) whose conjunction is not
contradicting. The total solution is a set of such mappings, one for each
disjunct of the normalized formula.

\end{enumerate}

\sectionmark{Manuscripts and Other Archival Documents}
\section{Manuscripts and Other Archival Documents by Behmann Related to Second-Order Quantifier Elimination}
\label{sec-manuscripts}
\sectionmark{Manuscripts and Other Archival Documents}

This section provides a commented listing of the manuscripts and other
archival documents in Behmann's scientific bequest \cite{beh:nl} that are
related to second-order quantifier elimination, in chronological order.  Of
these, the manuscripts \mref{man:beh:21:ms:k9:37}, \mref{man:beh:23:ms:k9:45},
and \mref{man:beh:26:beziehungen:incomplete} were in the bequest originally
in the so-called ``brown box'' (\de{braune Box}), a cuboid cardboard folding
box covered with brown-black marbled paper with a handwritten note by Behmann
\name{``Important records for own lectures in Göttingen and Halle!
  H.~Beh\-mann''} (\deq{Wichtige Aufzeichnungen für eigene Vorlesungen in
  Göttingen und Halle!  H.~Behmann}) \cite[p.2, p.~78]{nlv} which included
what is now registered as Kasten~9, Einheit~1--45 and Kasten~10,
Einheit~46--91 of \cite{beh:nl}.

Aside of the individual manuscripts listed below, further documents that
concern the elimination and decision problem can be found in
\cite[Kasten~11]{beh:nl} which contains among other documents about 200 pages
of notes that are not registered in \cite{nlv}.  Most of these notes are
inscribed almost purely with formulas, occasionally with graphical
visualizations.  Some of them evidently concern the variant of Ackermann's
resolution-based elimination method as presented in 1934 by Behmann in
manuscript \mref{man:beh:34:k8} and letter \cref{corr:ba:1934:10:22}, which is
summarized in Sect.~\ref{sec-corr-elim-1934}. For others it may be conjectured
that they relate to the elimination and decision problem for relations and to
Skolemization.  A particular ordering of the notes is not immediate.  Some of
them bear explicit dates, including particular days in February, March and
April 1924, July 1926, March 1928 and August 1934 (some are just dated
\name{August 1934}, without specification of a day). The respective folders in
\cite[Kasten~11]{beh:nl} where they are contained are
\name{hsl. Aufzeichnungen mathematischer u. physikalischer Art} and
\name{Logik~I, Logik~II hsl. Aufzeichnungen}.

\begin{enumerate}[label=M\arabic*,ref=M\arabic*,leftmargin=0.81cm,itemsep=5pt,listparindent=\parindent]

\item \label{man:beh:21:ms:k9:37} \dename{Entscheidungsproblem und Algebra der
  Logik}. 1921.  In \cite[Kasten~9, Einheit 37]{beh:nl}, see
  \cite[p.~88]{nlv}.  Handwritten manuscript. 17~numbered pages.  Dated 10~May
  1921. Written in ink with red underlines and side notes, corrections and
  addenda in pencil, some in shorthand.  Starts after the title with \deq{Von
    Kronecker stammt, soviel ich weiß, ein Distichon, das ungefähr so
    lautet:}.

  \lskip
  Summarizes the material in \cite{beh:22}.  Seems to be the manuscript for
  the talk on 10 May 1921 at \de{Mathematische
    Gesellschaft in Göttingen} listed in \cite{beh:jahresbericht:30:2:47}
  as \de{Das Entscheidungsproblem der mathematischen Logik}.
  A transcript has been published along with an introduction and English
  translation as \cite{mancosu:zach:2015}.

\item \label{man:beh:22:axiomatik} \dename{Das Problem der Axiomatik vom
  Standpunkt der Algebra der Logik}.  1922. In \cite[Kasten~9, Einheit
  33]{beh:nl}, see \cite[p.~87]{nlv}.  Handwritten manu\-script. Dated 28
  November 1922. 4 pages on 1 folded double sheet.  Starts after the title
  with \deq{\emph{Bemerkungen} zu dem Vortrag von H. \emph{Neues und
      bemerkenswertes Problem} aufgeworfen. Allerdings recht speziell.}

\item \label{man:beh:21:carbon} \dename{Beiträge zur Algebra der Logik,
  insbesondere zum Entscheidungsproblem}. 1921. In \cite[Kasten~7,
  1922]{beh:nl}, see \cite[p.~61]{nlv}. Typescript (carbon copy) of the thesis
  of Behmann's Habilitation at Universität Göttingen with handwritten
  mathematical symbols and corrections that have been considered in the
  published version \cite{beh:22}. 79 numbered pages plus two initial pages
  for title and table of contents.

\item \label{man:beh:22:offprint} Author's offprint of \cite{beh:22}. 1922.  In
  \cite[Kasten~6, 1921.1]{beh:nl}, see \cite[p.~110]{nlv}.  With handwritten
  notes and corrections.  On top of the title page a handwritten dedication in
  English \textit{``With the author's compliments.''} and fragments of a
  postmark.

\item \label{man:beh:21:ms:k6:v23} \dename{Algebra der Logik und
  Entscheidungsproblem}. 1923?. In \cite[Kasten~6, V~23]{beh:nl}, see
  \cite[p.~53, also p.~44]{nlv}.  Handwritten manuscript. 39~pages in a
  checkered notebook containing also other texts on other topics.

  \lskip
  The notebook also contains various other long and short texts, some
  seemingly related to lectures given by Behmann, such as \de{Darstellende
    Geometrie}. Further texts include drafts for the letter dated 10 May 1923
  to Bertrand Russell \cite[Kasten~2, I~60]{beh:nl} and an excerpt from the
  article on \name{Calculating Machines} in \name{Encyclopedia
    Britannica}.\footnotemark

\urldef\urlbrittcm\url{http://en.wikisource.org/wiki/1911_Encyclop\%C3\%A6dia_Britannica/Calculating_Machines}

\footnotetext{The excerpt corresponds in the online edition of the 1911
  Encyclopedia Britannica to the span from \name{``Machines of far greater
    powers''} to \name{``published by his son, General Babbage.''}. See
  \urlbrittcm\ (accessed 17 August 2015).
  Behmann also annotates the author of the article as \name{``Henrici''}, with
  \name{``1792--1871''} added in pencil. In fact, the author was Olaus
  Magnus Friedrich Erdmann Henrici (1840--1918).}

\lskip
The manuscript \dename{Algebra der Logik und Entscheidungsproblem} starts
 after the title with \deq{Vortragstext (stellenweise weiter ausgeführt) bis
   [bei?] S.~4}.  A paragraph header \deq{Aufgabe der Vorlesung} on the first
 page suggest that this text is a draft manuscript for a lecture.

\item \label{man:beh:23:ms:k9:45} \dename{Algebra der Logik und
  Entscheidungsproblem}. 1923/1924.  In \cite[Kasten~9, Einheit 45]{beh:nl},
  see \cite[p.~91]{nlv}.  Handwritten manuscript. 60 numbered sheets, with
  p.~6, p.~32--35, p.~47, p.~50, p.~57 inscribed also recto. One additional
  sheet inserted after the page~1. Ink with pencil additions.  Dated
  \deq{W. S. 1923--24}. Starts after the title with \deq{Kronecker
    (Mathematiker), Distichon:}

  \lskip
  Seems to be the manuscript for Behmann's lecture with the same title listed
  in the register of lectures at Göttingen university for winter semester
  1923/24 (\de{Verzeichnis der Vorlesungen der Universität Göttingen
    Winterhalbjahr 1923/24}).\footnotemark

\footnotetext{Behmann regularly gave lectures in mathematics, 1921--26 as
  doctor in Göttingen and 1926--45 as professor in Halle. Only very few of
  these were on logic or related topics: WS 1921/22 Set Theory, SS 1922
  Mathematical Logic (Also for Non-Mathematicians), WS 1923/24 Algebra of
  Logic and Decision Problem, SS 1927 Mathematics and Logic.}

  \lskip
  In \cite{nlv} it is conjectured that the untitled 30-page document
  \cite[Kasten~10, Einheit 90]{beh:nl} might be a continuation of
  \mref{man:beh:23:ms:k9:45}. However, consideration of three-valued logic and
  notation of existential quantifiers as $\ex$ in \cite[Kasten~10, Einheit
    90]{beh:nl} suggests that it stems from a later period.


\item \label{man:beh:jahresbericht:32:2:66f}
  \dename{Entscheidungsproblem und Algebra der Logik}. 1923.  In
  \cite[Kasten~7, 1923]{beh:nl}, see \cite[p.~61]{nlv}. Typescript (carbon
  copy) with handwritten corrections. 3 pages. Manuscript for the talk
  abstract \cite{beh:jahresbericht:32:2:66f}.

\item \label{man:beh:26:beziehungen:draft:1} \dename{Entscheidungsproblem für
  Beziehungen}. 1926?.  In \cite[Kasten~8, 1925/ 26.1]{beh:nl}, see
  \cite[p.~58]{nlv}.  Handwritten draft. 4~pages. Starts after the title with
  \deq{Wir haben als Bestandteil einer Aussage den folgenden:}.

\item \label{man:beh:26:beziehungen:draft:2} \de{Entscheidungsproblem und Logik
  der Beziehungen}.  1926?. In \cite[Kasten~8, 1925/26.2]{beh:nl}, see
  \cite[p.~58]{nlv}. Handwritten draft. 4~pages.  Starts after the title with
  \deq{Wer über ein Problem der mathematischen Logik vorzutragen gedenkt}.

  \lskip
  Seems to be a draft for the talk on 23 September 1926 in Düsseldorf, whose
  abstract has been published as \cite{beh:26:beziehungen}.  The
  document~\mref{man:beh:26:beziehungen:incomplete} seems to be a later but
  incomplete version.

\item \label{man:beh:26:beziehungen:incomplete} \dename{Entscheidungsproblem und
  Logik der Beziehungen}.  1926?. In \cite[Kasten~10, Einheit 46]{beh:nl}, see
  \cite[p.~91, also p.~57]{nlv}. Handwritten manuscript. 11~pages (9~numbered
  sheets with text on verso only, 1 sheet with formulas inscribed verso and
  recto).  Filed together with an additional seemingly unrelated double sheet
  that is inscribed on 3 pages.  Starts after the title with \deq{Wer die
    Absicht hat, über ein Thema aus dem Gebiet der mathematischen Logik
    vorzutragen}.

  \lskip
  Seems to be a later but incomplete version of
  \mref{man:beh:26:beziehungen:draft:2}. The text on page~9 ends abruptly,
  where \mref{man:beh:26:beziehungen:draft:2} continues after the matching
  position.  As noted in Sect.~\ref{sec-corr} about \cref{corr:ba:1928:9:29},
  Behmann has sent in 1928 the technical second part of an elaborate
  transcript of his talk on 23~September 1926 in Düsseldorf to Ackermann,
  which suggests that this incomplete manuscript might be the first part of
  this transcript. The later manuscript~\mref{man:beh:34:k8} also includes a
  presentation of Behmann's work on elimination for relations around 1926.

\item \label{man:beh:26:beziehungen} \dename{Entscheidungsproblem und
  Logik der Beziehungen}. 1926. In \cite[Kasten~8, VII~01--VII~06]{beh:nl}.
  Not registered in \cite{nlv} (see \cite[p.~61]{nlv}). The document is the
  last one in the aforementioned folder.  Handwritten draft manuscript of the
  talk abstract published as \cite{beh:26:beziehungen}.  2~pages on a 4~page
  double sheet, along with a draft of a short letter from Rome, dated 23
  October 1926, to accompany the submission of the manuscript and parts of an
  unrelated manuscript. The draft letter is addressed to a professor, probably
  Bieberbach as editor, but does not contain any discussion regarding content
  (compare p.~\pageref{bieberbach}).

  \lskip The following two portions of the draft manuscript do not appear in
  the published version: (1.) \deq{[\ldots Elimination in der Tat vollziehen
      läßt.]  So hat das erste die bekannte Lösung $x(\overline{p})F_{xp}$
    bzw. $x\,F_{\behfalse x}\, F_{\behtrue x}$, während z.B. für das zweite
    [XXX] die Lösung
   \[x(\overline{p})\, y(\overline{q})\, F_{pq\,xy}\, .\,
   y(\overline{q})\, x(\overline{p})\, F_{pq\,xy}\] ermittelt wurde. [So lange
     \ldots]}.  (Transcript with the symbolic notation in the manuscript.
  \name{[XXX]} marks a word that could not be identified.)  (2.) \deq{[\ldots
      auch verwickeltere Verkettungen] dieser Operationen, insbesondere
    Zyklen, die etwa nicht mehr der Transitivität des Früher oder Später
    genügen, [zuzulassen\ldots].} Behmann uses in this text \deq{Dingoperator}
  for \name{quantified variable}, \deq{Begriffsoperator} for \name{quantified
    predicate} and \deq{Operation} for \name{quantifier}. In the published
  version, the overline indicating existential quantification has been
  erroneously omitted at some occurrences of $\varphi$: two occurrences in the
  displayed row of formulas in the center of p.~17 and in the first line of
  the last paragraph of p.~17. This manuscript is further discussed in
  Sect.~\ref{sec-corr-elim-1928a} and \ref{sec-corr-elim-1928b}.

\item \label{man:beh:34:k8} \dename{Ein wichtiger Fortschritt im
  Entscheidungsproblem der Mathematischen Logik (Ackermann Math. Annalen 110
  S.390)}.  1934.  In \cite[Kasten~8, 1934]{beh:nl}, see \cite[p.~59]{nlv}.
  Handwritten manuscript. Dated 14 December 1934. 21 numbered pages.  Written
  in ink with some variants and corrections added in pencil.

  \lskip
  For a summary, see Sect.~\ref{sec-corr-elim-1934}.  The technical part
  concerning the elimination problem for relations overlaps with
  \cref{corr:ba:1934:10:22}. A few phrases could be interpreted as suggesting
  that the manuscript was intended as basis for a talk (on p.~17: \deq{Ich
    möchte darauf verzichten, [\ldots] in der von ihm gegebenen Allgemeinheit
    vorzutragen}, on p.~19: \deq{Sie werden nun genauer wissen wollen}).

\end{enumerate}

\section{The Correspondence between Heinrich Behmann and Wilhelm Ackermann}
\label{sec-corr}

The complete correspondence between Heinrich Behmann and Wilhelm Ackermann, as
far as preserved in \cite[Kasten~1, I~01]{beh:nl}, see \cite[p.~3]{nlv}, is
listed here with English abstracts. Some technical content beyond second-order
quantifier elimination is also summarized, however, the discussion of
Behmann's ideas for the resolution of the paradoxes and related issues such as
restricted (\dename{limitierte}) variables and ultrafinite logic is only
briefly indicated here, since it would form a major topic on its own.

The letters by Behmann archived in \cite{beh:nl} are handwritten copies or
carbon copies of typescripts, respectively.  Seemingly, there is an erratum in
\cite{nlv}: Instead of the two letters from Ackermann to Behmann dated 29
October 1934 \cref{corr:ab:1934:10:29} and 9 January 1953
\cref{corr:ab:1953:1:9}, a single letter dated \name{29 October 1953}, not
present in the bequest, is listed there.

Upon request by Christian Thiel, Ackermann's son Hans-Richard Ackermann writes
on 23 August 1981 \cite[Kasten~12]{beh:nl} that, as far as Behmann is
concerned, in a first sighting of his father's correspondence, he could only
find a carbon copy of the letter to Behmann from 29 October 1934
\cref{corr:ab:1934:10:29}, whose original is already present in Behmann's
bequest.  In the selection from the correspondence of Wilhelm Ackermann
\cite{ackermann:briefwechsel} published in 1983 by Hans-Richard Ackermann, the
correspondence with Paul Bernays concerns the elimination problem (see
p.~\pageref{page-ackermann-bernays}).

For a general introduction to the correspondence between Behmann and Ackermann
as well as technical summaries of the content regarding second-order
quantifier elimination see Part~\ref{part-polyadic}.

\begin{enumerate}[label=L\arabic*,ref=L\arabic*,leftmargin=0.81cm,itemsep=5pt,listparindent=\parindent]

\item \label{corr:ab:1928:8:16}
Ackermann to Behmann, Lüdenscheid, 16 August 1928, handwritten, 2~pa\-ges.

\lskip
  See Sect.~\ref{sec-corr-elim-1928a}.

\item \label{corr:ba:1928:8:21} Behmann to Ackermann, Nieblum auf Föhr, 21
  August 1928, handwritten, 2~pages.  In addition to answering
  \cref{corr:ab:1928:8:16}, this letter refers to an unidentified preceding
  mail from Ackermann with the copy of the book \cite{hilbert:ackermann:28}.

\lskip
See Sect.~\ref{sec-corr-elim-1928a} for the content regarding elimination.
Behmann wishes Ackermann success on this way, emphasizing the importance of
the decision problem for strengthening the reputation of symbolic logic in
wider circles, mentioning Heinrich Scholz, who had to experience combating
against symbolic logic with unpleasant means. In addition, Behmann thanks
Ackermann for transmitting a copy of his book \cite{hilbert:ackermann:28} and
indicates that he does not agree with all details, in particular regarding the
Hilbert-Bernays symbolism, but admires it as a whole.

\item \label{corr:ab:1928:9:1}
Ackermann to Behmann, Münster i. W., 1 September 1928, handwritten, 2~pages.

\lskip
See Sect.~\ref{sec-corr-elim-1928a} for the content regarding elimination.
Ackermann further reports that he just sent off the corrections for his paper
\cite{ackermann:28:zaehlausdruecke} (decidability of the class called today
\name{Ackermann class}), motivated by the work of Bernays and Schönfinkel
\cite{bernays:schoenfinkel:28}, who, in Ackermann's view, take too much effort
for simple special cases. Ackermann requests more information about the
details in his book \cite{hilbert:ackermann:28} with which Behmann disagrees
(see \cref{corr:ba:1928:8:21}). In a possible second edition some changes
would be made. Among other things, Ackermann intends to include a presentation
of the results of Behmann's Habilitation thesis \cite{beh:22}.

\item \label{corr:ba:1928:9:29} Behmann to Ackermann, Halle (Saale), 29
  September 1928, typescript, 2~pages. The letter originally enclosed two
  manuscripts by Behmann.

\lskip
Behmann writes that he had not found the time to compile his results on the
decision problem in a well arranged way (as he had announced in
\cref{corr:ab:1928:9:1}, see also
p.~\pageref{page-beh-announce-to-ackermann}).  Thus he encloses the transcript
of his talk in Düsseldorf, which, he writes, is much more elaborated than the
talk itself could be. He sends only the part that would be interesting to
Ackermann, omitting the introductory sections.  What turned out incorrect is
with red pencil put into parentheses or struck out. Behmann writes that there
is no rush to return the transcript, if he would need it in foreseeable
time, he would give notice.

\lskip
So far, the part of the transcript sent to Ackermann could not be located in
Behmann's bequest.  Possibly \mref{man:beh:26:beziehungen:incomplete} is the
first part that retained with Behmann and
\mref{man:beh:26:beziehungen:draft:2} is an early version.

\lskip
Another topic discussed at length in the letter is Behmann's work on the
resolution of the paradoxes.  He encloses a carbon copy of a manuscript about
this, which he had sent to Hilbert a few days before.\footnotemark\
Behmann concludes the letter with announcing that he will sent his comments on
the details in the book \cite{hilbert:ackermann:28} (see
\cref{corr:ab:1928:9:1}) as soon as he has looked through it for these.

\footnotetext{There are just two letters from Behmann to Hilbert in Behmann's
  bequest \cite{beh:nl}, carbon copies of typescripts dated 18~September 1928
  and 25~September 1928 \cite[Kasten I~32]{beh:nl}, both of them with
  technical discussions of paradoxes, the latter with corrections to the
  first.  Behmann also sent these on 29~September 1928 to Frank P. Ramsey,
  who gave a detailed reply on 4~October~1928, leading to a further letter by
  Behmann dated 9~October and by Ramsey dated 16~October \cite[I~56]{beh:nl}.
  Behmann refers briefly to that correspondence in
  \cite[p.~220]{beh:37:paradoxes}.}

\item \label{corr:ab:1928:11:1}
 Ackermann to Behmann, Münster i. W., 1 November 1928, handwritten, 4~pages.

\lskip
See Sect.~\ref{sec-corr-elim-1928b} for the content regarding elimination.
Ackermann further considers the question of what formulas can be proven and
what can not be proven from the axioms on p.~53 in
\cite{hilbert:ackermann:28}.  He sketches an idea for a method to verify
formulas that are valid for all domains which operates by successively
strengthening formulas by melting quantifiers until their matrix is
propositionally valid. He gives two examples, presented here in tabular form
and modern notation:
 
\begin{equation}
\begin{array}{r@{\hspace{0.5em}}ll}
(1.) & & \all x \ex y\, (fxx \lor \lnot fxy)\\
(2.) & \entailedby & \all x\, (fxx \lor \lnot fxx)\\
(3.) & \equiv & \true.
\end{array}
\end{equation}

\smallskip

\begin{equation}
\label{eq-ack-instance-method}
\begin{array}{r@{\hspace{0.5em}}ll}
(1.) & & \ex x \all y\, (fxx \lor \lnot fyy \lor fxy)\\
(2.) & \equiv & 
       \ex x \all y\, (fxx \lor \lnot fyy \lor fxy) \lor
    \ex u \all v\, (fuu \lor \lnot fvv \lor fuv)\\
(3.) & \equiv &  \ex x \all y \ex u \all v\, 
(fxx \lor \lnot fyy \lor fxy \lor
fuu \lor \lnot fvv \lor fuv)\\
(4.) & \entailedby & \ex x \all y \all v\, 
(fxx \lor \lnot fyy \lor fxy \lor
   fyy \lor \lnot fvv \lor fyv)\\
(5.) & \equiv & \true.
\end{array}
\end{equation}
Step~(2.) of (\ref{eq-ack-instance-method}) is formed by disjoining (called
\deq{multiplizieren} by Ackermann) the given expression with itself. Step~(4.)
is obtained by melting $\all y$ and $\ex u$, that is, the $\ex u$
is ``sucked up'' by the preceding universal quantifier.  The difficulty is,
according to Ackermann, that an expression has to be first disjoined with
itself a finite number of times before the finite number of possible variable
meltings can be investigated, and the number of disjoined copies has to be set
in advance.  Ackermann notes that the method is certainly insufficient for
formulas with quantifier prefixes \[\all x_1 \ldots \all x_n \ex y_1
\ldots \ex y_m\] and
\[\all x_1 \ldots \all x_n \ex y \ex z_1 \ldots \ex z_m,\] but
that so-far he could not provide a general proof.

\lskip
The letter further discusses Behmann's idea about the resolution of the
paradoxes (see \cref{corr:ba:1928:9:29}). Ackermann writes that some issues
did not become fully clear to him and gives an example to clarify.

\item \label{corr:ba:1934:10:22} Behmann to Ackermann, Halle (Saale), 22
  October 1934, typescript with handwritten insertions at formulas, 3~pages.
  Behmann sends at the same time some offprints.  Reference to an
  unidentified previous mail in which Ackermann had sent an offprint of the
  paper \cite{ackermann:35} to Behmann.

  \lskip
  See Sect.~\ref{sec-corr-elim-1934} for the content regarding elimination.
  Behmann gives some remarks on notation and presentation: He suggests to
  replace the confusing terms \deq{Summe} and \deq{Produkt} by
  \deq{Konjunktion} and \deq{Disjunktion} and recommends to call $f, g, h, k$
  constants instead of variables, referring to \cite[footnote~25,
    p.~196]{beh:22} (see also Sect.~\ref{sec-vars-constants}).  In a
  postscript he proposes to proceed with simplifying Hilbert's symbolism by
  removing \emph{everything} that is dispensable until the differences to
  Behmann's notation are no longer worth mentioning.  In particular, the large
  extent of the formulas and the clutter of the many parentheses (\deq{das
    Gestrüpp der vielen Klammern}) are considered by Behmann as quite
  perturbing at practical work.

  \lskip
  Behmann writes that references to sources, wherever appropriate, would be
  desirable, giving \cite[512--516]{schroeder:3} for p.~412 (which actually
  concerns Skolemization) as an example. In particular, he does not see to
  what extent Ackermann has made explicit use of Behmann's earlier
  communications (in relation to his talk in Düsseldorf), or came to the same
  findings already independently.

  \lskip
  Behmann concludes with asking Ackermann to send an offprint of
  \cref{corr:ba:1934:10:22} to Boskovitz.\footnotemark

  \footnotetext{Alfred Boskovitz, see footnote~\ref{footnote-boskovitz} on
   p.~\pageref{footnote-boskovitz}.}

  \lskip
  This letter is the only one in the Behmann-Ackermann correspondence in
  \cite{beh:nl} that ends with a Nazi salutation, \de{mit deutschem Gruß},
  common in Germany at that time for official letters \cite{ehlers:nazisalut}.

\item \label{corr:ab:1934:10:29} Ackermann to Behmann, Burgsteinfurt, 29
  October 1934, typescript with handwritten formulas, 5~pages.

  \lskip  
  See Sect.~\ref{sec-corr-elim-1934} and also \ref{sec-corr-elim-1928b} for
  the content regarding elimination.  Ackermann replies to Behmann's points of
  criticism, kindly remarking that, as can already be seen at the formulas in
  the letter, he has adopted Behmann's notation for the universal and
  existential quantifiers for his own practical work, and referring to
  Behmann's small book \cref{corr:ab:1934:10:29} as a rich source in this
  respect.\footnotemark\ However, he considered himself as bound to the
  notation he and Hilbert used in their introduction to logic
  \cite{hilbert:ackermann:28}.

\footnotetext{\de{Was die Symbolik anbetrifft, so sehen Sie schon an den
    hingeschriebenen Formeln, dass ich mir für meine praktischen Arbeiten auch
    Ihre Schreibweise für das All- und das Seinszeichen angeeignet habe, wie
    überhaupt Ihr Büchlein über ``Mathematik und Logik'' eine wahre Fundgrube in
    dieser Beziehung ist.}}

  \lskip
  Schröder's idea underlying the \deq{Belegungsfunktionen} (known today as
  Skolem functions) has, according to Ackermann, become such common knowledge
  of logicians that in his view there would be no need to cite Schröder.  As
  an example he refers to Skolem's use and simplified notation in 1920. He
  writes that he had investigated the decision problem and elimination problem
  in the case where Skolem functions do occur already in 1925: he found an
  excerpt of \cite{schroeder:3} from that time where he has proven for nine of
  ten resultants that have in part just been conjectured by Schröder and
  brought somehow into connection with Peirce (Ackermann is only in possession
  of this excerpt of Schröder's book) the correctness of Schröder's
  conjectures.  The method is the same as Example~(26) in \cite{ackermann:35},
  which actually belongs to the mentioned ones by Schröder. Thus, Ackermann
  believes to be independent of Behmann's communications in his remarks in
  Section~6 of \cite{ackermann:35} (where existential quantifiers are
  discussed), the only section to which Behmann's comment in
  \cref{corr:ba:1934:10:22} might refer. Ackermann notes that he obviously
  could have mentioned Behmann's talk in Düsseldorf and their correspondence
  that followed. He writes that the large temporal distance might be an excuse
  and announces to make up for it at the next occasion.

  \lskip
  Skolem functions are no longer considered as an advantage by Ackermann (see
  Sect.~\ref{sec-ackermann-switching}).  His alternative approach is
  elaborated later in \cite{ackermann:35:arity} (also see
  Sect.~\ref{sec-ackermann-switching}), where, in a footnote he explains --
  related to the suggestions by Behmann in \cref{corr:ba:1934:10:22} -- the
  history of Skolemization and mentions that Behmann has brought the
  advantages of the introduction of Skolem functions for the elimination
  problem to attention already in his talk at the
  \dename{Mathematikerversammlung} in Düsseldorf 1926.

  \lskip
  Concerning Behmann's ultrafinite logic (seemingly the topic on an offprint
  or manuscript sent by Behmann with \cref{corr:ba:1934:10:22}), Ackermann
  urges him to elaborate a precise formulation and brings the evidently
  similar theory by Church \cite{church:postulates:1,church:postulates:2} to
  attention.
  
  \lskip
  Ackermann confirms that he had sent an offprint of \cite{ackermann:35} to
  Boskovitz (see \cref{corr:ba:1934:10:22}).  In a postcard to Behmann dated
  17 June 1935 \cite[Kasten~1, I~08]{beh:nl}, Boskovitz writes that he
  received the offprint in fall and thanks Behmann, conjecturing that the
  address and probably the idea are due to him.  However, technical aspects of
  \cite{ackermann:35} are not discussed by Boskovitz.

  \lskip
   The letter concludes with Ackermann remarking that he has heard from Arnold
   Schmidt\footnotemark\ that Behmann had been on the congress in Pyrmont
   (\dename{Jahrestagung der Deutschen Mathematiker-Vereinigung 1934}) and
   thus regrets to have dropped his original intention to also go there.

  \footnotetext{Hermann Arnold Schmidt (1902--1967), German mathematician.}

\item \label{corr:ab:1953:1:9} Ackermann to Behmann, Lüdenscheid, 9 January
  1953, typescript, 1~page.  Reference to an unidentified previous postcard
  and mail by Behmann, where Behmann sent a copy of his treatise
  \deq{Deskription und limitierte Variable} \cite{beh:52:deskription} and
  requested the address of Sören Halldén\footnotemark.  
  
  \lskip
  Ackermann thanks Behmann for his treatise, and writes that it had interested
  him very much -- as well as a continuation (\de{Weiterführung}) would do.
  He cites it later in \cite{ackermann:58:typenfrei} with just \name{1944 and
    extended 1952} as bibliographic details. He does not know the full address
  of Sören Halldén but says that mails with just \dename{Universität Uppsala}
  as address arrived. As American experts that would be interested in
  \cite{beh:52:deskription}, he mentions Church and Quine.

  \footnotetext{Sören Halldén (1923--2010), Swedish philosopher and logician.}

\item \label{corr:ab:1954:10:22} Ackermann to Behmann, Lüdenscheid, 22 October
  1954, handwritten, 2~pa\-ges.  Reference to a previous mail by Behmann with
  copies of the treatises \dename{Zur Technik des Schließens und Beweisens}
  \cite{beh:54:schliessen} and \dename{Ein neuer Vorschlag für eine
    einheitliche logische Symbolik} \cite{beh:54:symbolik}.

  \lskip Ackermann discusses the standardization of logic symbolism,
  referencing also to Behmann's earlier works on the subject, one from 1935
  and \cite{beh:37:symbolik}. He suggests that Behmann should submit
  \cite{beh:54:symbolik} to the Journal of Symbolic Logic.

\item Behmann to Ackermann, Bremen-Aumund, 6 October 1955, typescript,
  2~pages. Carbon copies of a draft manuscript \name{Die Besetzungskette und
    der widerspruchsfreie Prädikatenkalkül} are originally enclosed.

 \lskip Planning to submit a revision of his manuscript to the Journal of
 Symbolic Logic (which was published there indeed as \cite{beh:59:limitierte}),
 Behmann asks Ackermann for his judgment regarding content as well as form,
 and to forward copies to Hermes, Scholz and probably
 Hasenjaeger.\footnotemark\ Behmann intends to visit Münster for a week and
 give a talk.  He stresses the importance of the problem of a
 contradiction-free predicate calculus.  He would have also sent copies of his
 talk \cite{beh:55:paradox}, which he had submitted but not given, but the
 proceedings did not yet appear. Finally, he asks Ackermann again about his
 judgment on the previously sent treatise \cite{beh:54:schliessen}.

\footnotetext{In 1955 Heinrich Scholz (1884--1956) was professor emeritus,
  Hans Hermes (1912--2003) and Gisbert Hasenjaeger (1919--2006) were
  professors, and Ackermann was honorary professor at \dename{Westfälische
    Wilhelms-Universität Münster} with its renowned \dename{Institut für
    mathematische Logik und Grundlagenforschung} founded in 1950 by Scholz and
  lead after his retirement 1953--66 by Hermes.}

\item \label{corr:ab:1955:10:27}
Ackermann to Behmann, Lüdenscheid, 27 October 1955, typescript, 3~pages.

\lskip
Ackermann discusses Behmann's manuscript \name{Die Besetzungskette und der
  widerspruchsfreie Prädikatenkalkül} sent with \cref{corr:ab:1954:10:22}.  He
has not yet talked with the \de{Herren} in Münster about it, since he will be
there again only in about 10 days, when the lecture period begins.  Behmann's
ideas in \cite{beh:31:widersprueche} have always interested him very much, but
they did not have received greater attention because so far they have just
been program and it did not came to a precise calculus based on them. Although
in the recent manuscript the ideas are explicated further in some respects,
the situation has not much changed.  Ackermann continues to describe the
conditions that must be satisfied to allow an exact discussion about the
ideas: (1) An overview on all used symbols -- which would be easy to satisfy;
(2) An exact specification of the combinations of symbols that should be
considered as meaningful; (3) The formal rules and base forms that should be
used.  He refers to his publication~\cite{ackermann:41:typenfrei}, where he
pointed out the difficulties that arise if one wants to pass from Behmann's
ideas to a precise calculus and mentions that Behmann has never commented on
them.  It does not suffice to show that the paradoxes in the known way are
avoided, but has to be ensured that they do not appear in modified
form. Ackermann brings the related work by Church
\cite{church:postulates:1,church:postulates:2} again (see
\cref{corr:ab:1934:10:29} ) to attention, which later showed to be
contradictory by the paradox discovered by Kleene and Rosser. Ackermann
proceeds with an example intended to be helpful for pointing out the
difficulties that can arise if there is no precise version of all inference
rules available.  He concludes with noting that in the referenced work by
Church there is already a restriction operator (\de{Limitator}), at least in
connection with the universal quantifier, asking for clarification of the
exact correspondence of Behmann's notation for this case with Church's and
Ackermann's own in the mentioned work, which follows Church's.

\item \label{corr:ba:1955:12:6} Behmann to Ackermann, Bremen-Aumund, 6
  December 1955, typescript, 6~pages.  Carbon copies of two revised manuscript
  sections were originally enclosed.

\lskip
Thanking Ackermann for his letter \cref{corr:ab:1955:10:27}, Behmann writes
that he has incorporated some of the addressed points essentially in form of
explanatory footnotes. He then discusses the three particular points (1)--(3).
About (2) he mentions that he succeeded to bring substitution algorithm
(\de{Einsetzungsalorithmus}) and at the same time and in immediate
relationship the criterion for meaningfulness (\de{Sinnhaftigkeitskriterium})
into a much simpler and now completely determined form (\de{zwangsläufige
  Gestalt}). Accordingly, he changed the subject of his work to \de{Der
  Algorithmus der Einsetzung und der widerspruchsfreie Prädikatenkalkül}.
Concerning~(3) he is reluctant to fix a thorough axiomatization at that point,
because he does not consider axiomatic representation as helpful for getting a
first understanding of a subject. For him, a strong axiomatic pre-load
(\de{Vorbelastung}) was always a obstacle to get deeper into the meaning and
essence of a subject.  He had studied \cite{ackermann:41:typenfrei} at its
time, but the difficult outer and personal circumstances as well as the wish
to not just criticize but also give a positive image had prevented that he
commented it.  Ackermann's use of restriction differs from his, in particular
because the former is through the addition of provability loaded with
modality.  Behmann conjectures a circularity, if the concept of provability is
already presupposed at the introduction of the junction. He further expresses
doubts that simple applications such as syllogisms based on empirical facts
can be expressed with Ackermann's provability-based notion.  Behmann explains
the differences between Ackermann's and his standpoint: His own proceeding is
crucially determined by starting from plain propositional calculus, with
implication in the Stoic, or Fregean, resp., sense, without modal or proof
technical pre-load. Only after its fixation quantification is added, then
modality in form of modal operators, and finally meta logic with inclusion of
proof theory.  For a presentation of the principles of this building-up, he
refers to \cite{beh:55:paradox}, still waiting for the proceedings to appear.
He emphasizes that his notion of meaningfulness (\de{Sinnhaftigkeit}) is
unrelated to decidability as provability or refutability (\de{Entscheidbarkeit
  (als Beweisbarkeit oder Widerlegbarkeit)}).  In \cite[p.~80]{halden} he
found an expression that corresponds to the translation of the restriction
operator suggested by Ackermann for his and Church's notation.  Further
discussions concern the inclusion of the third truth value \name{meaningless}
(\de{sinnlos}) into the calculus. Behmann mentions that he intends to include
in his paper also a form of Russell's paradox that was earlier communicated by
Ackermann. He can now present its resolution much shorter and clearer than in
\cite{beh:53:ausschaltung}.  Enclosing carbon copies of new Sections 6 and 7
of his paper, he asks Ackermann for judgment and transferal to his associates
in Münster. He intends to send copies of the two sections as well as a copy of
the present letter directly to Scholz.

\end{enumerate}

\sectionmark{Discussions in Behmann's Further Correspondence}

\section{Discussions Related to Second-Order Quantifier Elimination
in Behmann's Further Correspondence}
\label{sec-corr-further}

\sectionmark{Discussions in Behmann's Further Correspondence}

In this section, further discussions in Behmann's correspondence -- beyond
that with Ackermann -- with immediate relationship to his work on the decision
and elimination problem are summarized. They are presented chronologically,
headed by the respective correspondence partners, Bertrand Russell, Rudolf
Carnap, Heinrich Scholz and Alonzo Church.

\subsection{Bertrand Russell}

Behmann writes on 8 August 1922 in English to Bertrand Russell \cite[Kasten~2,
  I~60]{beh:nl}, sending him an offprint of \cite{beh:22}.  He sketches the
role of Russell's work with respect to his dissertation, whose theme was
proposed by Hilbert. Leaving this subject, Behmann continues: \enq{A I already
  remarked, my article in the Mathematische Annalen follows another way. Not
  withstanding my statement of its purpose and character in $\S~1$ I beg leave
  for a few additional words. It was what I call the Problem of Decision,
  formulated in the said paragraph, that induced me to study the logical work
  of Schröder.  And I soon recognized that in order to solve my particular
  problem, it was necessary first to settle the main Problem of Schröder's
  Calculus of Regions, his so-called Problem of Elimination. And -- here I
  quote a sentence from a lecture of mine held before the Göttinger
  Mathematische Gesellschaft -- ``I believe it to be a very lucky circumstance
  that now an opportunity presents itself to embrace the earlier investigation
  relevant to that topic, questionable, especially as regards the form of
  presentation, and of difficult access as they are, under a new, uniform, and
  valuable point of view, thus saving from oblivion a great deal of profound
  and hard work of thought.'' -- Indeed, the chief merit of the said problem
  is, I daresay, due to the fact that it is a problem of fundamental
  importance on its own account, and, unlike the application of earlier
  Algebra of Logic, not at all imagined for the purpose of symbolic treatment,
  whereas, on the other hand, the only means of any account from its solution
  are exactly those of Symbolic Logic. -- But I wish to avoid anticipation.}
(Parts of this letter are also quoted in \cite{mancosu:behmann:99}).

On 16 September 1922 Bertrand Russell \cite[Kasten~2, I~60]{beh:nl} replies
that, so far, he not have had time to read Behmann's paper carefully, but
notes \enq{I see that you have a symbolism very admirably adapted to your
  subject}. He brings Sheffer's article \cite{sheffer:13} to attention and
explains how \enq{the logic of propositions can be developed with only
  \emph{incompatibility} (not-p or not-q), instead of both negation +
  disjunction.} (that is, Sheffer's stroke \name{NAND}) or equally well
\enq{(NOT-p and NOT-q)} (\name{NOR}).

\subsection{Rudolf Carnap}

In his postcard dated 20 November 1922 \cite[Kasten~1, I~10]{beh:nl} Carnap
thanks Behmann for sending him an offprint of \cite{beh:22}. He writes that it
did interest him very much and asks Behmann for sending a copy to
Gerhards\footnote{Karl Gerhards (1888--1957), German philosopher,
  mathematician and physicist.}.

In a letter dated 19 February 1924 \cite[Kasten~1, I~10]{beh:nl}, Carnap
relates Behmann's work on the decision problem for relations to Behmann's
refusal to contribute to the development of a symbolism that is based on
Russell's. Carnap writes (in translation): \enq{If you really succeed to solve
  the decision problem so far that it also includes the of theory of relations
  (\de{Beziehungslehre}), and in particular also that which uses constant
  non-logical relations, this would be a very lucky and valuable progress and
  extraordinarily useful for my works.  However, I believe that the
  development until practical applicability will still need a lot of
  time. Thus, \emph{for the time being} I still want to stick to
  Russell's symbolism.} Carnap can understand that Behmann, being now busy
with the extension of his own symbolism, is not inclined to contribute to the
appearance of the other symbolism.

\subsection{Heinrich Scholz}
\label{sec-letter-scholz}

As discussed in Sect.~\ref{sec-contrib-decision} and
\ref{sec-contrib-problem-of-soqe}, Behmann explains in his letter from 27
December 1927 \cite[Kasten 3, I~63]{beh:nl} to Heinrich Scholz his
contribution to the decision and elimination problem.  Here is an excerpt of
the respective parts of the original German letter:

\label{page-scholz-lengthy}
\de{Bezüglich des Entscheidungsproblems ist zu unterscheiden zwischen der
  Entscheidung innerhalb der rein aussagenlogischen -- die nur Begriffe der
  linken Seite meiner Tabelle enthalten -- und im Gesamtbereich der
  Begriffslogik. Wo die beiden Lösungen des elementaren Problems, die
  Verfahrend der Einsetzung und der konjunktiven Normalform, zuerst erwähnt
  und systematisch dargestellt worden sind, ist mir nicht bekannt. Vielleicht
  kommt hier Whiteheads ``Universal Algebra'' in Frage.  Nach meiner Meinung
  führt die elementare Aussagenlogik so zwangsläufig auf dieses Problem, dass
  seine Aufweisung und Erledigung wohl mit der ersten strengen Darstellung
  jener überhaupt zeitlich zusammenfallen wird. So weit ich mich erinnere,
  habe ich das Verfahren der Normalform -- das andere steht ja in den PM --
  durch Hilbert, der sich ja schon vor dem Kriege selbständig mit dem Problem
  der symbolischen Darstellung der Logik, und zwar insbesondere der
  Aussagenlogik, beschäftigte, kennen gelernt und weiss gerade aus diesem
  Grunde hinsichtlich der Literatur dieses Punktes keine sichere Auskunft zu
  geben.}

\de{Was nun das allgemeinere Entscheidungsproblem betrifft -- nicht zu
  verwechseln mit dem z.B. von Hessenberg in einer seiner Schriften in den
  Abh. d. Friesschen Sch. besprochenen Entscheidbarkeitsproblem (im
  Zusammenhang mit dem Paradoxon von Richard) --, so ist dieses
  \emph{ausdrücklich} als solches meinen Wissens vor mir von niemandem
  behandelt worden. Wohl aber ist unabhängig von dieser Fragestellung oder
  doch zum mindesten ohne ihre ausdrückliche Erwähnung eine gewisse
  Teilaufgabe, das sogenannte Eliminationsproblem, das übrigens zugleich einem
  passend eingeschränkten Entscheidungsproblem äquivalent ist, behandelt
  worden, und zwar zunächst von den Amerikanern, insbesondere Peirce, und
  hierauf mit besonderer Liebe und Ausdauer von Schröder, und schliesslich hat
  es einen Spezialisten in Löwenheim gefunden, der in den Math. Annalen
  darüber mehrere Abhandlungen schrieb. Die wichtigste von ihnen ist ``Ueber
  Möglichkeiten im Relativkalkül''; sie erschien, glaube ich, einige Jahre vor
  dem Kriege. L. ist dort schon zu wesentlichen Teilergebnissen meiner
  Abhandlung über das Entscheidungsproblem gelangt, allerdings -- in einer
  weder mathematisch strengen noch hinreichend verständlichen Darstellung, so
  dass ich selber diese erst nach dem Erscheinen meiner eigenen Schrift
  richtig herausgelesen habe. Dazu gebraucht er eine ziemlich sonderbare
  Terminologie; was ich ``Aussage des Bereiches A'' nenne, heisst bei ihm
  ``Zählgleichung'', usw. Wie ich es mir erkläre, dass man zu dieser
  Fragestellung gelangen und dabei an dem Entscheidungsproblem so völlig
  vorbeigehen konnte, habe ich in meiner Schrift über das Entscheidungsproblem
  S. 218 -- 19 auseinandergesetzt.}

\subsection{Alonzo Church}
\label{sec-letters-church}

In his letter dated 15 April 1937 to Alonzo Church \cite[Kasten~1,
  I~11]{beh:nl}, Behmann comments seemingly an offprint sent earlier by
Church, where a system by Quine is discussed.  Behmann states that the idea as
such to combine (\de{zusammenfassen}) several argument to complexes that can
then be handled like uniform (\de{einheitliche}) arguments is rather obvious.
He had applied it himself in \cite{beh:26:beziehungen} (see
Sect.~\ref{sec-corr-elim-1928a}) to reduce the general elimination problem --
as preliminary step (\de{Vorsutufe}) to the decision problem -- for the case
of universal individual quantifiers onto the sequence (\ref{eq-sequence-sol}),
p.~\pageref{eq-sequence-sol}. However, without making the principles of the
reduction explicit. He considers the idea to use this possiblity already for
building-up propositional and predicate logic as new and valuable.  But, he
continues, it needs to be done in way such that natural associativity is
preserved; Quine's scruples on this should be countered by a reasonable
formulation of the logical rules. (Behmann also remarks that the statement in
\cite{beh:26:beziehungen} that elimination can be performed in general for that
case does not hold.)

Church answers on 20 May 1937 \cite[Kasten~1, I~11]{beh:nl}: \enq{I believe
  that I am in substantial agreement with what you say about Quine. There is
  only one remark which it occurs to me to make. He does not introduce a law
  of associativity such as you suggest. With him the ordered pair $(x,y)$ is
  introduced as an undefined concept; $(x,(y,z))$ is not the same as
  $((x,y),z)$; the ordered triad $(x,y,z)$ is defined to be $((x,y),z)$, the
  ordered tetrad $(x,y,z,t)$ to be $(((x,y),z),t)$, and so on. This seems to
  me satisfactory; from certain points of view one could even go as far as to
  say that an associative law would be undesirable.}

In a much later letter from Behmann to Church, dated 30 January 1959, where
the main topics are the relationship of \cite{beh:59:limitierte} to lambda
conversion, as well as, possibilities to simplify propositional formulas,
Behmann sketches his method in \cite{beh:61:vereinfachungsproblem}: As first
step there, by means of simple transformation rules -- of the kind of
Gentzen's \de{Mischungsregel}, which, Behmann writes, he as already stated in
1922 (i.e. in \cite{beh:22}) -- the totality of prime implicants is obtained
in a deterministic way (\de{zwangsläufig gewonnen}). Also the name
\name{innex} form is suggested in this letter for the result form of his
method in \cite{beh:22}, where quantifiers are propagated inwards (see
p.~\pageref{foot-innex}, footnote~\ref{foot-innex}).

\section{Documents that could not be Located}

The following documents would be relevant for the consideration of
second-order elimination and are referenced in the Behmann's bequest, but, so
far could not be located:

\begin{enumerate}
\item The technical part of Behmann's transcript of his talk
  \de{Entscheidungsproblem und Logik der Beziehungen} on 23 September 1926 in
  Düsseldorf, which he sent in 1928 to Ackermann.  See notes on
  \mref{man:beh:26:beziehungen:incomplete} and \cref{corr:ba:1928:9:29}.

\item A contribution by Behmann on the solution and elimination problem
  (\de{zum Problemkreis Auflösungs- und Eliminationsproblem}) in a memorial
  publication (\de{Festschrift}) for Ernst Schröder, mentioned in 1942 in the
  correspondence with Scholz.
\end{enumerate}

\noindent
Two issues where, so far, no further relevant documents could be found in the
bequest are normalization with respect to given predicates, suggested in
\cite[p.~201]{beh:22} (see Sect.~\ref{sec-nf-predicate}) and the actual
influence of Löwenheim's work \cite{loewenheim:15} and possibly Skolem's
papers \cite{skolem:19,skolem:20} (with exception of Behmann's 1927 letter to
Scholz see Sect.~\ref{sec-contrib-decision}, \ref{sec-contrib-problem-of-soqe}
and \ref{sec-letter-scholz}).

\part{Conclusion}
\label{part-conclusion}

\section{Concluding Remarks}

The early results on the decision problem by Löwenheim, Skolem and Behmann
1915--22 included the identification of relational monadic formulas as
decidable fragment of first-order logic, the development of a decision
procedure for this fragment, as well as the first explicit statement of the
decision problem itself by Behmann 1921.  
As became evident from our detailed inspection of Behmann's Habilitation
thesis from 1922, where he made precise and extended earlier work by Schröder,
these early investigations of the decision problem were closely tied to the
problem of second-order quantifier elimination. The basic connection there is
that a relational monadic formula can be decided by eliminating all its
predicates one by one.

The decision problem has been much researched since then, with several
monographs
\cite{dc:54:ackermann,dc:59:suranyi,dreben:goldfarb,classical:decision},
where in particular the last cited one gives a comprehensive survey, also on
the history and literature. Despite many ``theoretical'' results, the
development of calculi that decide specific classes is still an active issue
of research in automated deduction.

The elimination problem was brought to larger attention in computational logic
in the 1990s with the development of two algorithms, the resolution-based SCAN
\cite{scan} and DLS \cite{dls:early}. The latter initiated the so-called
\name{Ackermann approach}, which is, like Behmann's method, based on
equivalence preserving formula rewriting. While DLS explicitly involves an
essential idea from Ackermann's 1935 paper \cite{ackermann:35}, SCAN is in
part a re-discovery of another technique from that paper
\cite{nonnengart:elim:1999}.  Early considered applications of second-order
quantifier elimination in computational logic were the computation of
first-order correspondence properties of modal formulas
\cite{scan,sqema,schmidt:2012:ackermann} and non-monotonic reasoning by
computing circumscription \cite{dls}, which can be extended to model various
semantics for logic programming \cite{cw-logprog-short} and to the computation
of abductive explanations \cite{kakas-kowalski-toni} with respect to these
semantics \cite{cw-abduction}.  Since the mid 2000s, variants of second-order
quantifier elimination like uniform interpolation, forgetting and projection
receive great interest as operations for the processing of description logic
knowledge bases, e.g. \cite{dl-conservative,lutz:ijcai11,ludwig:dl14}.  A
number of specialized elimination methods for particular description logics
have been developed, where advanced ones explicitly relate to SCAN and the
Ackermann approach \cite{ks:2013:frocos}.  Another current activity related to
second-order quantifier elimination is in SAT solving the investigation of
formula simplifications that perform elimination of Boolean variables in
restricted ways \cite{biere:resolve2004}.

In contrast to the decision problem, there is only a single monograph
\cite{soqe} on the elimination problem, with focus on modern
developments. Some early and fundamental results such as the success of
elimination on relational monadic formulas and Behmann's rewriting-based
method as a blueprint of the more advanced DLS apparently came to surface
again only recently \cite{cw-relmon}. As shown there, one can even today learn
from Behmann's method possible improvements of modern methods such as DLS, the
success on relational monadic formulas emerges as a useful completeness
property of elimination methods, and on a closer look the seemingly not very
expressive class of relational monadic formulas (\name{Klassenlogik}, as it
was called at Behmann's time), when considered together with second-order
quantification, shows interesting relationships to description logics, the
logics regarded today as adequate to represent concept (or \name{class})
ontologies.

Behmann's primary concern was the decision problem.  In his correspondence
with Ackermann it can be observed that he talks about the decision problem,
even if the discussed method actually performs elimination, whereas Ackermann
often speaks about the elimination problem.  The decision problem was first
stated explicitly by Behmann, the elimination problem, in contrast, was also
investigated earlier by others, in particular Schröder and Löwenheim.
Nevertheless, the methods developed by Behmann decide formulas by performing
elimination.  He was dissatisfied with later decidability results by Bernays,
Ackermann, Schönfinkel and Schütte that reside on satisfiability for domains
with finite cardinalities determined from syntactic structure, since these
results only lead to primitive methods, not suited for practical application.
Thus, Behmann's main concern was actually not just the decision problem: it
was the problem of finding \emph{practically applicable decision methods}.

This places his work into the context of computer science.  He used a
relatively small set of essentially syntactic tools, which let his methods
smoothly fit into modern computational logic: Rewriting formulas such that
equivalence is preserved or to entailed formulas.  Moving arguments inward and
outward of terms. Adding auxiliary definitions.  Distribution among
connectives.  Propagating quantifiers inward and outward.  Various ways of
normalization, including generalizations of disjunctive and conjunctive normal
form and Boolean combinations of specific basic forms.  Skolemization and
un-Skolemization.  And, as outlined in a manuscript and the correspondence
with Ackermann, representing infinite sets of formulas by schematic formulas
with superimposed graph structure. In view of his later works, quantifier
restriction and variants of lambda conversion have to be added.

Some of the material in this report should be useful for further
investigations in the technical-historical realm, such as the clarification of
the exact relationship between the works of Löwenheim \cite{loewenheim:15},
Skolem \cite{skolem:19,skolem:20} and Behmann \cite{beh:22}, or an in-depth
investigation of Ackermann's resolution-based method for second-order
quantifier elimination in \cite{ackermann:35} in presence of the related
unpublished works by Behmann summarized here, related work by Craig
\cite{craig:bases}, modern re\-solu\-tion-based elimination methods
\cite{scan,scan:complete} and fixpoint techniques as used to specify logic
programming semantics \cite{emden:kowalski} and in elimination
\cite{nonnengart:fixpoint}.

\bigskip
\noindent
\textbf{Acknowledgements. }
This work was supported by DFG grant~WE~5641/1-1.

\newpage

\clearpage
\addtocontents{toc}{\protect{\vspace*{5.4ex}}}
\phantomsection
\addcontentsline{toc}{section}{References}

\begingroup
\setlength{\emergencystretch}{2em}
\printbibliography[heading=mybibheading]
\endgroup

\end{document}